\documentclass[lettersize,journal]{IEEEtran}
\usepackage{amsmath,amssymb,amsfonts}
\usepackage{amsthm}
\usepackage[ruled,linesnumbered]{algorithm2e}
\usepackage{array}
\usepackage[caption=false,font=normalsize,labelfont=sf,textfont=sf]{subfig}
\usepackage{textcomp}
\usepackage{stfloats}
\usepackage{url}
\usepackage{tcolorbox}
\usepackage{verbatim}
\usepackage{graphicx}
\newtheorem{theorem}{Theorem}
\newtheorem{lemma}{Lemma}

\newtheorem{definition}{Definition}
\def\BibTeX{{\rm B\kern-.05em{\sc i\kern-.025em b}\kern-.08em
    T\kern-.1667em\lower.7ex\hbox{E}\kern-.125emX}}
\usepackage{balance}
\usepackage{hyperref}
\hypersetup{
    colorlinks=true,
    linkcolor=blue,
    filecolor=magenta,      
    urlcolor=cyan,
    }
\begin{document}
\title{Influential Billboard Slot Selection using Spatial Clustering and Pruned Submodularity Graph}
\author{Dildar Ali, Suman Banerjee, and Yamuna Prasad \\
%Department of Computer Science and Engineering \\
%Indian Institute of Technology Jammu, Jammu \& Kashmir 181221, India 
\thanks{Manuscript created October, 2020; This work was developed by the IEEE Publication Technology Department. This work is distributed under the \LaTeX \ Project Public License (LPPL) ( http://www.latex-project.org/ ) version 1.3. A copy of the LPPL, version 1.3, is included in the base \LaTeX \ documentation of all distributions of \LaTeX \ released 2003/12/01 or later. The opinions expressed here are entirely that of the author. No warranty is expressed or implied. User assumes all risk.}}

%\markboth{Journal of \LaTeX\ Class Files,~Vol.~18, No.~9, September~2020}%
%{How to Use the IEEEtran \LaTeX \ Templates}

\maketitle

\begin{abstract}
Billboard advertising is a popular out-of-home advertising technique adopted by commercial houses. Companies own billboards and offer them to commercial houses on a payment basis. Given a database of billboards with slot information, we want to determine which $k$ slots to choose to maximize influence. We call this the \textsc{Influential Billboard Slot Selection (IBSS)} Problem and pose it as a combinatorial optimization problem. We show that the influence function considered
in this paper is non-negative, monotone, and submodular. The incremental greedy approach based on the marginal gain computation leads to a constant factor approximation guarantee. However, this method scales very poorly when the size of the problem instance is very large. To address this, we propose a spatial partitioning and pruned submodularity graph-based approach that is divided into the following three steps: pre-processing, pruning, and selection. We analyze the proposed solution approaches to understand their time, space requirement, and performance guarantee. We conduct extensive set of experiments with real-world datasets and compare the performance of the proposed solution approaches with the available baseline methods. We observe that the proposed approaches lead to more influence than all the baseline methods within reasonable computational time.
\end{abstract}


\section{Introduction}
\label{sec:introduction}
\IEEEPARstart{O}{ne} of the key objectives of a commercial house is to
create and increase influence among clients. They invest between $7$ to $10\%$ of their annual revenue in the advertisement process. How this budget can be used effectively is an important research problem in the domain of computational advertisement. Advertisements can be done in several ways social media, television, etc. Recently, the out-of-home advertisement technique has emerged as an effective technique where advertisement contents (e.g., video, animation, etc.) are displayed through digital billboards. As this method gives a significantly higher return on investment this has been adopted largely. According to a recent market survey, billboard advertising has a $65\%$ higher return on investment compared to other advertising techniques. Additionally, Zhang et al. \cite{10.1145/3292500.3330829} conducted a recent study that demonstrated that over $50\%$ of travelers are impressed by at least $5$ billboards during each trip. Previous research in consumer behavior literature has indicated that when a traveler is exposed to an advertisement multiple times, it becomes highly unlikely that they will take any action as a result \cite{10.2307/1153228,f9b31366586b47d389955f209b69da27,SIERZCHULA2014183}.
 
 \par In this advertisement technique, it is assumed that the billboards are with a company and an E-Commerce house wants to choose $k$ of them to maximize the influence. Here, the hope is that if the E-Commerce house can select billboard slots such that the advertisement content is visible to a wider range of people, then this may create a significant influence among the people. This can boost the product sales and financial gains. Due to budgetary constraints, the E-Commerce business can only afford a minimal number of billboard slots. Therefore, a crucial concern is which billboard slot should be picked for publishing the advertisement material for a particular value of $k \in \mathbb{Z}^{+}$. Numerous researchers have recently addressed this issue, and many approaches to its solution have been put forth \cite{wang2022data}. The development of wireless technology and mobile internet has made it simpler today, to track down the location of moving things. As a result, several trajectory datasets are made available in various repositories. These trajectories are properly utilized to solve a variety of real-world issues, including route suggestion \cite{dai2015personalized,qu2019profitable}, predicting driving behavior of vehicles  \cite{xue2019rapid}, and many more. In recent studies, these trajectory datasets have been used for locating billboards in the most influential zones, as was previously described \cite{zhang2020towards}.
 
\par Consider that a trajectory database $\mathbb{T}$ for any city is with us. This database contains the locations of different users along with the time stamps. Now, locations that are of our interest are crowded places like `Shopping Malls', `Street Junctions', `Metro Stations', etc. Digital billboards are placed in these locations and can be hired by E-Commerce houses slot-wise on a payment basis. The information about different billboards (e.g., location, slot duration, cost of each slot, etc.) are available as a billboard database $\mathbb{D}$. It is assumed that the influence is calculated by the \emph{triggering model}, which is well-studied in the influence maximization literature \cite{li2018influence}. Naturally, the key question arises in this context is that given the trajectory and billboard database, and a positive integer $k$, which $k$ slots are to be chosen to maximize the influence?. This problem has been referred to as the \textsc{Influential Billboard Slot Selection (IBSS)} Problem. Though this question has been addressed in a few existing studies, the scalability and efficiency remain a bottleneck. To address this issue, in this paper, we have proposed three solution approaches. The first one is the pruned submodularity graph-based approach. The second approach is spatial clustering along with the incremental greedy approach, and the third approach is based on spatial clustering and pruned submodularity graph. In summary, our contributions in this paper are as follows:

\begin{itemize}
\item We study the \textsc{IBSS} problem which is \textsf{NP-hard} and hard to approximate within a constant factor and is of practical importance.
\item We propose three approaches to solve the IBSS Problem, namely pruned submodularity graph-based approach, spatial clustering+incremental greedy approach, and spatial clustering+pruned submodularity graph approach.
\item All the proposed solution approaches have been analyzed to understand their time and space complexities as well as performance guarantees.

\item All the methodologies have been implemented with real-world trajectory datasets and a number of experiments have been conducted to demonstrate the performance and the efficiency  of the proposed solution approaches.
\end{itemize}

The remaining part of the paper is structured as follows. In Section \ref{Sec:RW}, we discuss previous studies related to our problem. Section \ref{Sec:PPD} provides background information and formally defines the problem statement. In Section \ref{Sec:PA}, we present our proposed solution approach. Section \ref{Sec:EE} details our experimental evaluations of the proposed solution using real-world datasets. Finally, in Section \ref{Sec:CFD}, we present the conclusions of our study and suggest future research directions.

\section{Related Work} \label{Sec:RW}
In this section, we describe the previous studies on influence maximization in billboard advertisements. We divide our literature survey into four major categories: Influential Site Selection, Influential Billboard Selection, Trajectory Driven Influence Maximization, and Regret Minimization.

\par The influential Site Selection problem has gained the researcher’s attention in previous studies due to its importance in broad application areas. Zhou et al. \cite{5767892} studied the MaxBRkNN problem in their article. This problem involves finding the optimal location, the $k$-nearest neighbor of the largest number of trajectories, according to the distance between this location and the trajectories locations. In addition, the MaxBRkNN problem assumes each user remains at a fixed place. But in our case, we fixed the billboard’s location and evaluated the influence on moving trajectories. The location-aware influence maximization problems \cite{10.1145/2588555.2588561,7534856} are motivated by influence maximization problems, which aim to choose a subset of nodes of size $k$ from the social network. Although both influence maximization problems and our proposed approach show the exact purpose of maximizing the influence, there is a difference when the $IM$ problem influence is spread from one user to another in a social network. Still, in our influence model, an audience is only influenced if they come across the range of a billboard. In the last few years, different spatial properties have been considered for developing efficient partition algorithms, such as Wang et al.\cite{10.14778/1687627.1687754} studied the problem to find an optimal region that maximizes the area of BRNN by intersecting geometric shapes. Xia et al. \cite{10.5555/1083592.1083701} studied the problem of finding the most influential places, and for this, they introduced a novel pruning technique to prune less influential sites.
 
\par The trajectory-driven billboard placement problem that Zhang et al.\cite{zhang2020towards} studied is closely related to our \textsc{IBSS} problem. They solved the problem of placing billboards in a way that maximizes their influence on trajectories. This problem gives a set of billboards with their respective locations and costs. The objective is to select a subset of $k$ billboards that most impact the trajectory within the given budget. But, in our work, we find the most influential billboard slots based on budget constraint $k$, and we have used the same triggering model of influence they used to calculate influence in their work. In our \textsc{IBSS} problem, our goal is to find a subset of billboard slots to maximize the influence, and a similar type of subset selection problem was studied by Kuller et al. in their article on the maximum coverage problem \cite{10.1016/S0020-0190(99)00031-9}. In this problem, each element in the subset has a cost associated with it, and the objective is to maximize the coverage within the given budget constraints. Further, they show that normal incremental Greedy does not produce solutions with the required approximation ratio. To overcome this issue, they have introduced a new variation of the greedy approach to achieve the approximation ratio of $(1-\frac{1}{e})$. In our case, we have used the spatial partition approach and subsequently applied pruned submodularity graph followed by incremental greedy, showing that it successfully achieved an approximation ratio of $(1-\frac{1}{e})$. Later, in a similar direction, the Targeted Outdoor Advertising Recommendation problem was studied by Wang et al. \cite{wang2022data}. In this problem, they aim to select a subset of billboard slots for a given budget constraint. Their primary focus was developing a targeted influence model to capture the spread of advertising influence concerning user mobility. This model was the primary contribution of their study. They developed two solution strategies based on the divide and conquer approach and considering user-profiles and advertisement topics. To find out the actual price for outdoor advertising, Zahradka et al.\cite{zahradka2021price} studied the cost of out-of-home advertising in the different regions of the Czech Republic and showed a detailed analysis of how the price based on area varies. 

\par Due to the effectiveness of influence maximization over trajectory in the past few years, Guo et al.\cite{7929916} studied the problem of finding $k$ best trajectories instead of finding $k$ influential places. Their main goal is to find the best trajectories attached to particular advertisements to maximize the influence among target trajectories. Later, Zhang et al.\cite{10.1145/3292500.3330829} studied a problem that found a set of top influential billboards that impress more trajectories under some budget constraints. They introduce a tangent line-based algorithm to estimate an upper bound of billboard influence and reduce the computational cost using a user-defined parameter-based termination method. It achieves an approximation ratio of $\frac{\theta}{2}(1- \frac{1}{e})$. Recently, Wang et al.\cite{8118111} studied a new kind of query, $RkNNT$, which estimates route capacity and planning by considering the user’s source and destination information from the trajectory dataset. Now, in the case of any outdoor or online advertising technique, choosing the right audience for the right advertisement is a very challenging task due to the lack of an audience profile. Now, Wang et al. \cite{8604082} studied that the outdoor advertising industry suffers from delivery of influence from the targeted audience. They introduced a divide and conquer based search approach to resolve this gap. Instead of considering only billboards and trajectories, they considered three factors, i.e., advertisement content, trajectory behavior, and mobility transition.

\par In the field of billboard advertising, there have been limited studies that have taken into account the regret minimization that arises from providing influence to advertisers. Aslay et al. studied the problem of allocating advertisements using viral marketing \cite{10.14778/2752939.2752950}. Here, the advertiser aims to get sufficient virality within the budget constraint. On the other hand host's goal is to provide only the required virality to the advertiser, as giving more virality does not provide any extra incentive to the host. Later, Zhang et al. addressed this issue in their article \cite{zhang2021minimizing}, where they studied the MROAM problem of regret minimization. They proposed a randomized local search algorithm that can minimize the total regret from the perspective of the billboard provider. They also proved that this algorithm achieves an approximation ratio to a dual problem.

\section{Preliminaries and Problem Definition} \label{Sec:PPD}
In this section, we describe the background of the problem and defines it formally. Let us assume, there is a set of $m$ billboard $\mathbb{B}=\{b_1, b_2, \ldots, b_m\}$. Each one of them are operating for the time interval $[t_1,t_2]$, and we assume that $T = t_{2}- t_{1}$. Consider each billboard is allocated slot-wise for displaying advertisement and each slot have a time duration of $\Delta$. Basically, these billboards are located in different places like restaurant, metro places, road junctions etc. in a city. Assuming that a person $u_i$ crossing across a billboard, $b_i$ at a time $t_i$. Now assume that the advertisement content of an E-Commerce house is displayed on that billboard in the slot $[t_x,t_y]$, and $t_i \in [t_x,t_y]$. Then $u_i$ is likely to be influenced by the advertisement content with certain probability. The billboard $b_i$ will influence the trajectory $t_j$ with probability $Pr(b_i,t_j)$. One of the way to calculate this value as,  $Pr(b_i,u_i) = \frac{Size(b_i)}{\underset{b_{i} \in \mathbb{B}}{max} \ Size(b_{i})} $ where $Size(b_{i})$ is the billboard panel size. We adopt this probability settings in our experiments as well. Although it can be calculated in several ways depending to the needs of applications \cite {10.1145/3292500.3330829, 8604082,zhang2020towards}.
\par For any positive integer $n$, $[n]$ denotes the set $\{1,2, \ldots, n\}$.  We denote a billboard slot as a tuple of two attributes in which first attribute contains a unique billboard number, and second one contains starting and ending time of that slot. Considering this notation the set of all billboards can be expressed as $\mathcal{BS} = \{(b_i,[t_j,t_j+\Delta]): i \in [m] \text{ and } j \in \{1, 1+\Delta, 1+2\Delta, \ldots, \frac{T}{1+\Delta}\} \}$, and if total number of billboard slots $|\mathcal{BS}|=Q$ then, we say $Q= m \cdot \frac{T}{\Delta}$. As mentioned in the billboard slot selection problem, two databases namely Trajectory and Billboard Database are part of inputs and they are defined in Definition \ref{Def:TD} and \ref{Def:BD}, respectively.  

\begin{definition} [Trajectory Database] \label{Def:TD}
A trajectory database $\mathbb{D}$ contains the tuples of the following form: $<u_{id}, \texttt{u\_loc},\texttt{time-slot}>$. Below are the descriptions for each attribute:
\begin{itemize}
\item $u_{id}$: This attribute holds unique id assigned to a person.
\item \texttt{u\_loc}: This attribute contains information about the location of the users identified by $u_{id}$.
\item \texttt{time}: This attribute holds information related to time duration.
\end{itemize}
In $\mathbb{D}$, the presence of the tuple $<u_{225}, \texttt{Republic\_Airport}, [1400,1600]>$  means that the person with the unique id $u_{225}$ was present at \texttt{Republic\_Airport} during the time interval from $1400$ to $1600$.
\end{definition}
 In Definition \ref{Def:TD}, we have only included the necessary attributes, but in actual datasets may have additional attributes, such as \texttt{trip\_id} and \texttt{vehicle\_id} etc. It may so happen that in $\mathbb{D}$, one user with a different time interval can be present in many tuples because of the mobility of the user. The set of all persons covered by $\mathbb{D}$ is denoted by $\mathcal{U}=\{u_1, u_2, \ldots, u_{p}\}$. Next, we describe Billboard Database in Definition \ref{Def:BD}.

\begin{definition}[Billboard Database]\label{Def:BD}
 A Billboard database $\mathbb{B}$ is a collection of tuples of the form $<b_{id}, \texttt{b\_{loc}}, \texttt{b\_cost}>$, where the meaning of each of the attributes are as follows:
 \begin{itemize}
 \item $b_{id}$: This store unique ids assigned to billboards.
 \item \texttt{b\_loc}: This stores billboard location information.
  \item  \texttt{b\_cost}: This attribute stores the costs for corresponding billboard for one slot. 
 \end{itemize}
 \end{definition}
If $\mathbb{B}$ contains a tuple $<b_{139}, \texttt{Republic\_Airport}, 25>$ means that billboard $b_{139}$ is already placed in the \texttt{Republic\_Airport} and $\$25$ is the cost for renting this billboard for one slot.

Consider, a billboard slot $b_{i} \in \mathcal{BS} $ running an advertisement of the brand `ABC' at the location $\texttt{Republic\_Airport}$ for the time interval $[t_{i}, t_{i}+ \Delta]$ and the person $u_{j}$ is present at $\texttt{Republic\_Airport}$ within the $\lambda$ distance from the billboard for the time duration of $[t_{j},t_{j}+\Delta]$. If $[t_{i}, t_{i}+ \Delta] \ \bigcap \ [t_{j},t_{j}+\Delta] \neq \emptyset$, then we can say that the $u_{i}$ is influenced by the advertised item with probability $Pr(b_i,u_j)$. Now, we defines the notion of influence in Definition \ref{Def:3}.

\begin{definition} [Influence of Billboard Slots] \label{Def:3}
Given a trajectory database $\mathbb{D}$, billboard database $\mathbb{B}$, and a subset of billboard slots $\mathcal{C} \subseteq \mathcal{BS}$, we denotes the influence of $\mathcal{C}$ as $I(\mathcal{C})$, and referred to be the addition of influence probability for each users. Mathematically, we defined it in Equation \ref{Eq:1}.
 
 \begin{equation} \label{Eq:1}
 I(\mathcal{C})= \underset{u_j \in \mathbb{U}}{\sum} [1 \ - \ \underset{b_i \in \mathcal{C}}{\prod} (1-Pr(b_i,u_j))]
 \end{equation}
\end{definition} 

Now, from Equation \ref{Eq:1} it is clear that the influence function $I()$ maps every possible subset of billboard slots to corresponding influence, i.e., $I: 2^{\mathcal{BS}} \longrightarrow \mathbb{R}^{+}_{0}$. Here,  $I(\emptyset)=0$. Next, we state the property of $I()$ in Lemma \ref{Lemma:1}.
 
 \begin{lemma} \label{Lemma:1}
Given a trajectory database $\mathbb{D}$, and billboard database $\mathbb{B}$, the influence function $I()$ is non-negative, monotone, and submodular in nature.
 \end{lemma}

Since the selection of billboard slots involves a monetary cost, only a minimal number of them can be chosen. In any commercial campaigns, the objective is to maximize the views within the budget constraints. This raises the query of which $k$ billboard slots should be selected to maximize the views (i.e., influence). This problem is known as the Influential Billboard Slot Selection \textsl{IBSS} problem, as stated in Definition \ref{Def:Problem}.  
 
\begin{definition} [\textsl{IBSS} Problem]\label{Def:Problem}
Given a set of billboard slots $\mathcal{BS}$, a set of users $\mathcal{U}$, with the influence probability $Pr(b_i,u_j)$ for all $b_i \in \mathcal{BS}$ and $u_j \in \mathcal{U}$, IBSS Problem asks to find a subset $\mathcal{C} \subseteq \mathcal{BS}$ such that $|\mathcal{C}| = k$ and $I(\mathcal{C})$ is maximized as stated in Equation \ref{Eq:1}. Mathematically, this problem can be posed using Equation \ref{Eq:Problem}.
  
  \begin{equation} \label{Eq:Problem}
  \mathcal{C}^{OPT}= \underset{\mathcal{C} \ \subseteq \ \mathbb{BS}, \ | \mathcal{C}|=k}{argmax} \ I(\mathcal{C})
  \end{equation}
\end{definition} 
  
From the computational perspective, \textsf{IBSS Problem} has been written in the following  text box.
\begin{center}
\begin{tcolorbox}[title=\textsl{\textsl{IBSP} Problem}, width=8.5cm]
\textbf{Input:} Billboard Slots $\mathcal{BS}$, Trajectory Database $\mathbb{D}$, Budget $k$, Influence Function $I()$.

\textbf{Problem:} Find out $\mathcal{C} \subseteq \mathcal{BS}$ for $|\mathcal{C}|=k$ such that $I(\mathcal{C})$ is maximum.
\end{tcolorbox}
\end{center}    
  
We take any instance of the \textsl{IBSP} Problem $I=(\mathbb{BS}, \mathcal{C},k)$, and show that the \textsl{IBSP} Problem is \textsf{NP-Hard} in Theorem \ref{Th:1} by reducing it to Set Cover Problem. Due to space limitations proofs of many lemmas and theorems are omitted here which are available in the associated technical report with this paper. 

  \begin{theorem} \label{Th:1}
  Given a trajectory database $\mathbb{D}$, and billboard database $\mathbb{B}$, the IBSS Problem is \textsf{NP-hard} and hard to approximate within any constant factor.
  \end{theorem}
    
%\begin{proof} (Outline)
%We prove this statement by a reduction from the Hitting Set Problem. We denote an arbitrary instance of the set cover problem by $I^{'}=(\mathcal{U}, \mathcal{X}, k^{'})$. Here, $\mathcal{U}$ is the ground set, $\mathcal{X}$ is the collection of the subsets over the ground set. The goal here is to choose $k^{'}$ many elements from the ground set $\mathcal{U}$ such that every subset in $\mathcal{X}$ contains at least one element from the chosen elements. It is well known that this problem is NP-Hard \cite{vazirani2001approximation}.
%\par Now, we provide a polynomial time reduction from the Hitting Set Problem to the Influential Billboard Slot Selection Problem. Without loss of generality, we assume that the elements of the ground set are the subset of the set of natural numbers and starting from $1$. Also, for the simplicity we assume that there is only one slot in a billboard. Now, the  construction is as follows. For every $i \in \mathcal{U}$, we create one location with its id as $\ell_{i}$, one billboard with id $b_{i}$, and place the billboard at that location. For every subset $x \in \mathcal{X}$, we create one trajectory that contains the locations. We fix $k=k^{'}$. We want to influence all the trajectories. Now, it is easy to observe that the hitting set problem instance will have a solution of size $k^{'}$ if and only if the influential billboard slot selection problem instance has a solution of size $k$. Due to space limitation, we have only given an outline of the proof. 
%\end{proof}      
 
 \section{Proposed Solution Methodologies}  \label{Sec:PA}
In this section we discuss the proposed solutions for addressing the \textsf{IBSS Problem}. The first approach involves leveraging the submodular function maximization technique, specifically the pruned submodularity graph. The second one performs spatial clustering on the top of submodular function maximization.

\subsection{Submodularity Graph Based Approach} \label{Sec:PSG}
In Definition \ref{Def:MIG} we describe the marginal gain of a slot, which is the starting point for addressing this problem. Given a subset of billboard slots $\mathcal{C} \subseteq \mathcal{BS}$ and a billboard slot $b \in \mathcal{BS}$ its marginal gain is stated in Definition \ref{Def:MIG}. 

\begin{definition} [Marginal Influence Gain of a Billboard Slot] \label{Def:MIG}
For any subset of billboard slots $\mathcal{C} \ \subseteq \ \mathcal{BS}$, the marginal gain of a slot, $b \in \mathcal{BS} \setminus \mathcal{C}$ is  denoted as $\Delta (b| \mathcal{C})$, and mathematically, it is defined using Equation No. \ref{Eq:MG}. 
\begin{equation}\label{Eq:MG}
\Delta (b \ | \ \mathcal{C}) = I(\mathcal{C} \ \cup \ \{b\}) \ - \ I(\mathcal{C})
\end{equation}

Here, $\Delta (b| \mathcal{C})$ represents the difference in influence obtained when the billboard slot $b$ is included in $\mathcal{C}$ and when $b$ is not included with $\mathcal{C}$. 
\end{definition}

According to Lemma \ref{Lemma:1}, influence function $I()$ holds the  submodularity property. Therefore, we can pose the $\textsf{IBSS Problem}$  into a submodular function maximization problem with cardinality constraint. One traditional approach to solve submodular maximization problem is the incremental greedy approach based on marginal gain computation. For a given positive integer $k$, starting with an empty set and in each iteration we select a billboard slot that causes maximum marginal gain. According to Nemhauser et al. \cite{nemhauser1978analysis,fisher1978analysis}, this method gives $(1-\frac{1}{e})$-factor approximation guarantee. However, this method requires $\mathcal{O}(n)$ computation of marginal gain in each iteration, where $|\mathcal{BS}|=n$, leading to significant computational burdens for real-world problems. To solve this issue, an efficient approach has been proposed that uses a combinatorial optimization called pruned submodularity graph \cite{zhou2017scaling,10.1007/978-3-031-22064-7_17}. In this paper, we adopt this technique to solve the IBSS Problem. In Definition \ref{Def:PS_Graph}, we define the Pruned Submodularity Graph.

 \begin{definition}  [Pruned Submodularity Graph (PSG)] \label{Def:PS_Graph}
Pruned Submodularity graph is a directed, weighted graph $(V,E,\mathcal{W})$, where the vertex set of $G$, i.e., $V(G)$ contains the set of billboard slots $\mathcal{BS}$, and between each pair of slots there is a directed edge $e = (b_x \rightarrow b_y)  \in E(G)$ has weight $\mathcal{W}$ and it can be defined in Equation \ref{Eq:weight}.
\begin{equation} \label{Eq:weight}
\mathcal{W}_{b_x \rightarrow b_y}= I(b_y \ \vert \ b_x) \ - \ I(b_x \ \vert \ \mathcal{BS} \setminus b_x)
\end{equation}

Here, $(b_x,b_y)$ denotes any pair of billboard slot and the edge set $E(G)=\{(b_x,b_y):\ x,y \in [n] \text{ and } x \neq y\}$.
\end{definition}

According to the Algorithm \ref{Algo: Algorithm1_Simple Greedy}, the weight of any edge $(b_xb_y)$, i.e., $\mathcal{W}_{b_x \rightarrow b_y}$, computes the net loss during maximizing $I()$ on the reduced set $\mathcal{C}^{'} \subseteq \mathcal{BS}$ with $b_y$ removed and $b_x$ retained. In Equation \ref{Eq:weight}, $I(b_y \ \vert \ b_x)$ denotes maximum possible influence billboard slot $b_y$ can be offer while slot $b_x$ is already given. On the other hand, $I(b_x \ \vert \ \mathcal{BS} \setminus b_x)$ denotes the minimal influence slot $b_x$, can contribute to the solution $\mathcal{C}$ as $I()$ function holds by submodularity i.e, $I(x \vert \mathcal{C}) \geq I(x \vert \mathcal{BS} \setminus x)$. Therefore, a minimal $I(b_y \vert b_x)$ shows that $b_y$ is not so important while a larger value of  $I(b_x  \vert \mathcal{BS} \setminus b_x)$ implies that $b_x$ is always important. In this paper rest of our discussions we will use term `slot' in place of  `vertex' and vice versa. Next, in Definition \ref{Def:Divergence}, we define slots divergence in context of \textsl{PSG}.

\begin{definition}  [Divergence of a Billboard Slot] \label{Def:Divergence}
Given a pruned submodularity graph $G(V,E,\mathcal{W})$, the divergence of a billboard slot $b_y \in \mathcal{C}$ from the reduced ground set $\mathcal{C}^{'}$ can be defied as, $\mathcal{W}_{\mathcal{C}^{'}b_y} = \underset{x \in \mathcal{C}^{'}}{min} \ \mathcal{W}_{xy}$.
\end{definition}

\par Now, we describe the proposed solution approach. Our proposed solution methodology involves a preprocessing task that removes billboard slots with an individual influence of $0$ from the list. Next, in Line No. $8$ we generate the Pruned Submodularity Graph with the remaining billboard slots. In Line No. $10$ each iteration of the \texttt{while} loop, we randomly pick $h \cdot \log n$ many slots and put them into $\mathbb{U}$. We then compute the value of $\mathcal{W}_{\mathbb{U}d}$ for all remaining billboard slots in $\mathbb{BS}$, where $\mathcal{W}_{\mathbb{U}d}$ is defined in Definition \ref{Def:Divergence}. We remove a fraction of $(1-\frac{1}{\sqrt{l}})$ many slots having the smallest value of $\mathcal{W}_{\mathbb{U}d}$ from the list. Once the pruning step is complete, we use the incremental greedy approach to select $k$ billboard slots from the reduced ground set. Algorithm \ref{Algo: Algorithm1_Simple Greedy} describes this process in the form of pseudocode. Subsequently, we analyze Algorithm \ref{Algo: Algorithm1_Simple Greedy} to understand its time and space requirement.

 \begin{algorithm}[h]
\SetAlgoLined
\KwData{$\mathcal{BS}$, $\mathbb{D}$, $h$, $I()$, $l$, and, $k$}
\KwResult{ $\mathbb{X} \subseteq \mathcal{BS}$, $\ \text{such that} \ |\mathbb{X}|= k$ }
 Initialize a subset $\mathbb{X} \leftarrow \emptyset$, $\ \mathcal{\overline{BS}} \leftarrow \mathcal{BS}$\;
 \For{$\text{All }d \in \mathcal{\overline{BS}}$}{
 \If{$I(d)==0$}{
 $\mathcal{\overline{BS}} \longleftarrow \mathcal{\overline{BS}} \setminus \{d\}$\;
 }
 }
 Initialize $\mathcal{\overline{S}} \leftarrow \emptyset$, $|\mathcal{\overline{BS}}|=n$, $\mathbb{U} \longleftarrow \emptyset$ \;
 $\text{Generate \textsl{PSG} with }\mathcal{\overline{BS}}$\;
 \While{$|\mathcal{\overline{BS}}| > \ h \cdot \log n$}{
 $\text{Pick }h \cdot \log n \text{ slots randomly from }\mathcal{\overline{BS}}  \text{ and store in }\mathbb{U}$\; 
 $\mathcal{\overline{BS}} \longleftarrow \mathcal{\overline{BS}} \setminus \mathbb{U}$, $\mathcal{\overline{S}} \longleftarrow \mathcal{\overline{S}} \cup \mathbb{U} $\;
 \For{$\text{All }d \in \mathcal{\overline{BS}} $}{
 $\mathcal{W}_{\mathbb{U}d} \longleftarrow \ \underset{u \in \mathbb{U}}{min} \ [I(d|u)-I(u| \mathcal{\overline{BS}} \setminus \{u\})]$\;
 }
 $\text{Take out }(1-\frac{1}{\sqrt{\ell}}) \cdot |\mathcal{\overline{BS}}| \text{ elements from }\mathcal{\overline{BS}} \text{ having smallest }\mathcal{W}_{\mathbb{U}d}$\;
 }
 $\mathcal{\overline{BS}} \longleftarrow \mathcal{\overline{BS}} \  \bigcup \ \mathcal{\overline{S}} $\;
 \While{$|\mathbb{X}| \neq k$}{
 $d^{*} \longleftarrow \underset{d \in \mathcal{\overline{BS}} \setminus \mathbb{X}}{argmax} \ I(\mathbb{X} \cup \{d\}) - I(\mathbb{X})$\;
 $\mathbb{X} \longleftarrow \mathbb{X} \cup \{d^{*}\}$\;
 }
Return $\mathbb{X}$\;
 \caption{ \textsl{PSG} + Incremental Greedy approach for \textsl{IBSP} Problem}
 \label{Algo: Algorithm1_Simple Greedy}
\end{algorithm}	
\par The initialization process of Line No. $1$  will take  $\mathcal{O}(1)$ time. Next computing influence of any arbitrary billboard slots will take $\mathcal{O}(t)$ time if there are $t$ tuples in the trajectory database. In the worst case if there is $n$ number of billboard slots then total time required for Line $2$ to $6$ is $\mathcal{O}(n\cdot t)$. Line No. $7$ initialization statement will take $\mathcal{O}(1)$ time. Now, in Line $8$ graph is constructed for pruning and in the worst case if we consider that there is $n$ number of non zero billboard slots then there is $\mathcal{O}(n^{2})$ billboard slot pairs. From Equation \ref{Eq:weight}, it is clear that computing edge weight will take for any pair of slots will take $\mathcal{O}(n \cdot t)$ time.  Hence construction of the PSG having $\mathcal{O}(n^{2})$ edges will take $\mathcal{O}(n^{3}\cdot t)$ time. It can be observed that the \texttt{while loop} of Line $9$ will execute for $\mathcal{O}(\log_{\sqrt{\ell}} n)$ times. Also, sampling $\mathcal{O}(\log n)$ many elements randomly from $\mathcal{\overline{BS}}$ will take $\mathcal{O}(\log n)$ time. For every $u \in \mathbb{U}$, the required computations of Line $13$ will take $\mathcal{O}(1)$ time. The reason here is that we can reuse the computations which we have done for computing the corresponding edge weights of the pruned submodularity graph. Considering $\mathcal{O}(\log n)$ many elements in $\mathbb{U}$, for one billboard slot computation of Line $13$  will take $\mathcal{O}(\log n)$ time. As there are $\mathcal{O}(n)$ billboard slots, hence execution from Line $12$ to $14$ will take $\mathcal{O}(n\cdot\log n)$ time. Sorting these weight values will take $\mathcal{O}(n\cdot \log n)$ time. In the worse case, removing $(1-\frac{1}{\sqrt{\ell}})$ fraction elements will take $\mathcal{O}(n)$ time. Hence, at each iteration \texttt{while loop} will take  $\mathcal{O}(\log n + n \cdot \log n + n)=\mathcal{O}(n \cdot \log n)$. So, Line No. $9$ to $16$ will take $\mathcal{O}(n \cdot \log n \cdot \log_{\sqrt{\ell}} n)=\mathcal{O}(n \cdot \log^{2} n)$ times as the \texttt{while loop} will execute for $\mathcal{O}(\log_{\sqrt{\ell}} n)$ many times. After completion of the \texttt{while loop} the size of $\mathcal{\overline{BS}}$ will be of $\mathcal{O}(\log^{2}n)$. Choosing $k$ elements from the reduced ground set will take $\mathcal{O}(k \cdot t \cdot \log^{2}n)$. Hence, the total time requirement for executing Algorithm \ref{Algo: Algorithm1_Simple Greedy} will be $\mathcal{O}(n \cdot t + n^{3} \cdot t+ n \cdot \log^{2}n+k \cdot t \cdot \log^{2}n)=\mathcal{O}( n^{3} \cdot t+ n \cdot \log^{2}n+k \cdot t \cdot \log^{2}n)$. Extra space requirement for Algorithm \ref{Algo: Algorithm1_Simple Greedy} is as follows. To store the adjacency matrix of the pruned submodularity graph will take $\mathcal{O}(n^{2})$ space. Additional space needed to store $\mathbb{S}$, $\mathbb{U}$, and $\mathbb{S}^{'}$ is $\mathcal{O}(n)$, $\mathcal{O}(\log n)$, and $\mathcal{O}(n)$ respectively. Storing the intermediate computations while computing the edge weights of the pruned submodularity graph will take $\mathcal{O}(n)$ space. Hence, the total space requirement will be of $\mathcal{O}(n^{2})$. Hence, Theorem \ref{Th:Time_and_Space} holds.

%       Now, it is clear that Line No. $18$ to $21$ \texttt{while loop} will execute for $k$ times and calculating marginal gain will take $\mathcal{O}(n\cdot t)$ time. So, \texttt{while loop} takes $\mathcal{O}(k \cdot n \cdot t)$ time. Hence, we can conclude that total time requirement for Algorithm \ref{Algo: Algorithm1_Simple Greedy} will be $\mathcal{O}(n \cdot t + n^{2} \cdot t + n \cdot \log^{2} n + k \cdot n \cdot t)$. Finally we can rewrite it as $\mathcal{O}( n^{2} \cdot t + n \cdot \log^{2} n)$.

%\par Extra space requirement for Algorithm \ref{Algo: Algorithm1_Simple Greedy} is to store pruned submodularity graph is $\mathcal{O}(n^{2})$ and to storing edge weights adjacency matrix is used.  So total space requirment for Algorithm \ref{Algo: Algorithm1_Simple Greedy} will be $\mathcal{O}(n^{2} + n + \log n)=\mathcal{O}(n^{2})$. Therefore we can say Theorem \ref{Th:Time_and_Space} holds.

\begin{theorem} \label{Th:Time_and_Space}
Algorithm \ref{Algo: Algorithm1_Simple Greedy} takes $\mathcal{O}( n^{3} \cdot t+ n \cdot \log^{2}n+k \cdot t \cdot \log^{2}n)$ for time and $\mathcal{O}(n^{2})$ for space to execute.
\end{theorem}

\par Now, we discuss the performance guarantee of Algorithm \ref{Algo: Algorithm1_Simple Greedy} which directly follows from the study by Zhou et al. \cite{zhou2017scaling} and presented in Theorem \ref{Zohu_Result}.

\begin{theorem} \cite{zhou2017scaling} \label{Zohu_Result}
Given a set of billboard slots $\mathcal{\overline{BS}}$, if we apply pruned submodularity graph on $\mathcal{\overline{BS}}$ then the size of the reduced output $\mathcal{BS}^{'}$ in Algorithm \ref{Algo: Algorithm1_Simple Greedy} will be $|\mathcal{BS}^{'}|=\frac{\ell p}{\log \sqrt{\ell}} \cdot k \cdot \log^{2}n$ with very high probability i.e., $n^{1-qp} \cdot \log_{\sqrt{\ell}}n$, $\forall \ b \in \mathcal{\overline{BS}} \setminus \mathcal{BS}^{'}$, $\mathcal{W}_{\mathcal{BS}^{'}b} \leq 2 \cdot \mathcal{W}_{\mathcal{BS}^{*}}$. Thus the greedy approach on $\mathcal{\overline{BS}}$ generate a solution $\mathcal{C}^{'}$ such that
\begin{equation} \label{Eq:Approximation}
I(\mathcal{\mathcal{C}}^{'}) \geq (1-\frac{1}{e}) \cdot (I(\mathcal{\mathcal{C}}^{OPT})- 2k\epsilon)
\end{equation}
where $\mathcal{C}^{OPT}$ is the optimal solution of size $k$. 
\end{theorem}
 \par The above analysis shows a trade off between some parameter like the approximation bound, size of $\mathcal{BS}^{'}$, computational cost using $\epsilon$. From Equation \ref{Eq:Approximation}, it is clear that if $\epsilon$ is smaller then optimal solutions size become larger. In Algorithm \ref{Algo: Algorithm1_Simple Greedy}, two parameter $\ell$ and $h$ controls the efficiency. The parameter $\ell$ controls the shrinking rate of billboard slots at each iteration while $h$  controls the size of output billboard slot set i.e., $\mathcal{BS}^{'}$. Hence, it is very important to choose the right values for $h$ and $\ell$. In the previous study by Zhou et al. \cite{zhou2017scaling}, uses both $h$ and $\ell$ as $8$ and they have showed that for this parameter setting pruning method converge quickly. So, in our proposed approach we also consider the same parameter setting for $h$ and $\ell$.

\subsection{Spatial Clustering Approach}
Though the method described in Section \ref{Sec:PSG} provides an approximation guarantee and requires much less computational time compared to the incremental greedy approach, still this is not adequate to apply this approach when the billboard and trajectory database size is huge. To address this issue, in this section, we propose a spatial partitioning approach which contains three steps:
\begin{itemize}
\item The first step is to partition the billboard slot set $\mathcal{BS}$ into a set of clusters as per the influence overlap.
\item The next step is to find a threshold by calculating the average influence value of all clusters. Based on the threshold value, some of the clusters are pruned.
\item From the remaining clusters, local influential billboard slots are aggregated. Finally, Pruned Submodularity Graph+Incremental Greedy approach is applied to generate the global solution.
\end{itemize}

\paragraph{\textbf{Partition-based Framework}} Most of the existing solutions for the IBSS Problem does not consider the distance from the billboard as a parameter of influence. However, in this study we consider a distance factor as well. We assume that within distance if a trajectory come near to a  billboard location then the trajectory must be influenced by that particular billboard slot with some probability. In our datasets we observe that there is a small overlap between billboard slots of different locations to their influenced trajectories as most of the trajectories move in between $5$ kms and the same observation also reported by Zhang et. al. \cite{zhang2021minimizing} in their study. Therefore, we partition billboard slots into several small clusters and take locally maximum influential billboard slots from each cluster and merge them to get global solutions for selecting influential billboard slots. As the size of each cluster is smaller than the ground set ($\mathcal{BS}$) of billboard slots, the computation cost reduces significantly and generating more or less same influence. In this regard, first we describe the notion of $\theta$-Partition in Definition \ref{Def:Partition}.

\begin{definition}  [Partition]\cite{zhang2020towards} \label{Def:Partition} 
A partition of billboard slots set $\mathcal{BS}$, is a set of clusters $\{{\mathcal{K}_1,\mathcal{K}_2,\ldots,\mathcal{K}_m}\}$, where $\mathcal{BS}= \mathcal{K}_1 \cup \mathcal{K}_2 \cup \ldots \cup \mathcal{K}_m$, such that $\forall$ x $\neq$ y, $\mathcal{K}_x \;\cap$ $\mathcal{K}_y = \emptyset$.
\end{definition}
Now, the question is how we can cluster the billboard slots. The easiest way to partition the slots randomly. However, random partitioning of the billboard slots scheme does not give satisfactory results as there exists significant influence overlap between clusters. To minimize the influence overlap between different clusters, we use the notion of overlap ratio \cite{zhang2020towards} as presented in Definition \ref{Def:Overlap Ratio}.

\begin{definition}  [Overlap Ratio] \label{Def:Overlap Ratio} For given any two clusters of billboard slots $\mathcal{K}_x$ and $\mathcal{K}_y$, the overlap ratio between $\mathcal{K}_x$ and $\mathcal{K}_y$ with respect to $\mathcal{K}_x$ is denoted by $\mathcal{V}_{x,y}$ and defined using Equation No. \ref{Eq:6}.
\begin{equation} \label{Eq:6}
\mathcal{V}_{x,y} = \ \underset{\forall \; \mathcal{S}_x \; \subseteq \;\mathcal{K}_x}{argmax} \  \frac{\sigma(\mathcal{S}_x|\mathcal{K}_y)}{I(\mathcal{S}_x)} 
\end{equation}
where, $\mathcal{S}_x$ is a subset of $\mathcal{K}_x$ and $\sigma$($\mathcal{S}_x|\mathcal{K}_y$) denotes the influence overlap between $\mathcal{S}_x$ and $\mathcal{K}_y$ and mathematically it can be written using Equation \ref{Eq:7}.
\begin{equation} \label{Eq:7}
\sigma(\mathcal{S}_x|\mathcal{K}_y) = I(\mathcal{S}_x) + I(\mathcal{K}_y) - I(\mathcal{S}_x \cup \mathcal{K}_y)
\end{equation}
\end{definition}
 
Our goal is to partition the billboard slots so that the overlap among clusters is minimized. To resolve this issue we introduce the notion of $\theta$-partitioning approach in which clusters are generated and the influence overlap between clusters are minimized. For a given overlap ratio we present $\theta$-partition in which $\theta$ is a parameter and can be varied by the user. Here, $\theta$ balances the clusters size and influence overlap between any two clusters. Now, we describe the $\theta$-partition in definition \ref{Def:Theta Partition}.

\begin{definition}  [$\theta$-Partition] \label{Def:Theta Partition} 
A partition of the billboard slots $\mathcal{P} =\{{\mathcal{K}_1, \mathcal{K}_2,\ldots,\mathcal{K}_m}\}$ is said to be $\theta$-partition if $\forall x,y \in [m]$ and the overlap ratio among the clusters $\{{\mathcal{K}_x, \mathcal{K}_y}\}$ is always less than $\theta$, where $\theta$ lies between $0$ and $1$.
\end{definition}

 \begin{lemma}\cite{zhang2020towards} \label{Lemma:2}
For a given subset of billboard slots, $\mathcal{Z} \subseteq \mathcal{BS}$ where $|\mathcal{Z}|=q$, and a $\theta$-partition $\{{\mathcal{K}_1, \mathcal{K}_2,\ldots,\mathcal{K}_m}\}$ of $\mathcal{BS}$, $I(\mathcal{Z}) \geq \frac{1}{2} \underset{\mathcal{Z}_{i} \in \mathcal{Z}}{\sum}I(\mathcal{Z}_{i})$, if $ q \in [m]$ and $q \leq (1 + \frac{1}{\theta})$. 
 \end{lemma} 

\begin{proof}
Consider, $\mathcal{Z} = \{\mathcal{Z}_{1},\mathcal{Z}_{2},\ldots,\mathcal{Z}_{q}\}$ and the billboard slots in $\mathcal{Z}$ are sorted in descending order according to their individual influence, i.e., $I(\mathcal{Z}_{1}) \geq I(\mathcal{Z}_{2}) \geq \ldots \geq I(\mathcal{Z}_{q})$. According to Definition \ref{Def:Theta Partition} given two slots $\mathcal{Z}_{x}$ and $\mathcal{Z}_{y}$, $\mathcal{Z}_{y}$ has at most $\theta$ percent of overlap with with the slots of $\mathcal{Z}_{x}$ i.e., $I(\mathcal{Z}_{x} \cup \mathcal{Z}_{y}) \geq I(\mathcal{Z}_{x}) + (1 - \theta)I(\mathcal{Z}_{y})$. We denote the average influence for all $\mathcal{Z}_{i} \in \mathcal{Z}$ as  $M_{I(\mathcal{Z})} = \frac{1}{q} \sum_{i=1}^q I(\mathcal{Z}_{i})$. Now, for all subset of $\mathcal{Z}$ we can write:\\
\begin{align}
I(\mathcal{Z}) & \geq I(\mathcal{Z}_{1}) + (1-\theta)I(\mathcal{Z}_{2})+ \ldots +[1-(q-1)\theta]I(\mathcal{Z}_{q}) \nonumber\\
               & = \sum_{i=1}^q I(\mathcal{Z}_{i}) - \theta[I(\mathcal{Z}_{2})+2I(\mathcal{Z}_{3}) + \ldots + (q-1)I(\mathcal{Z}_{q})]\nonumber\\
               & = \sum_{i=1}^q I(\mathcal{Z}_{i}) - \theta[\sum_{i=2}^q I(\mathcal{Z}_{i})+\sum_{i=3}^q I(\mathcal{Z}_{i}) + \ldots + I(\mathcal{Z}_{q})]\nonumber\\
               & \geq \sum_{i=1}^q I(\mathcal{Z}_{i}) - \theta[M_{I(\mathcal{Z})} + 2M_{I(\mathcal{Z})} + \ldots + (q-1)M_{I(\mathcal{Z})}]\nonumber\\
               & = \sum_{i=1}^q I(\mathcal{Z}_{i}) - \theta  \frac{q(q-1)}{2} M_{I(\mathcal{Z})}
\end{align}

Now, as we have, $m \leq (1 + \frac{1}{\theta})$, from this we observe that when, $\theta  \frac{q(q-1)}{2} M_{I(\mathcal{Z})} \leq \frac{q}{2} M_{I(\mathcal{Z})}$ we have,
\begin{align}
I(\mathcal{Z}) &\geq \sum_{i=1}^q I(\mathcal{Z}_{i}) - \frac{q}{2} M_{I(\mathcal{Z})}\nonumber\\
              &\geq \sum_{i=1}^q I(\mathcal{Z}_{i}) - \frac{q}{2} \cdot \frac{1}{q} \sum_{i=1}^q I(\mathcal{Z}_{i})\nonumber\\
            & \geq \frac{1}{2}\sum_{i=1}^q I(\mathcal{Z}_{i})
\end{align}
\end{proof}

\paragraph{\textbf{Finding Theta-Partition}} Choice of $\theta$-Partition plays an important role in our problem as the solution methodology proposed in this section takes this as an input. As mentioned previously, in $\theta$-partition $\theta$ determines the influence overlap between any two clusters. So, our goal is to find best choice of $\theta$ so that the effectiveness and efficiency of the proposed solution methodology increases. One point is that we want the billboard slot partitions to be balanced. Now, the problem of $\theta$-partitioning can be viewed as the $k$-cut problem as follows. Consider each billboard slot as a vertex and there is an edge between two vertices if there is a significant influence overlap between these two slots. Now, it is easy to observe that $\theta$-Partitioning of the billboard slots is equivalent to solve the $k$-cut problem in the constructed graph. It has been reported in the literature that $k$-Cut Problem is \textsf{NP-hard} \cite{10.1007/3-540-57182-5_65}. So, the remedy is to use approximate $\theta$-partition in which each billboard slot is considered as a single cluster and then iteratively merges two clusters into one cluster if their overlap ratio is greater than $\theta$ value. Now, for each pair of clusters $\mathcal{K}_i, \mathcal{K}_j \subseteq \mathcal{BS}$, if their overlap ratio greater than $\theta$, then $\mathcal{K}_i, \mathcal{K}_j$ will be merged. Now, by repeating this process we can get an approximate $\theta$-partition until no clusters in $\mathcal{BS}$ merge further. Theoretically, it may so happen that all the slots in $\mathcal{BS}$ will be merged into a single cluster. However, in our experiments we have observed that for all location types $\theta$-Partition generates a reasonable number of clusters. For example, in case of location type `Beach' and the value of $\theta = 0.2$, an approximate $\theta$-Partition generates $72$ clusters. Further, using the Spatial Partitioning approach from these $72$ clusters, $40$ clusters remain after pruning for further processing, as presented in Figure \ref{Fig:5_Theta_Cluster}.

\begin{algorithm}[h]
\SetAlgoLined
\KwData{ Trajectory Database $\mathbb{D}$, Billboard Slots $\mathcal{BS}$, Influence function $I()$ and $\theta$.}
\KwResult{An Approximate $\theta$-partition $\mathcal{P}$ of $\mathcal{BS}$}
Initially Each slot of $\mathcal{BS}$ is a single cluster\;
%\While{ $\forall \ \mathcal{K}_i, \mathcal{K}_j \in \mathcal{BS}$}{
\While{ $\text{No. of Iteration} < \text{Preset Count}$}{
\For{$\text{Each } \mathcal{K}_i,\mathcal{K}_j \in \mathcal{BS}$}{
 $\text{Calculate Overlap Ratio Between }\mathcal{K}_i \text{ and }\mathcal{K}_j$\;
\eIf{$Overlap\_Ratio(\mathcal{K}_i, \mathcal{K}_j) \geq \theta$}{
  Merge $\mathcal{K}_i$ and $\mathcal{K}_j$ into one cluster.
}
{
   Exit.
 }
 }
 }
 Return $\mathcal{BS}$\;
 \caption{ Approximate $\theta$-Partition Approach}
 \label{Algo: Spatial_Partition}
\end{algorithm}

\par{\textbf{Time and Space Complexity Analysis:}} Now, we analyze the time and space complexity of Algorithm \ref{Algo: Spatial_Partition}. The initialization process for Line No. $1$ will take $\mathcal{O}(n)$ time. It can be observed that the \texttt{while loop} of Line $2$ will execute for some constant amount of time and it will take $\mathcal{O}(1)$ time for execution. In Line No. $3$ \texttt{For loop} will iterate $n$ many times and calculating overlap ratio in line no $4$ will take $\mathcal{O}(n)$, if size of the largest cluster is $n$. Line No. $5$ checking condition will take $\mathcal{O}(1)$ time and merging two clusters will take $\mathcal{O}(m+n)$ time if the size of the clusters are $m$ and $n$ respectively. Therefore, from Line No. $3$ to $10$ will take $\mathcal{O}(n^2)$ time approximately. Hence, total time taken by Algorithm \ref{Algo: Spatial_Partition} is $\mathcal{O}(n^2)$. Extra space requirement for Algorithm \ref{Algo: Spatial_Partition} is as follows. To store the billboard slots in $\mathcal{BS}$ will take $\mathcal{O}(n)$ space. After calculating overlap ratio in Line No. $4$, storing it will take $\mathcal{O}(1)$ space. Hence, the total space requirements for Algorithm \ref{Algo: Spatial_Partition} will be $\mathcal{O}(n)$.

\begin{theorem} \label{Th:Algo 2}
The computational time and space requirement of Algorithm \ref{Algo: Spatial_Partition} will be of $\mathcal{O}(n^2)$  and $\mathcal{O}(n)$, respectively.
\end{theorem}

\paragraph{\textbf{Partition-based Selection Method}} Now, we describe our proposed partition-based framework to solve the IBSS Problem. This solution approach uses popular divide and conquer paradigm of algorithm design. First, the billboard slots are partitioned into clusters by approximate $\theta$-Partitioning Technique. Assume that this method generates $m$ many clusters and they are $\mathcal{P}=\{ \mathcal{K}_1, \mathcal{K}_2,\ldots,\mathcal{K}_m \}$. Also assume that all the clusters are sorted by their size, and hence, $\mathcal{K}_m$ is the largest cluster. Now, the influence of individual clusters determines the importance of that cluster in a global solution. So, after calculating the influence of individual clusters, the average influence of all the clusters is calculated and set as a threshold value. If the influence of the billboard slots of a cluster is less than this threshold value means that such clusters are less important, and hence, they can be ignored. This pruning step is used to reduces computational cost greatly while keeping competitive influence quality. Now, from the remaining clusters, we use the following two methods to find out $k$ influential billboard slots.

\begin{itemize}
\item \textbf{First Approach}: In the first approach, after pruning remaining billboard slots are merged and from these we choose Top-$k$ Influential Billboard Slots using incremental greedy approach in which starting with an empty set in each iteration we pick up the billboard slot that causes maximum marginal influence gain. This leads to Algorithm \ref{Algo: Algorithm2}.   
\item \textbf{Second Approach}: In other approach, after pruning all the remaining clusters are merged. On these billboard slots, we apply the pruned submodulariy graph-based technique (As described in Algorithm \ref{Algo: Algorithm1_Simple Greedy}) to select Top-$k$ influential billboard slots.
\end{itemize}
 
\paragraph{\textbf{Approach 1 (Spatial Partitioning + Incremental Greedy Approach)}} Algorithm \ref{Algo: Algorithm2} describes this process. First the set $\mathbb{X}$ is initialized with $\emptyset$ which will subsequently contains the selected billboard slots. $Influence$ is an array of size $m$ where $Influence[i]$ contains the aggregated influence of the billboard slots belong to the $i$-th Cluster. As mentioned previously, $m$ and $n$ denotes the the number of clusters and the number of billboard slots, respectively. First for all the clusters, the aggregated influence probability is computed from Line $3$ to $5$, and subsequently, the average is computed in Line $6$. From Line $7$ to $11$, the pruning of clusters have been done. In Line $12$, all the remaining clusters are merged and from Line $13$ to $16$, Top-$k$ influential billboard slots are chosen using incremental greedy approach. This leads to Algorithm \ref{Algo: Algorithm3}. 
%we use the pruned submodularity graph-based technique as described in Algorithm \ref{Algo: Algorithm1_Simple Greedy} is applied to obtain the solutions from each cluster. Finally, the influential billboard slots from each cluster is merged to obtain the solution of the whole.  Algorithm \ref{Algo: Algorithm2} describes this procedure in the form of psudocode.
   
%\paragraph{\textbf{Spatial Partition+Incremental Greedy Approach}}
%A straightforward greedy approach is inefficient for billboard slot selection problems for large billboard datasets as it chooses billboard slots based on the unit marginal influence at each iteration and requires substantial computational time. To overcome this issue, we have introduced the Spatial Partition Incremental greedy approach runs in two phases, as shown in Algorithm \ref{Algo: Algorithm2}. In the first phase (lines $1$-$12$) for each cluster, it calculates the influence average of all billboard slots and stores it in an array. After that, it calculates a threshold value $\gamma$ by averaging the influence of all the clusters, and based on $\gamma$ value, less influential clusters are removed. Next, billboard slots from the remaining clusters are stored in $\mathbb{\overline{BS}}$. Now, in the second phase (lines $13$-$17$), using incremental greedy select top-$k$ influential billboard slots from $\mathbb{\overline{BS}}$.

\begin{algorithm}[h]
\SetAlgoLined
\KwData{ An Approximate $\theta$-partition $\mathcal{P}$ of $\mathcal{BS}$, Trajectory Database $\mathbb{D}$, Billboard Database $\mathbb{B}$ and $k$.}
\KwResult{ $\mathbb{X} \subseteq \mathcal{BS}$, $\ \text{such that} \ |\mathbb{X}|= k$ }
 Initialize $\mathbb{X} \leftarrow \emptyset$,  $Influence \longleftarrow Array(n,0)$\;
 $m = |\mathcal{P}|$, $n=|\mathcal{BS}|$\;

  \For{$\text{i} \leftarrow  1\; to\; m $}{
   $Influence[i]$ = $I(\mathcal{K}_i)$\;
 }
 $\gamma \longleftarrow \frac{\underset{i \in [n]}{\sum} I(\mathcal{K}_{i})}{n} $\;
\For{$\text{i} \leftarrow  1\; to\; m $}{
 \If{$Influence[i] < \gamma$}{
 $\mathcal{P} \longleftarrow \mathcal{P} \setminus \{\mathcal{K}_i\}$\;
 }
 }
Store billboard slots from remaining clusters to $\mathcal{BS}$\;
 \While{$|\mathbb{X}| \neq k$}{
 $d^{*} \longleftarrow \underset{d \in \mathcal{BS} \setminus \mathbb{X}}{argmax} \ I(\mathbb{X} \cup \{d\}) - I(\mathbb{X})$\;
 $\mathbb{X} \longleftarrow \mathbb{X} \cup \{d^{*}\}$\;
 }
Return $\mathbb{X}$\;
 \caption{ Spatial Partition + Incremental Greedy Approach for IBSS Problem}
 \label{Algo: Algorithm2}
\end{algorithm}	

\par{\textbf{Time and Space Complexity Analysis:}} Now, we analyze the time and space complexity of Algorithm \ref{Algo: Algorithm2}. All the initialization statements of Line $1$, will take $\mathcal{O}(1)$ time. In each cluster the number of billboard slots can be $\mathcal{O}(n)$ in the worst case. Computing influence of all the clusters will take $\mathcal{O}(m \cdot n \cdot t)$ time and taking their average will consume $\mathcal{O}(m)$ time. It is easy to observe that the pruning step will take $\mathcal{O}(m)$ time. Till Line $12$, the computational time requirement will be of $\mathcal{O}(m \cdot n \cdot t)$. In the worst case, $\mathcal{O}(n)$ many billboard slots will be there after pruning. Now for selecting Top-$k$ Influential Billboard Slots, we need to conduct $k$ iterations. In each iteration, we need to conduct $\mathcal{O}(n)$ many marginal gain computations. For each marginal gain computation, time requirement is of $\mathcal{O}(n \cdot t)$. Hence, the time requirement for the incremental greedy approach is of $\mathcal{O}(k \cdot n^{2} \cdot t)$. So the total time requirement of Algorithm \ref{Algo: Algorithm2} will be of $\mathcal{O}(m \cdot n \cdot t+k \cdot n^{2} \cdot t)$. It can be observed that the number of cluster can be at most the number of billboard slots. As the value of $m$ can be at most $n$, the complexity of Algorithm \ref{Algo: Algorithm2} will be of $\mathcal{O}(k \cdot n^{2} \cdot t)$. Additional space consumed by Algorithm \ref{Algo: Algorithm2} to store $\mathbb{X}$ which will be of $\mathcal{O}(k)$, the array $Influence$ which is of size $\mathcal{O}(m)$. Also, $\mathcal{O}(n)$ space is required to store the marginal gains of the billboard slots. Hence, the total space requirement will be of $\mathcal{O}(k+n+m)$. As $m$ can be of $\mathcal{O}(n)$, the space complexity of Algorithm \ref{Algo: Algorithm2} will be reduced to $\mathcal{O}(k+n)$.

%   In Line No. $4$ calculating average for for each cluster will take $\mathcal{O}(n) + \mathcal{O}(1)$ time if in worse case number of billboard slots in each cluster in $n$ and the outer loop in Line No. $2$ will iterate $m$ many times. So, total time requirement for Line No. $2$ to $5$ is $\mathcal{O}(m) \cdot \mathcal{O}(n\cdot 1) = \mathcal{O}(m\cdot n)$. Again, in Line No. $6$ to calculate $\gamma$ it will take $\mathcal{O}(m) + \mathcal{O}(1)$ time for $m$ clusters. Now, in Line No. $7$ \texttt{for} loop iterates $m$ times and inside the loop in Line No. $8$ to $10$ takes 
%$\mathcal{O}(1)$ times to execute. Therefore, in Line no $7$ to $11$ takes $\mathcal{O}(m \cdot 1) = \mathcal{O}(m)$ time in total. Now, in Line No. $12$, from $m$ clusters storing billboard slots into $\mathcal{\overline{BS}}$ will take $\mathcal{O}(m)$ time. Now, it is clear that the \texttt{while} loop at Line No. $13$ will iterate $\mathcal{O}(k)$ many times. For any billboard slot $b \in \mathcal{\overline{BS}} \setminus \mathbb{X}$, its marginal gain computation will take $\mathcal{O}(n \cdot t)$ time. So, total time taken by the \texttt{while} loop is $\mathcal{O}(k \cdot n \cdot t)$. Hence, the total time requirement of Algorithm \ref{Algo: Algorithm2} will be of $\mathcal{O}(m \cdot n + m + k\cdot n \cdot t)$. As $m < < mn$, this quantity is reduced to $\mathcal{O}(m \cdot n + k\cdot n \cdot t)$.

\begin{theorem} \label{Th:Algo 2}
The computational time and space requirement of Algorithm \ref{Algo: Algorithm2} will be of $\mathcal{O}( k\cdot n^{2} \cdot t)$  and $\mathcal{O}(k + n)$, respectively.
\end{theorem}

\paragraph{\textbf{Approach 2 (Spatial Partition+ Pruned Submodularity Graph +Incremental Greedy Approach)}} 
In this approach also we use the similar kind of pruning technique of clusters, and hence, Line $1$ to $11$ of Algorithm \ref{Algo: Algorithm2} is repeated in Algorithm \ref{Algo: Algorithm3} as well. Once the pruning at the cluster level is done, we merge the billboard slots of all the remaining clusters and we use the Pruned Submodularity Graph-based technique to reduce the ground set from Line $3$ to $11$. Finally we use the incremental greedy approach to select Top-$k$ Influential Slots from Line $12$ to $15$.

%Although the goal of selecting the top-$k$ influential billboard slots is accomplished using Algorithm \ref{Algo: Algorithm2}, it still requires a lot of computing time. To address this issue, combine Algorithm \ref{Algo: Algorithm2} with a graph-based pruning strategy to remove less influential billboard slots before choosing the top-$k$ slots using incremental greedy as shown in Algorithm \ref{Algo: Algorithm3}. The proposed Algorithm \ref{Algo: Algorithm3} is three folded. In the first phase (lines $1$-$11$) spatial partition approach billboard slots from most influential clusters are stored in $\mathcal{\overline{BS}}$ while less influential clusters are removed from $\theta$-partition $P$. In the second phase, we perform the pruning from Line No. $12$ to $20$ and in each iteration, the \texttt{while} loop, randomly pick $h \cdot \log n$ many slots and inserted into $\mathbb{U}$. Now, From the remaining slots in $\mathcal{\overline{BS}} \setminus \mathbb{U}$, the value of $w_{\mathbb{U}d}$ computed as mentioned in definition \ref{Def:Divergence}. Now from this list, $(1-\frac{1}{\sqrt{l}})$ fraction of the slots with the smallest value of $w_{\mathbb{U}d}$ are pruned until $|\mathcal{\overline{BS}}| > \ h \cdot \log n$ satisfies. Lastly, in the third phase, the incremental greedy approach selects the slots from the final set $\mathcal{\overline{BS}}$ after pruning. As illustrated in Figure \ref{Fig:2_Time}, it accomplishes the needed objective and minimizes the computational time requirement.

\begin{algorithm}[h]
\SetAlgoLined
\KwData{ A $\theta$-partition $P$ of $\mathbb{BS}$}
\KwResult{ $\mathbb{X} \subseteq \mathcal{BS}$, $\ \text{such that} \ |\mathbb{X}|= k$ }
Execute Line $1$ to $11$ from Algorithm \ref{Algo: Algorithm2}\;

% Initialize $\mathbb{X} \leftarrow \emptyset$, $m$ $\leftarrow$ $|P|$, $\mathcal{\overline{S}} \leftarrow \emptyset$, $\mathcal{U} \longleftarrow \emptyset$, $|\mathbb{BS}|=n$, $Avg[ ]$\;
% 
%  \For{$\text{i} \leftarrow  1\; to\; m $}{
%   $Avg[i]$ = Average $I(\mathcal{K}_i)$\;
% }
% Calculate the Average of values in $Avg$ to compute $\gamma$\;
%\For{$\text{i} \leftarrow  1\; to\; m $}{
% \If{$Avg[i] < \gamma$}{
% $\mathbb{P} \longleftarrow \mathbb{P} \setminus \{\mathcal{K}_i\}$\;
% }
% }
Store billboard slots from remaining clusters to $\mathcal{\overline{BS}}$
$\text{ and construct \textsl{PSG}}$\;
\While{$|\mathcal{\overline{BS}}| > \ h \cdot \log n$}{
 $\text{Pick }h \cdot \log n \text{ slots randomly from }\mathcal{\overline{BS}}  \text{ and store in }\mathbb{U}$\; 
 $\mathcal{\overline{BS}} \longleftarrow \mathcal{\overline{BS}} \setminus \mathbb{U}$, $\mathcal{\overline{S}} \longleftarrow \mathcal{\overline{S}} \cup \mathbb{U} $\;
 \For{$\text{All }d \in \mathcal{\overline{BS}} $}{
 $\mathcal{W}_{\mathbb{U}d} \longleftarrow \ \underset{u \in \mathbb{U}}{min} \ [I(d|u)-I(u| \mathcal{\overline{BS}} \setminus \{u\})]$\;
 }
 $\text{Take out }(1-\frac{1}{\sqrt{l}}) \cdot |\mathcal{\overline{BS}}| \text{ elements from }\mathcal{\overline{BS}} \text{ having smallest }\mathcal{W}_{\mathbb{U}d}$\;
 }
 $\mathcal{\overline{BS}} \longleftarrow \mathcal{\overline{BS}} \  \bigcup \ \mathcal{\overline{S}} $\;
 \While{$|\mathbb{X}| \neq k$}{
 $d^{*} \longleftarrow \underset{d \in \mathcal{\overline{BS}} \setminus \mathbb{X}}{argmax} \ I(\mathbb{X} \cup \{d\}) - I(\mathbb{X})$\;
 $\mathbb{X} \longleftarrow \mathbb{X} \cup \{d^{*}\}$\;
 }
Return $\mathbb{X}$\;
 \caption{ Spatial Partition + \textsl{PSG} + Incremental Greedy Approach for \textsl{IBSP} Problem}
 \label{Algo: Algorithm3}
\end{algorithm}	

\par Now, we analyze the time and space complexity of Algorithm \ref{Algo: Algorithm3}. As mentioned in the analysis of Algorithm \ref{Algo: Algorithm2}, the pruning step at the cluster level requires $\mathcal{O}(m \cdot n \cdot t)$. In the worst case, even after pruning at the cluster level, the number of  billboard slots remaining will be of $\mathcal{O}(n)$. As mentioned in the analysis of Algorithm \ref{Algo: Algorithm1_Simple Greedy} reducing  a ground set of $n$ elements will take $\mathcal{O}(n \cdot \log^{2}n)$ time. From the reduced ground set, selecting Top-$k$ Influential Billboard Slots will take $\mathcal{O}(k \cdot t \cdot \log^{2}n)$ time. Hence, the total time requirement will be $\mathcal{O}(m \cdot n \cdot t+n \cdot \log^{2}n + k \cdot t \cdot \log^{2}n)$. As mentioned in the analysis of Algorithm \ref{Algo: Algorithm1_Simple Greedy} reduction of the ground set by Pruned Submodularity Graph will take $\mathcal{O}(n^{2})$ space. Also, it can be observed that the pruning step at the cluster level will take $\mathcal{O}(m)$ space. However, as the value of $m$ can be of $\mathcal{O}(n)$ in the worst case, hence the space complexity of Algorithm \ref{Algo: Algorithm3} will be of $\mathcal{O}(n^{2})$. Hence, Theorem 6 holds.

\begin{theorem}
Time and space requirement of Algorithm \ref{Algo: Algorithm3} will be of $\mathcal{O}(m \cdot n \cdot t+n \cdot \log^{2}n + k \cdot t \cdot \log^{2}n)$ and $\mathcal{O}(n^{2})$, respectively.

\end{theorem}

\par{\textbf{An Illustrative Example}}  Given Trajectory database $\mathbb{D}$, Billboard Database $\mathbb{B}$, Influence function $I()$, Billboard slots $\mathcal{BS}$ contains $\{ bs_1,bs_2, \ldots, bs_{30}\}$ and initially we take each billboard slots in $\mathcal{BS}$ as a cluster. Now, for a given $\theta$ value $(0 \leq \theta \leq 1)$ we iteratively computes overlap ratio between any two clusters and if the overlap ratio between two clusters greater than $\theta$, then merges two clusters as described in Algorithm \ref{Algo: Spatial_Partition}. In this way we generate a partition $P$ of $\mathcal{BS}$ where $P$ =  $\{{\mathcal{K}_1, \mathcal{K}_2, \mathcal{K}_3,\mathcal{K}_4,\mathcal{K}_5}\}$ and $\mathcal{K}_1$ = $\{{bs_4, bs_6, bs_{13}, bs_{19}, bs_{15}, bs_{11}, bs_{21}}\}$, $\mathcal{K}_2$ = $\{{bs_{17}, bs_{20}, bs_{9}, bs_{12},bs_{27}}\}$, $\mathcal{K}_3$ = $\{{bs_{3},bs_{7}, bs_{10}, bs_{16}, bs_{23}}\}$, $\mathcal{K}_4$ = $\{{bs_{1}, bs_{5}, bs_{8}, bs_{18}, bs_{25}, bs_{29}, bs_{24}}\}$, $\mathcal{K}_5$ = $\{{bs_{2}, bs_{26},bs_{14}, bs_{22}, bs_{28}, bs_{30}}\}$. For better understanding, we assume that cost of every billboard slot in $\mathcal{BS}$ is 1. At first, we calculate the influence of each clusters using Equation \ref{Eq:1} and then compute the average influence of each clusters. Let for clusters $\mathcal{K}_1$,  $\mathcal{K}_2$, $\mathcal{K}_3$, $\mathcal{K}_4$ and $\mathcal{K}_5$ average influence are 3, 2.8, 2.2, 2.5, 3.1 respectively. Now, we can calculate $\gamma$ from average influence of each clusters i.e. $\gamma$ = $\{{(3+2.8+2.2+2.5+3.1)/5}\}$ = 2.72. Now we want to prune all the clusters whose average influence value is less than $\gamma$ and for this reason cluster $\mathcal{\mathcal{K}}_3$ and $\mathcal{\mathcal{K}}_4$ are removed from $P$. After that, the remaining clusters in $P$ are merged and stored in $\mathcal{\overline{BS}}$ and applied pruned submodularity graph for further pruning. Let, after pruning we get billboard slots, $\mathcal{\overline{BS}}$ = $\{{bs_4, bs_{11}, bs_{14}, bs_{17}, bs_{19,} bs_{22}, bs_{26}, bs_{30}}\}$. For budget k = 5, we want to find maximum influential billboard slots from $\mathcal{\overline{BS}}$ and for this, we have applied incremental greedy approach and finally top influential billboard slots denoted as $\mathbb{X}$ =$\{{bs_4, bs_{14}, bs_{19,} bs_{22}, bs_{30}}\}$ are selected.

Now, from Lemma \ref{Lemma:2}, and Theorem \ref{Zohu_Result}, we derive the approximation ratio achieved by Algorithm \ref{Algo: Algorithm3} in Theorm \ref{Th:Theta1}.

\begin{theorem} \label{Th:Theta1}
Given a $\theta$-partition, $P =\{\mathcal{K}_{1}, \mathcal{K}_{2},\ldots,\mathcal{K}_{m}\}$ Algorithm \ref{Algo: Algorithm3} achieve an approximation ratio of $\frac{1}{2}^{\lceil \log_{(1+\frac{1}{\theta})}m\rceil} (1-\frac{1}{e})\cdot (I(\mathcal{Z}^{OPT})-2k\epsilon)$ to the Influential Billboard Slot Selection problem.
\end{theorem}

\begin{proof}
Let us assume $\mathcal{Z} = \{\mathcal{Z}_{1},\mathcal{Z}_{2},\ldots,\mathcal{Z}_{q}\}$, where $\mathcal{Z}_{i} = \mathcal{Z} \cap \mathcal{C}_{i}$ and $i \leq q \leq m$. When $\theta$ = 0, we have $I(\mathcal{Z})= \sum_{i=1}^q I(\mathcal{Z}_{i})$ and from the Lemma \ref{Lemma:2}, we have assume that $\sum_{i=1}^q I(\mathcal{Z}_{i})$ is the maximum influence. Now, from Theorem \ref{Zohu_Result}, we can say that pruned submodularity graph return $S^{'}$, and $I(S^{'}) \geq (1-\frac{1}{e})\cdot (I(\mathcal{Z}^{OPT})-2k\epsilon)$, where $\mathcal{Z}^{OPT}$ is the optimal solution of budget $k$. So, we can write $I(\mathcal{Z})= \sum_{i=1}^q I(\mathcal{Z}_{i}) \geq I(S^{'})$. \\
When $\theta > 0$, we have $I(\mathcal{Z})\leq \sum_{i=1}^q I(\mathcal{Z}_{i})$. Now, by considering iterative process, at the $0^{th}$ iteration we denote $\mathcal{Z}^{0}$ as set of billboard slot clusters in which $\mathcal{Z}^{0}_{x}$ belongs to $\mathcal{Z}_{x}$. Let, at the $j^{th}$ iteration we arbitrarily partition the clusters in $\mathcal{Z}^{j-1}$ and by combining them new disjoint clusters for $\mathcal{Z}^{j}$ are generated. So, every cluster $\mathcal{Z}^{j}_{x}$ in $\mathcal{Z}^{j}$ contains at most $(1+\frac{1}{\theta})$ clusters from $\mathcal{Z}^{j-1}$. Now, from the Lemma \ref{Lemma:2}, we can write, $ I(\mathcal{Z}^{j}_{x}) \geq \frac{1}{2} \sum_{\mathcal{Z}^{j-1}_{i} \in \mathcal{Z}^{j}_{x}} I(\mathcal{Z}^{j-1}_{i})$. Considering this iteration process runs $t$ times and as $q<m$, $t$ does not exceed ${\lceil \log_{(1+\frac{1}{\theta})}m\rceil}$. So, we can write $I(\mathcal{Z}^{t}) \geq \frac{1}{2}^{\lceil \log_{(1+\frac{1}{\theta})}m\rceil} \sum_{i=1}^q I(\mathcal{Z}_{i})$. 
Now, if we assume that $\mathcal{Z}^{t}$ have only one billboard slot set $\mathcal{Z}^{t}_{1}$, then $I(\mathcal{Z}^{t}) = \mathcal{Z}^{t}_{1}$, and $\mathcal{Z}^{t}_{1} = \mathcal{Z}$. Now, as $\sum_{i=1}^q I(\mathcal{Z}_{i}) \geq I(S^{'}) \geq (1-\frac{1}{e})\cdot (I(\mathcal{Z}^{OPT})-2k\epsilon)$, we can write,

\begin{align}
I(\mathcal{Z}^{t}) &\geq \frac{1}{2}^{\lceil \log_{(1+\frac{1}{\theta})}m\rceil} \sum_{i=1}^q I(\mathcal{Z}_{i})\nonumber\\
&\geq \frac{1}{2}^{\lceil \log_{(1+\frac{1}{\theta})}m\rceil} (1-\frac{1}{e})\cdot (I(\mathcal{Z}^{OPT})-2k\epsilon)
\end{align}
Hence, we can conclude that after $\theta$-partition, if directly pruned submodularity graph followed by incremental greedy applied in Algorithm \ref{Algo: Algorithm3}, then it achieves an approximation ratio of $\frac{1}{2}^{\lceil \log_{(1+\frac{1}{\theta})}m\rceil} (1-\frac{1}{e})\cdot (I(\mathcal{Z}^{OPT})-2k\epsilon)$.
\end{proof}

\section{Experimental Evaluation}\label{Sec:EE}
This section will discuss the experimental evaluation of our proposed solution approaches. Now, we start by describing the datasets used in our experiments.

\subsection{Datasets}
Previous studies have already used the datasets we used in the present study \cite{zhang2020towards}. We have used two real-world trajectory datasets for our experiments. The TLC trip record dataset \footnote{\url{http://www.nyc.gov/html/tlc/html/about/trip_record_data.shtml.}} for NYC is the first, while the Foursquare check-in dataset \footnote{\url{https://sites.google.com/site/yangdingqi/home.}} is the second. Both of these datasets include records of green taxi trips made between January 2013 and September 2016 at various locations, including malls, beaches,
airports, etc. We construct six alternative datasets from these trajectory datasets by separating the tuples that correspond to the following six locations: `Beach', `Mall', `Bank', `Bus Station', `Train Station', and `Airport'. Table \ref{table:Dataset} presents a summary of this dataset’s basic information.

\begin{table} [h!]
    \centering
    \caption{Description of the Datasets}
    \begin{tabular}{ || c c c c || }
    \hline
    Location Type & \# Billboard & \# Trajectory  & \# Billboard Slots\\ [0.5 ex]
    \hline \hline
    Beach & 76 & 575 & 21888 \\
    Mall & 86 & 1186 & 24768 \\
    Bank & 671 & 2232 &  193248 \\
    Bus Station & 1056& 4473 &  304128 \\
    Train Station & 288 & 6407 & 82944 \\
    Airport & 313 & 2852 & 90144 \\ [1ex]
    \hline
    \end{tabular}
    \label {table:Dataset}
    \end{table}

\subsection{Experimental Set Up}
\textbf{Billboard} data is crawled from one of the largest advertising company, LAMAR\footnote{\url{http://www.lamar.com/InventoryBrowser.}}, and the dataset contains attributes like billboard id, panel size, location, latitude, longitude, etc. Now, we choose the panel size of the largest billboard and compute the influence probability of a billboard, i.e., the ratio between the size of that billboard and the size of the largest billboard.

\textbf{Performance.} The runtime and influence value of the picked billboard slots is used to compute the effectiveness of all methods. The results of the average performance of each experiment are reported after being repeated three times.

\textbf{Setup.} In our experiments, all the codes are implemented in Python, and experiments are conducted on a Ubuntu-operated system with Intel Xeon(R) 3.50GHz CPU and 64 GB memory.

\subsection{Algorithms Compared}
The performance of the proposed solution approach is compared against the following approaches:
\begin{itemize}

\item \textbf{RANDOM}: This is the most straightforward approach to choosing the most influential billboard slots. Here, for a given value of $k$, we randomly select any $k$ billboard slots, and its running time will be $\mathcal{O}(k)$.

\item \textbf{Top-$k$}: In this approach, we first compute the individual influence for each billboard slot and then sort them in descending order. After that, from the sorted list, we pick top-$k$ slots from the list based on individual influence. Now, if there are $t$ many tuples in the trajectory dataset, computing a slot's influence will take $\mathcal{O}(t)$ time. As there is $n$ number of slots in the billboard dataset total time required will be $\mathcal{O}(n \cdot t)$. Sorting the list of slots will take $\mathcal{O}(n \cdot \log n)$ time. Finally, choosing $k$ number of billboard slots will take $\mathcal{O}(k)$ time. Hence, the total time required will be $\mathcal{O}(n \cdot t + n \cdot \log n + k)$. As $k << n$, we can write it as $\mathcal{O}(n \cdot t + n \cdot \log n)$.

\item \textbf{Maximum Coverage}: In this approach, we calculate the coverage of each billboard slot. We define the term `coverage' so that a slot $x \in \mathbb{BS}$ covers a tuple $t$ in the trajectory database if slots timing is present in the trajectory tuples. After that, we sort the billboard slots based on the coverage value of individual slots and pick top-$k$ slots having the highest coverage value. As there is $n$ number of slots total time required will be $\mathcal{O}(n \cdot t)$. Sort the slots will take $\mathcal{O}(n \cdot \log n)$ time. Finally, choosing $k$ number of slots will take $\mathcal{O}(k)$ time. So, the total time required will be $\mathcal{O}(n \cdot t + n \cdot \log n + k)$. As $k << n$, needed total time $\mathcal{O}(n \cdot t + n \cdot \log n)$.

\item \textbf{PartSel}: PartSel is a partition-based method, and it follows the divide and conquers approach to combine partial solutions from the clusters. At first, we partition billboard slots into a set of clusters. Then from each cluster, using EnumSel \cite{zhang2020towards}, algorithms select partial solutions and finally combine partial solutions to generate global solutions, i.e., a Set of most influential billboard slots. But we observe that the PartSel algorithm is inefficient due to huge computational time requirements when the size of the billboard dataset is very large.

\item \textbf{\textsl{PSG} + Random}: This approach applies a pruned submodularity graph to prune less influential billboard slots and obtain a subset $S \in \mathbb{BS}$  after pruning. Next, we pick $k$ of them from $S$ randomly. The time required for pruned submodularity graph is $\mathcal{O}(n^{2} \cdot t + n\cdot \log^{2} \cdot n)$ as discussed in Algorithm \ref{Algo: Algorithm1_Simple Greedy}. From $S$ choosing $k$ number of slots will take $\mathcal{O}(k)$ time. Hence, we can say the total time required
$\mathcal{O}(n^{2} \cdot t + n\cdot \log^{2} \cdot n)$.
\end{itemize}

\subsection{Goals of the Experiments}

In this study, we want to perform the experiments to show the following:
\begin{itemize}
\item \textbf{$\theta$ Vs. Influence} to show the effect of varying $\theta$ on different algorithms.
\item \textbf{$\theta$ Vs. Time} to show the efficiency of varying $\theta$ on different algorithms.
\item \textbf{$\theta$ Vs. Number of Clusters} to show the impact of $\theta$ to generate clusters on different datasets.
\item \textbf{Budget(k) Vs. Influence} to show the impact of different k values to generate the influence for different algorithms on different datasets.
\item \textbf{Budget(k) Vs. Time} to show the efficiency of different k values to generate the influence for different algorithms on different datasets.
\item \textbf{Distance$(\lambda)$ Vs. Influence} to show the impact of different $\lambda$ values to generate the influence for different algorithms on different datasets.

\end{itemize} 
\subsection{Experimental Results and Discussions}

\paragraph{\textbf{Budget Vs. Influence}}
Figure \ref{Fig:1_Influence} shows the plots for the number of billboard slots vs. the influence of different locations. From these figures, we can observe that among the existing solution approaches, in most cases, the PartSel leads to the comparable influence obtained by the proposed solution approaches. For example, when the number of billboard slots in Bank location type is $25$, the acquired influence by the existing methods, i.e., PartSel, Top-$k$, Coverage, Pruned Submodularity Graph+Random and Random are $155.79$, $154.67$, $146.95$, $39.42$, $16.45$ respectively. So, among these methods, PartSel leads to the maximum influence. However, the influence due to the proposed solution approaches, namely Spatial Partition Approach, Pruned Submodularity Graph+Incremental Greedy, $\theta$-Partition+Incremental Greedy $158.19$, $157.78$ and $157.48$ respectively. Hence, the Spatial Partition Approach leads to the maximum influence among them. These observations are consistent even for other location types as well.

\begin{figure*}[!ht]
\centering
\begin{tabular}{cccc}
\includegraphics[scale=0.16]{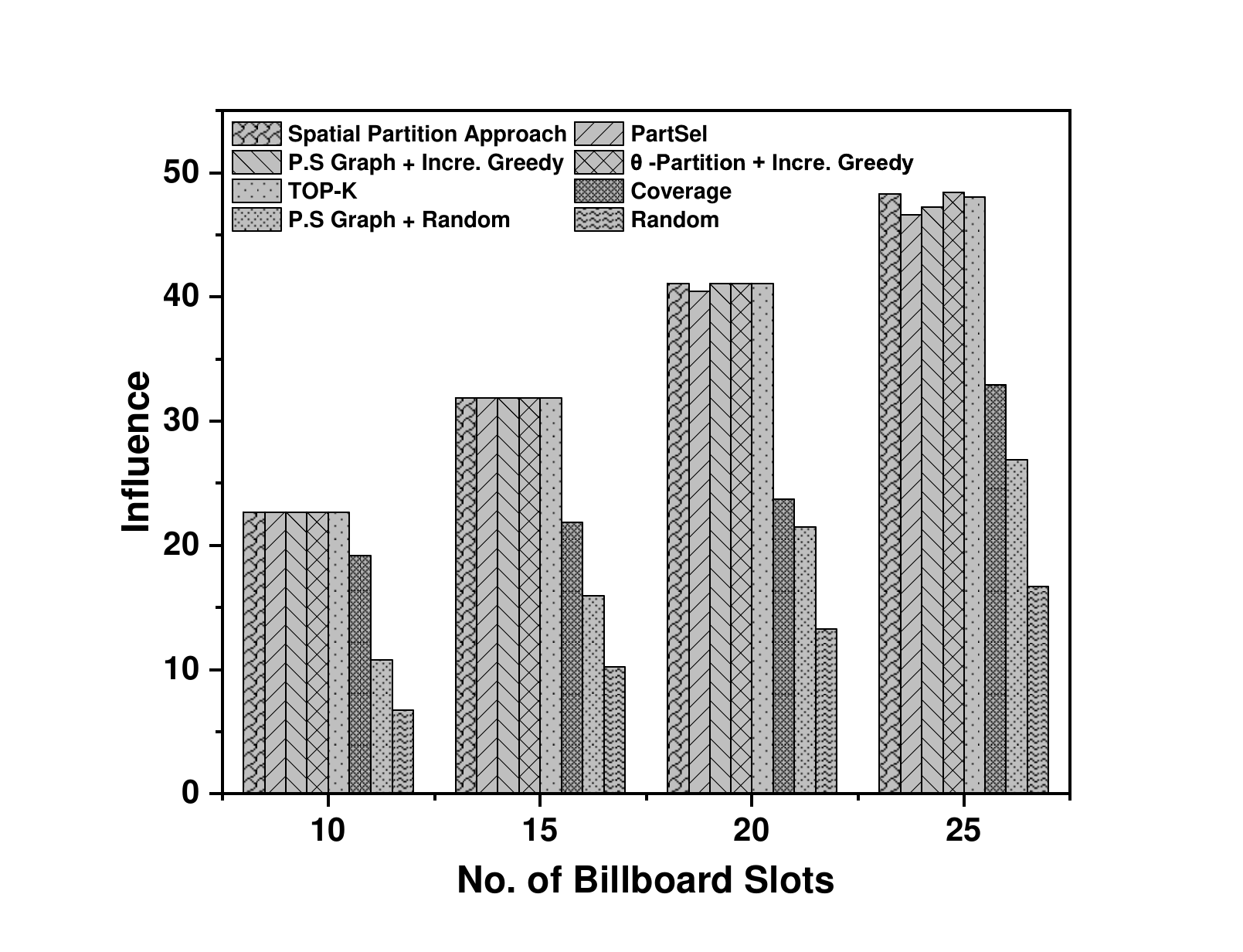} & \includegraphics[scale=0.16]{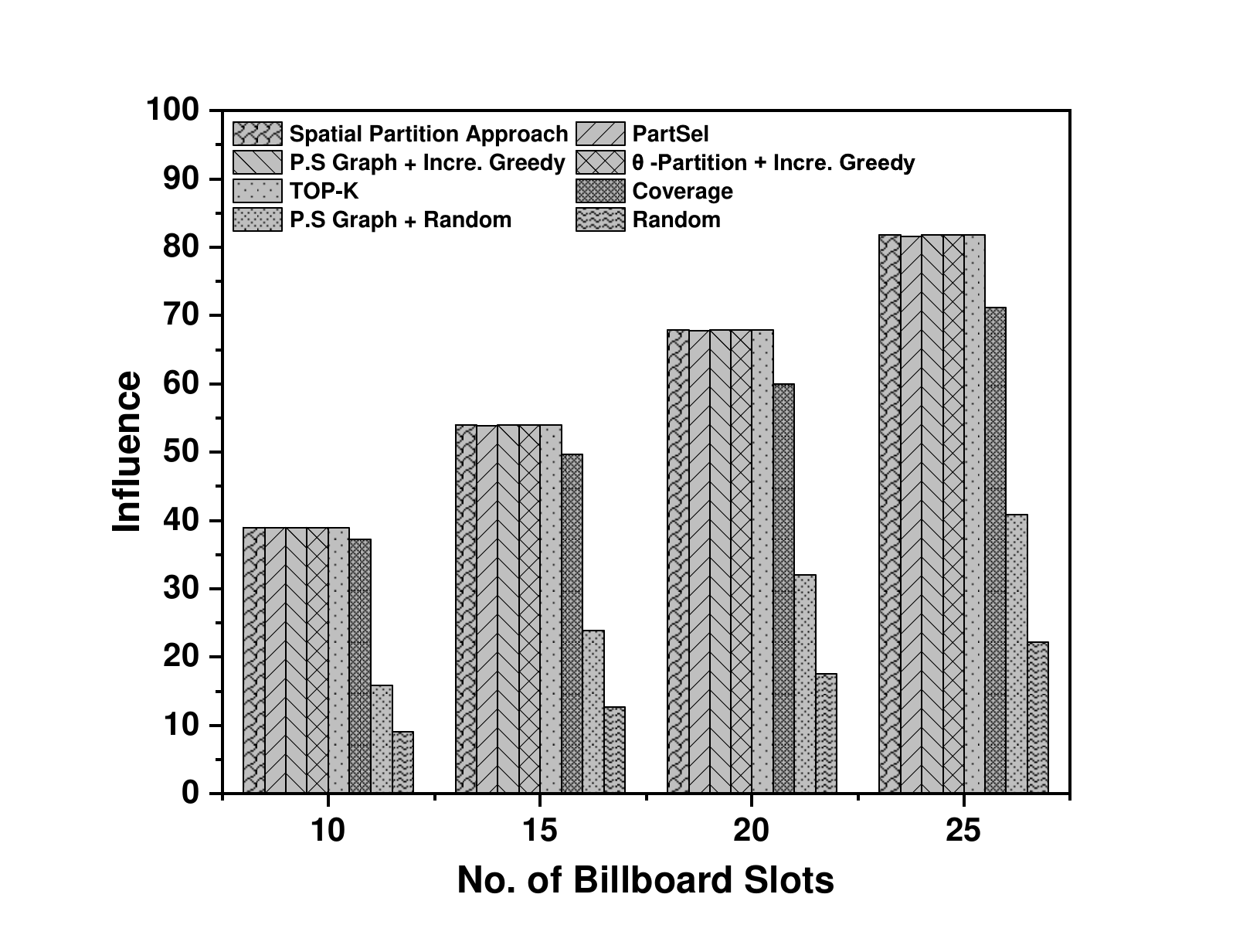} & \includegraphics[scale=0.16]{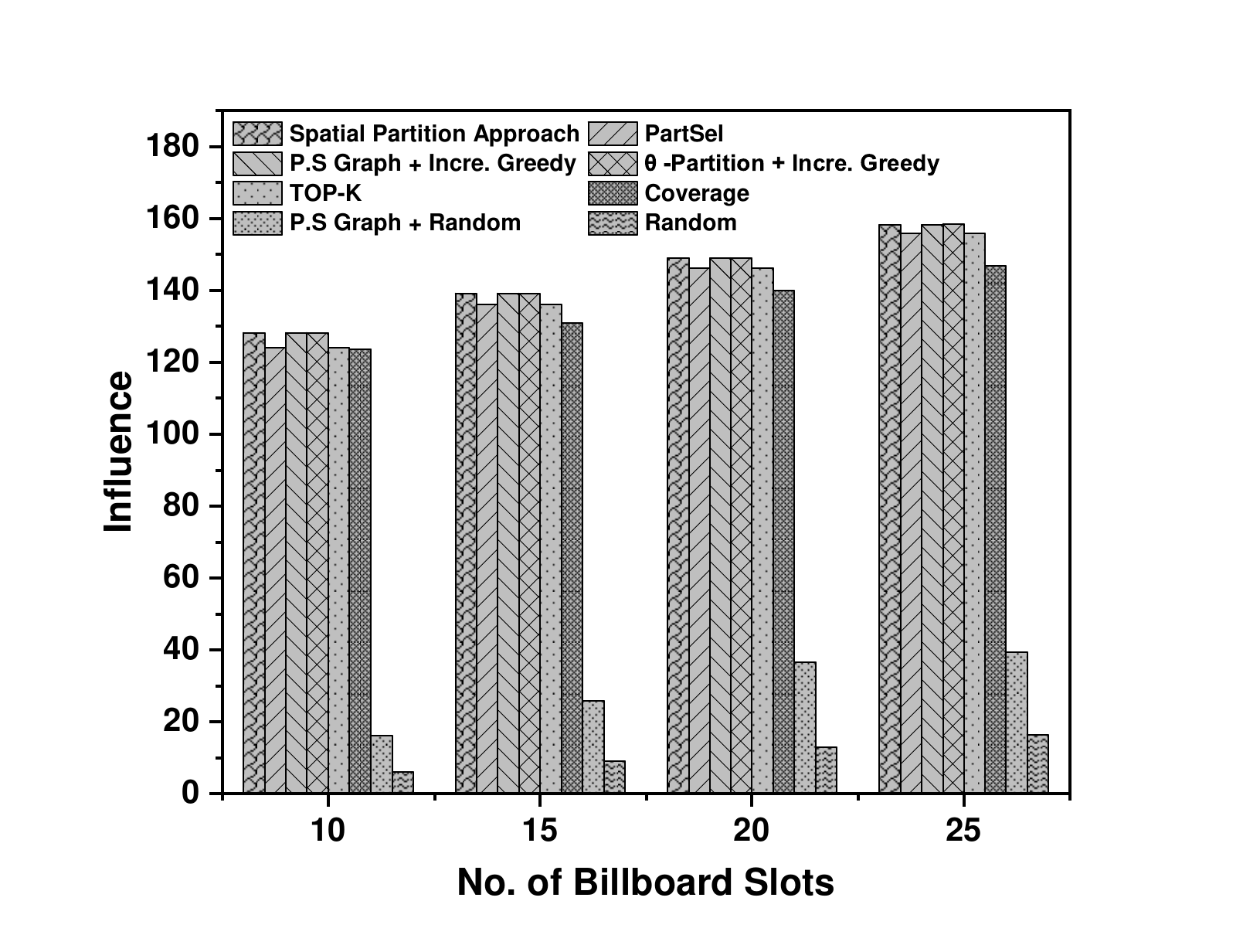} & \includegraphics[scale=0.16]{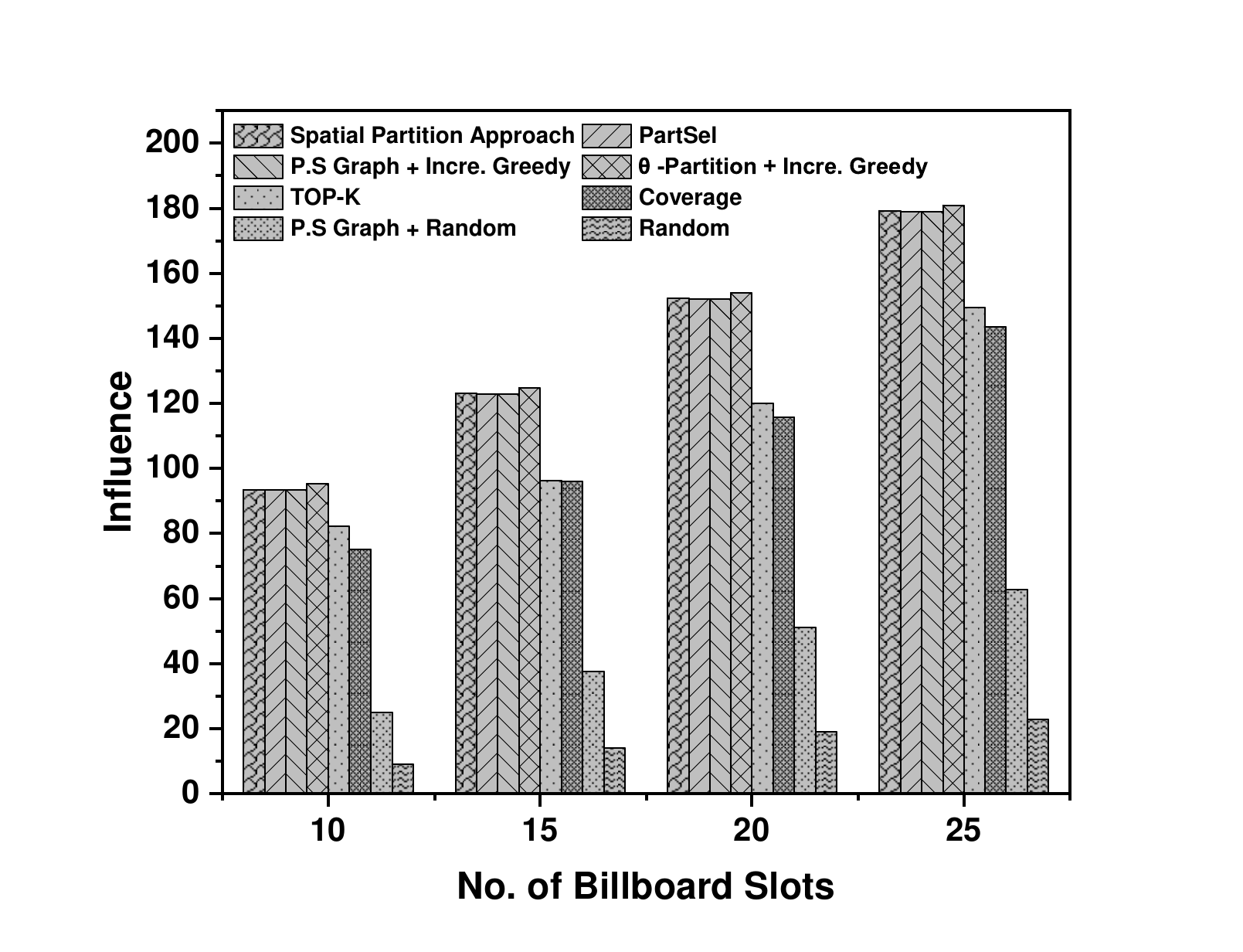} \\
(a) Location =`Beach' & (b) Location =`Mall' & (c) Location =`Bank' &(d) Location =`Park' \\
\includegraphics[scale=0.16]{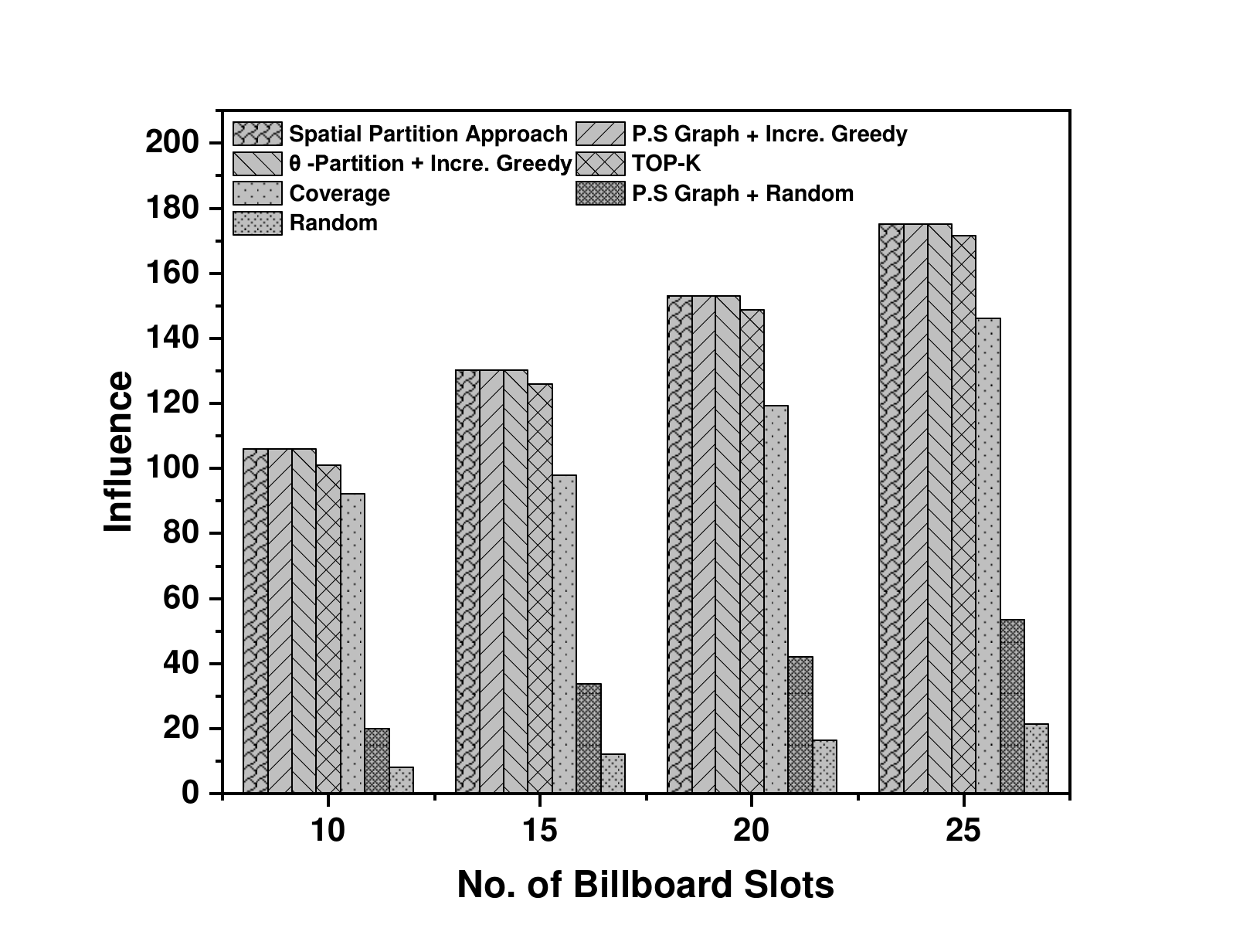} & \includegraphics[scale=0.16]{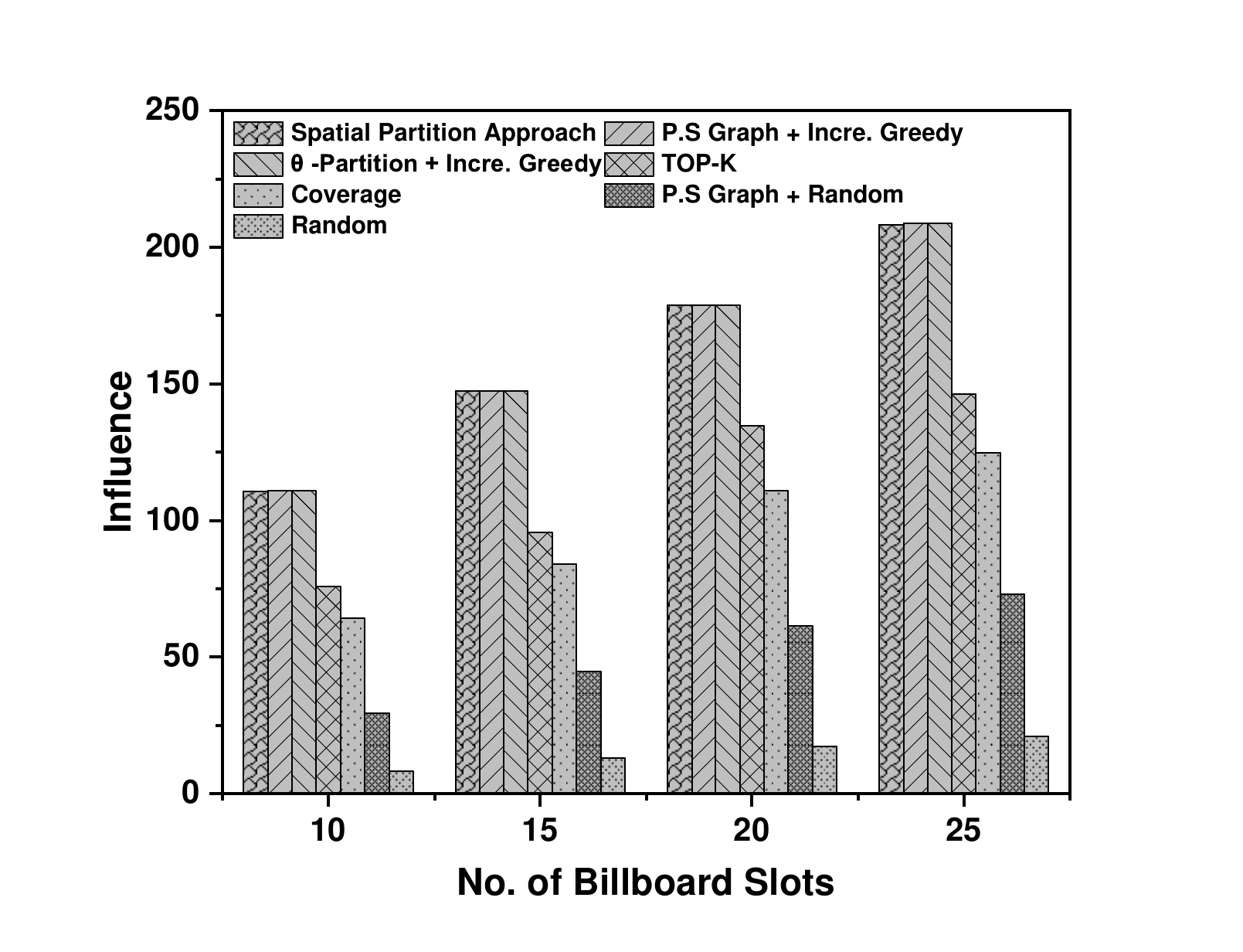} & \includegraphics[scale=0.16]{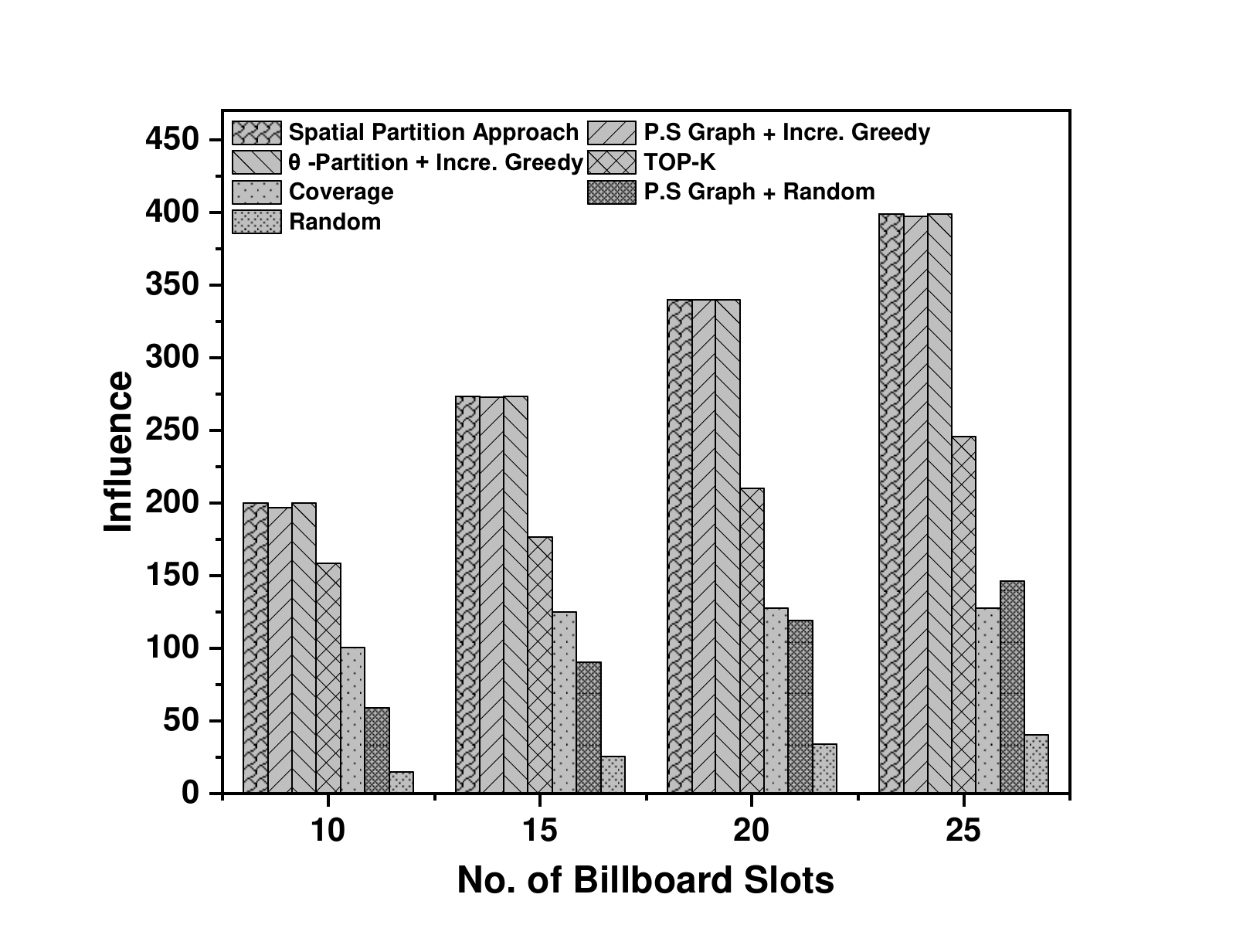} & \includegraphics[scale=0.16]{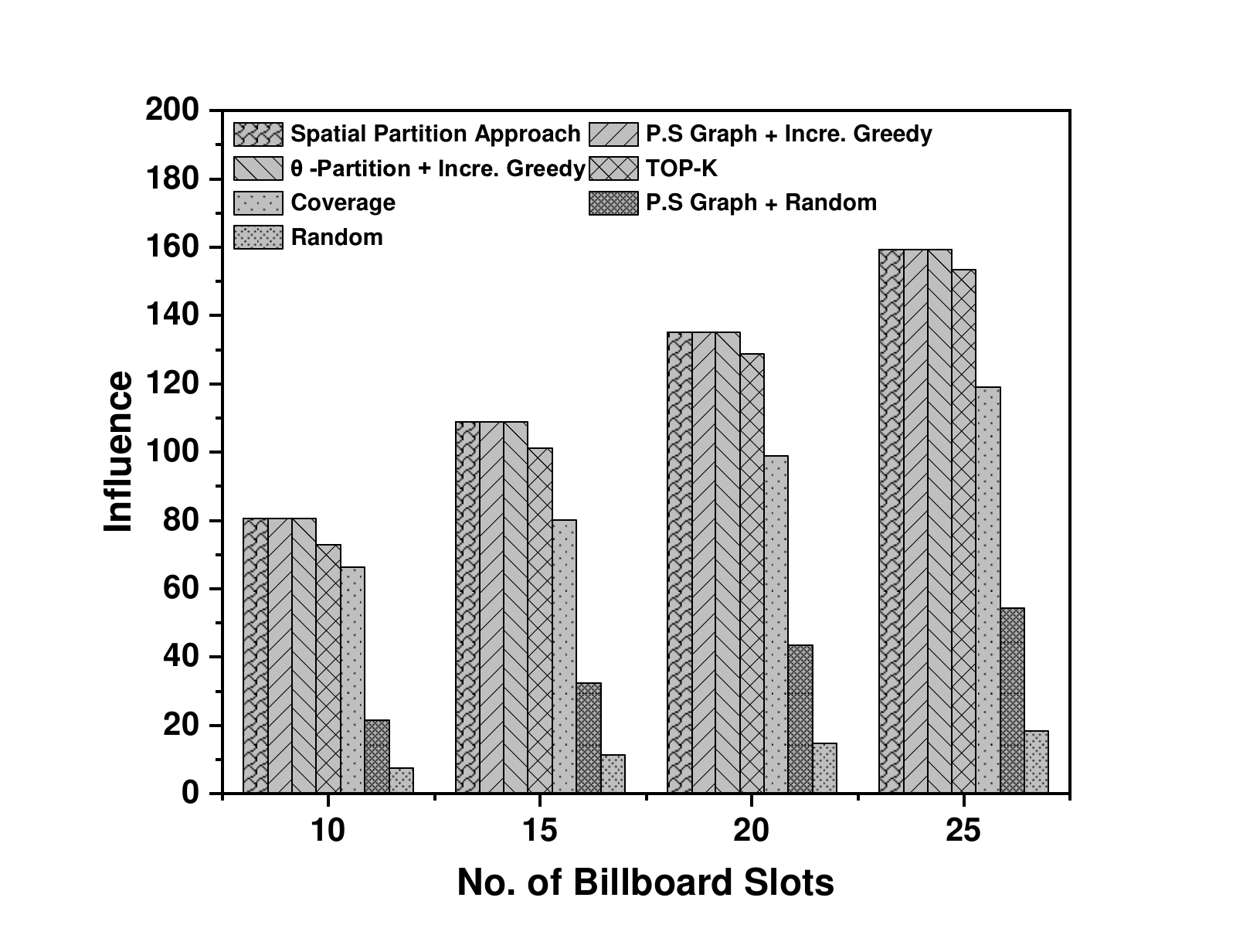} \\
(e) Location =`Airport' & (f) Location =`Bus Station' & (g) Location =`Train Station' &(h) Location =`Gym' \\

\end{tabular}
\caption{No. of Billboard Slot Vs. Influence plots for different location}
\label{Fig:1_Influence}
\end{figure*}
%%%%%%%%%%%%%%%%%%%%%%%%%%%%%%%%%%%%%%%%%%%%%%%%%%%%%%%%%

%%%%%%%%%%%%%%%%%%%%%%%%%%%%%%%%%%%%%%%%%%%%%%%%%%%%%%%%%
\paragraph{\textbf{Budget Vs. Time}}
Figure \ref{Fig:2_Time} shows the budget vs. computational time requirement plots for different datasets and all the algorithms. Figures \ref{Fig:1_Influence} and \ref{Fig:2_Time} show that the proposed solutions approach leads to almost equivalent influence compared to the PartSel Method. However, the computational time requirement by the Spatial Partitioning Approach is one-third of the PartSel method on  average. Here we highlight that the PartSel method is inefficient in handling large datasets. Hence, we do not consider this method for the last four datasets. As an example, when the location type is `Beach',  and the number of billboard slots is $21888$, the computational time requirement by the proposed spatial partitioning approach is $293$ Secs. Whereas the same for the PartSel method is $872$ Secs. So, it is approximately three times faster than the PartSel method. This occurs because PartSel follows a divide-and-conquer strategy to combine the solutions from the clusters. In PartSel, it recursively calls EnumSel \cite{zhang2020towards}, which takes huge computational time to select billboard slots from the clusters one by one to generate partial solutions. Also, among the three proposed solution approaches $\theta$-Partition+Incremental Greedy leads to the least computational time. Because this method, based on the cut-off influence, we prune out a significant number of clusters, and hence applying the incremental greedy approach on the top of it, takes much less time.

\begin{figure*}[!ht]
\centering
\begin{tabular}{cccc}
\includegraphics[scale=0.16]{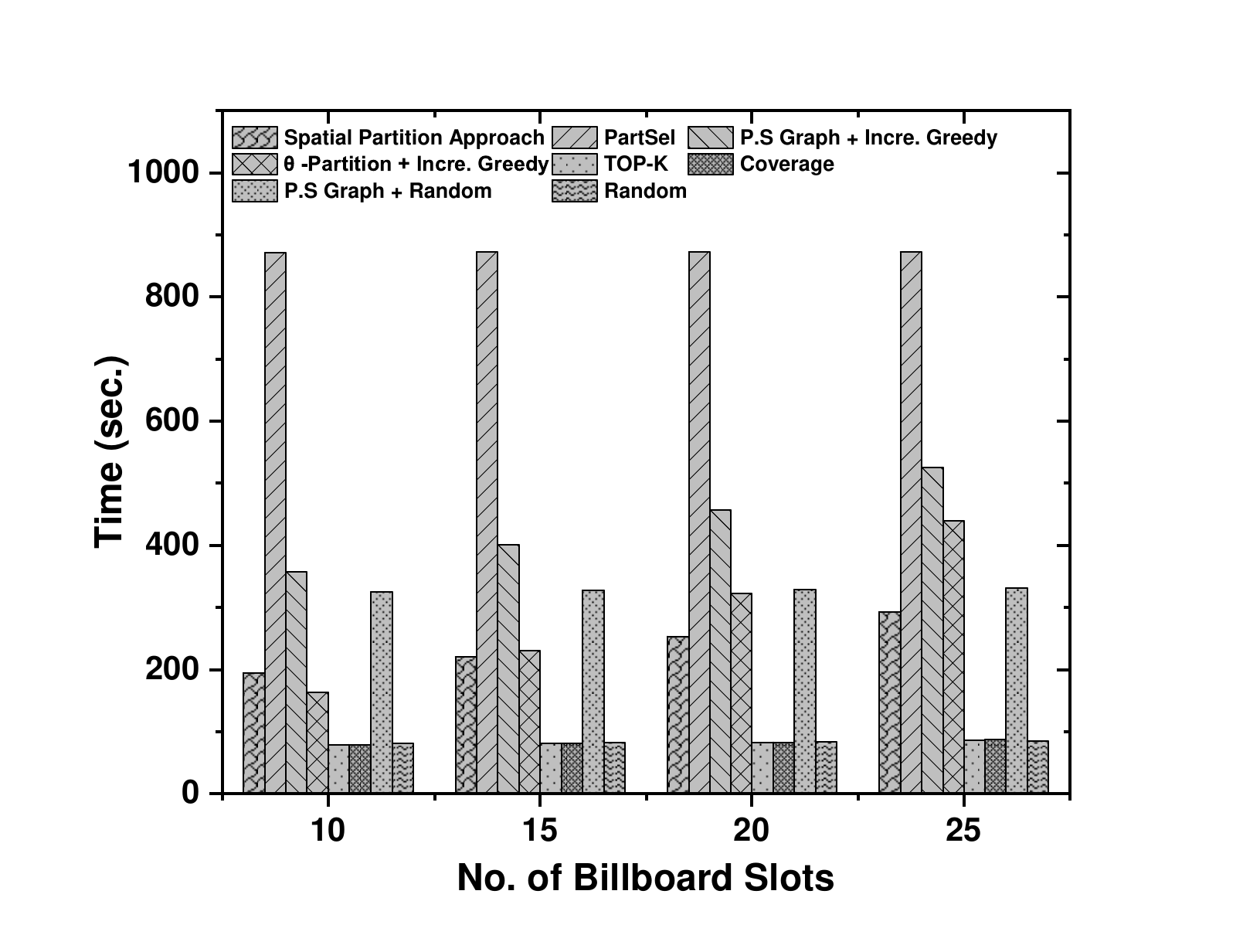} & \includegraphics[scale=0.16]{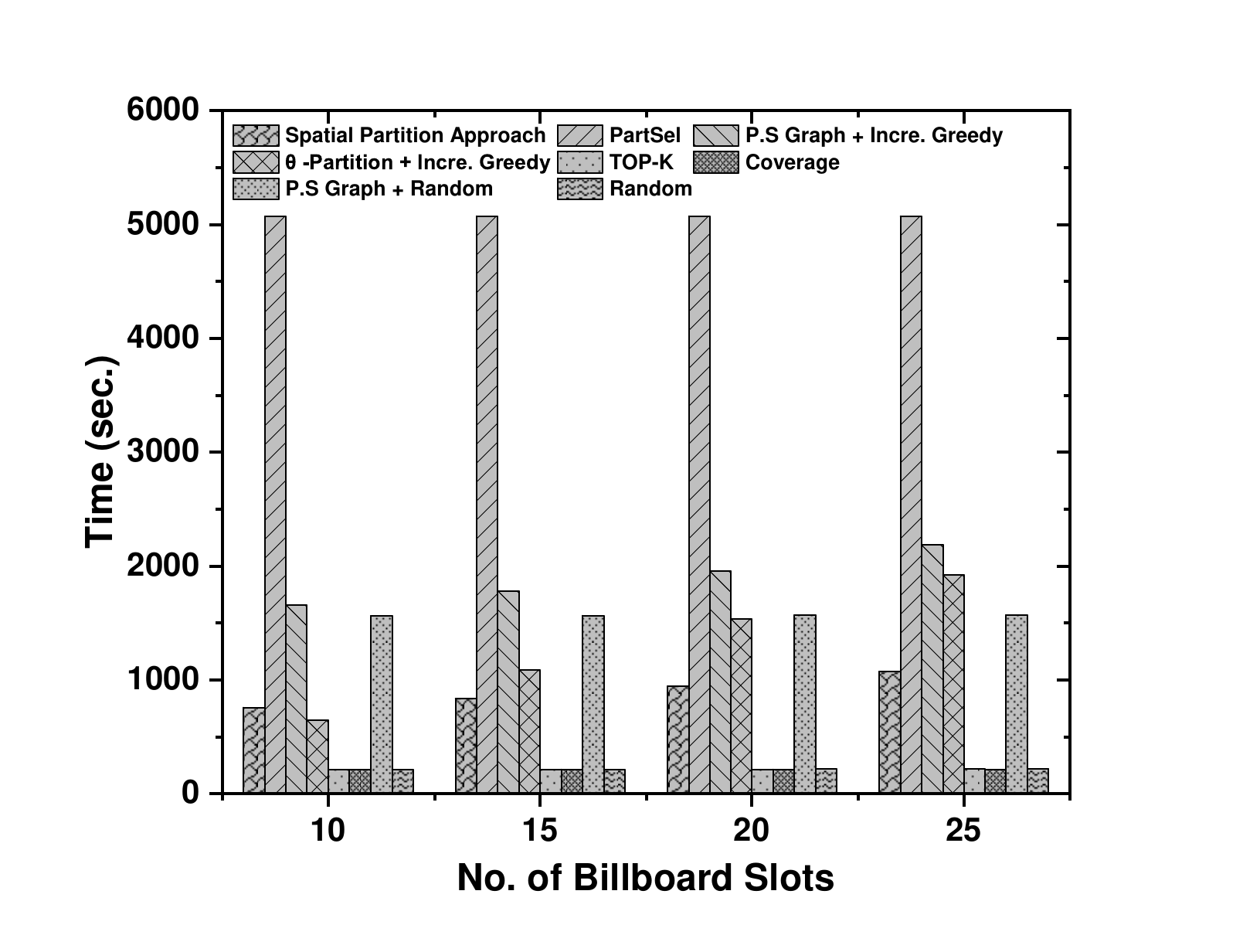} & \includegraphics[scale=0.16]{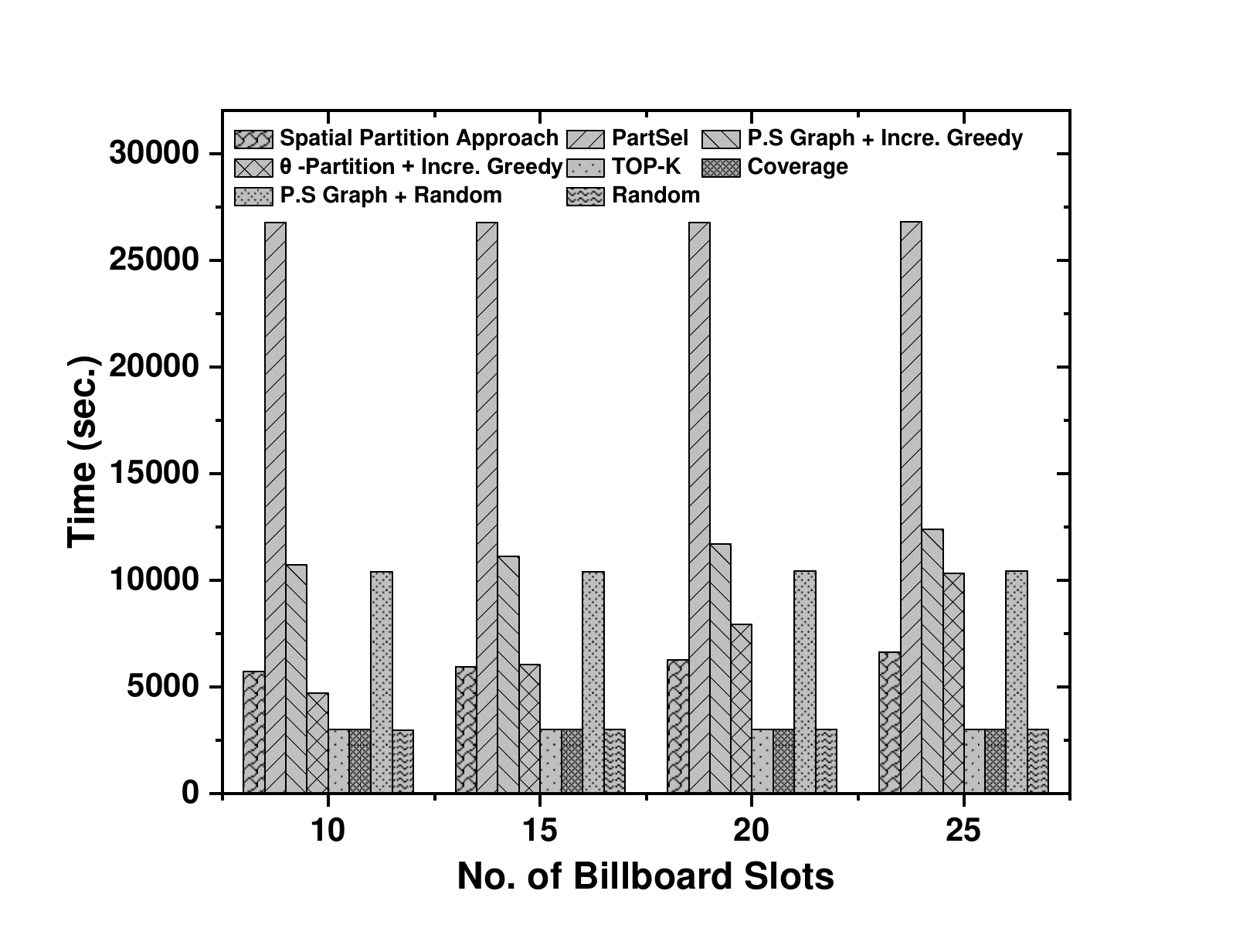} & \includegraphics[scale=0.16]{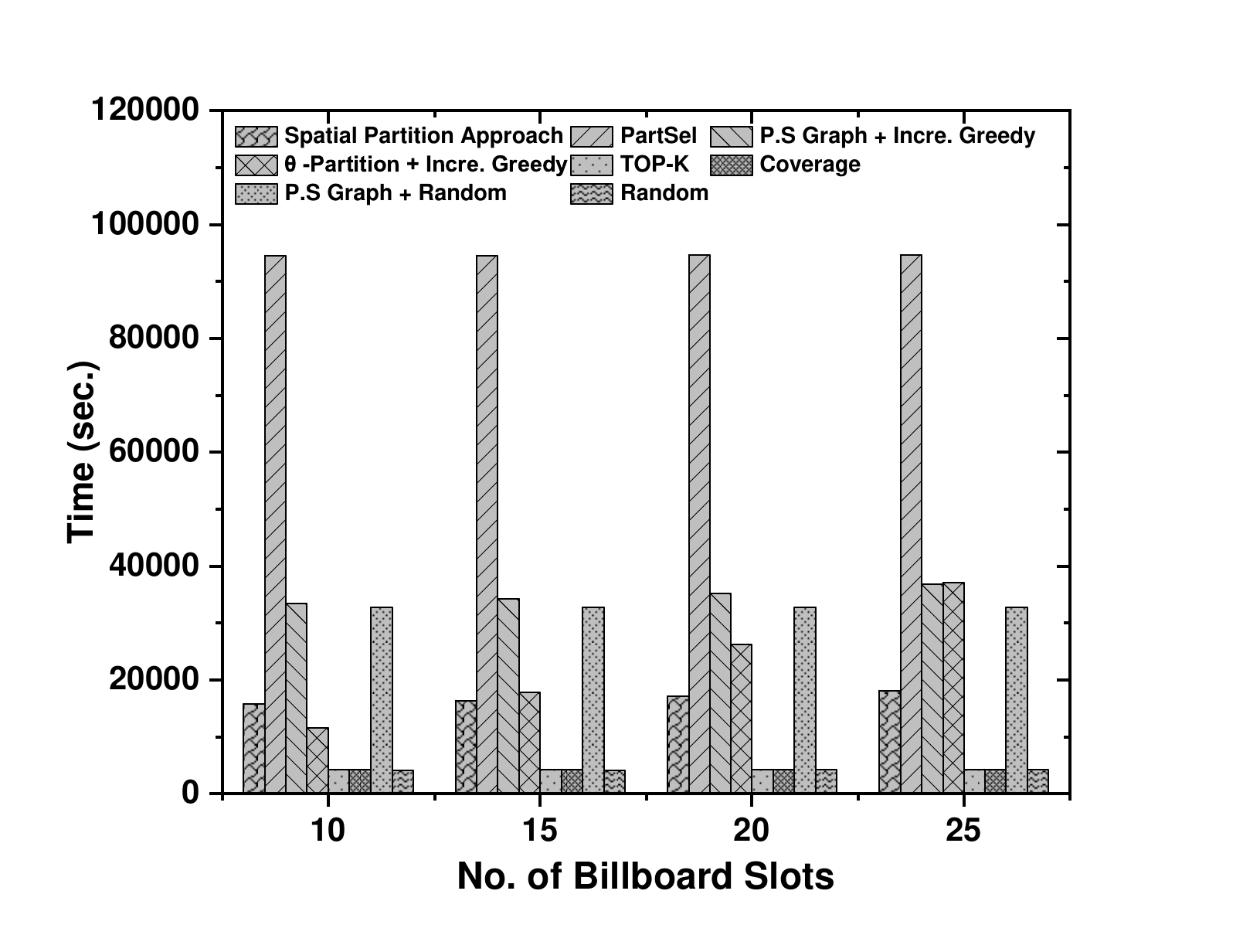} \\
(a) Location =`Beach' & (b) Location =`Mall' & (c) Location =`Bank' &(d) Location =`Park' \\
\includegraphics[scale=0.16]{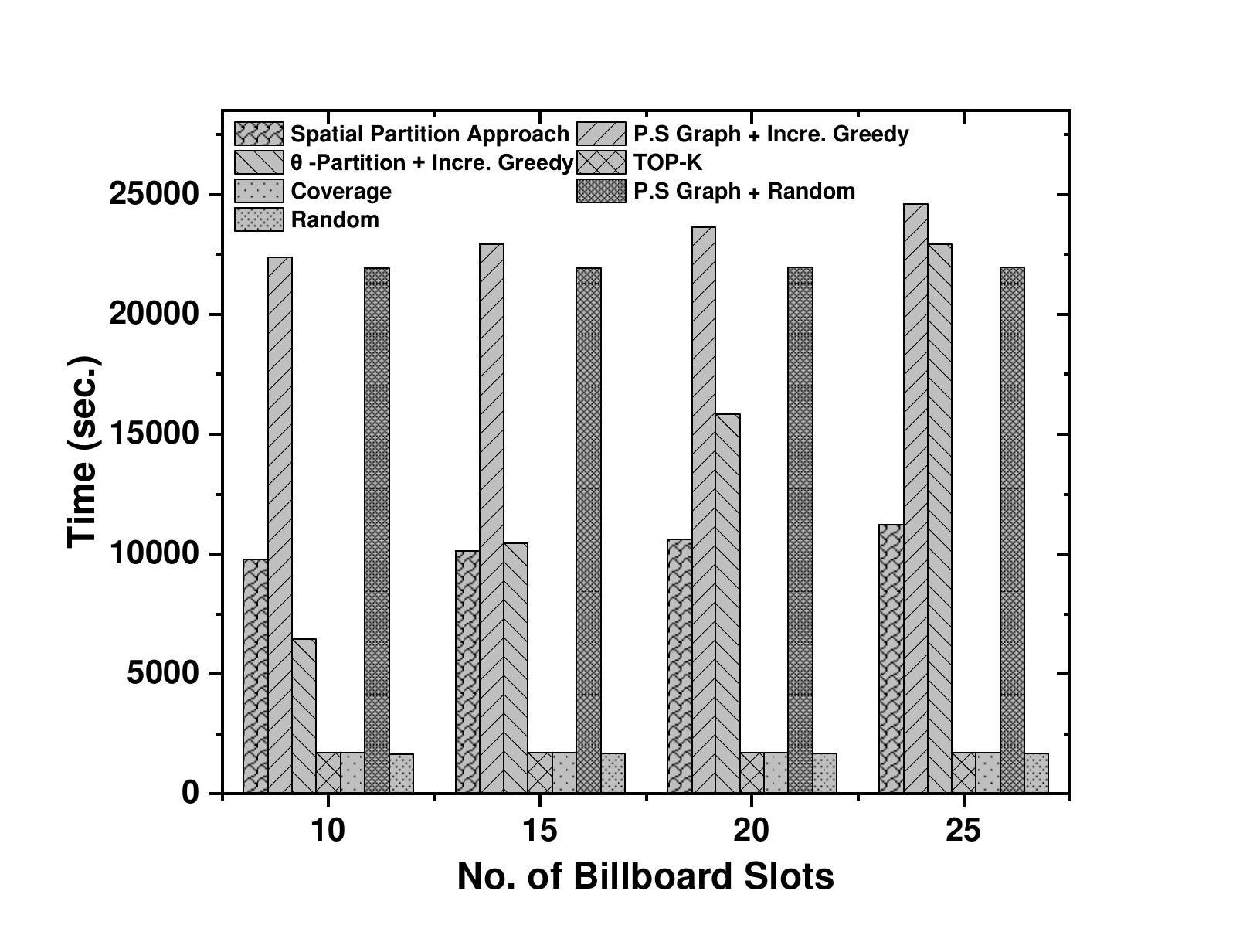} & \includegraphics[scale=0.16]{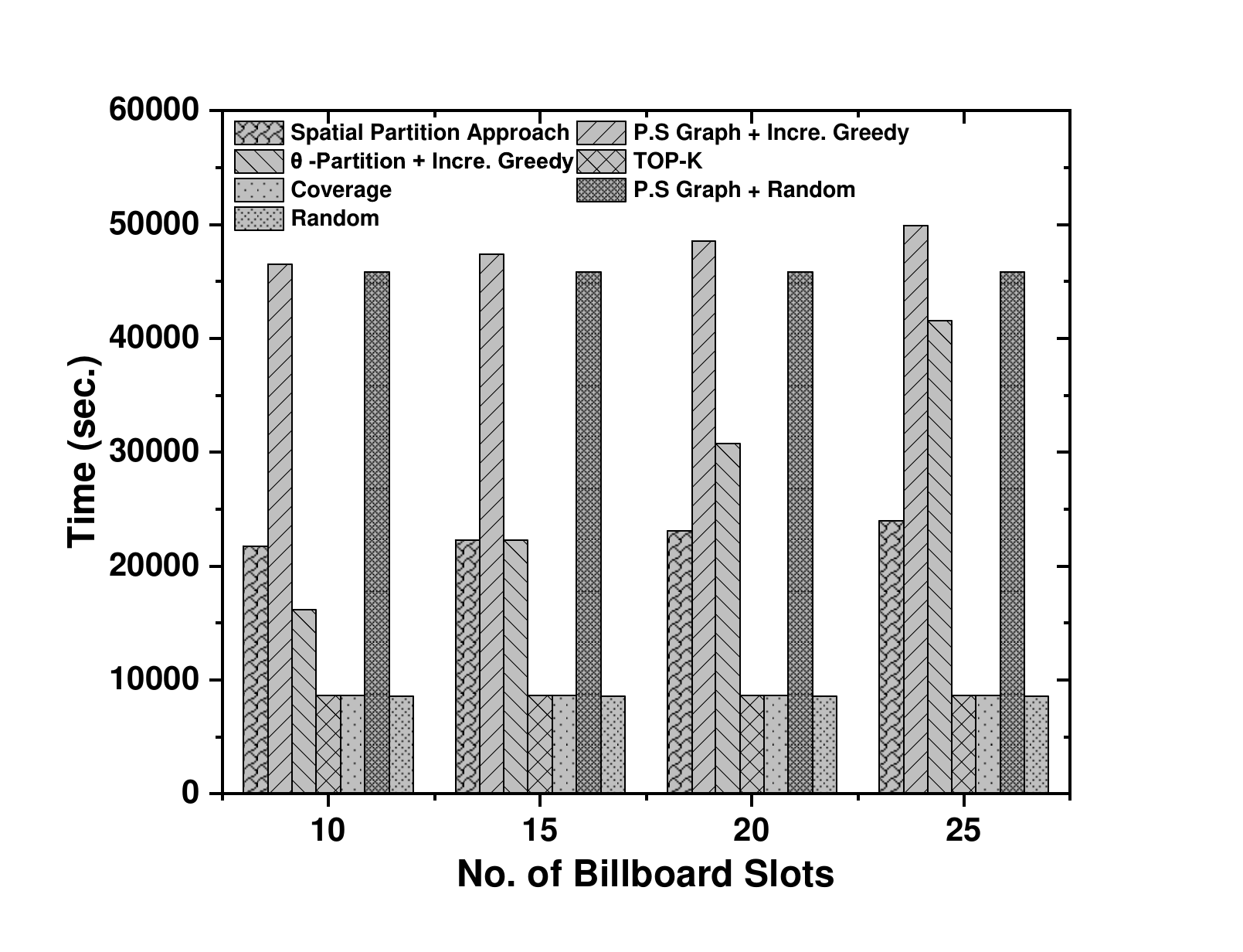} & \includegraphics[scale=0.16]{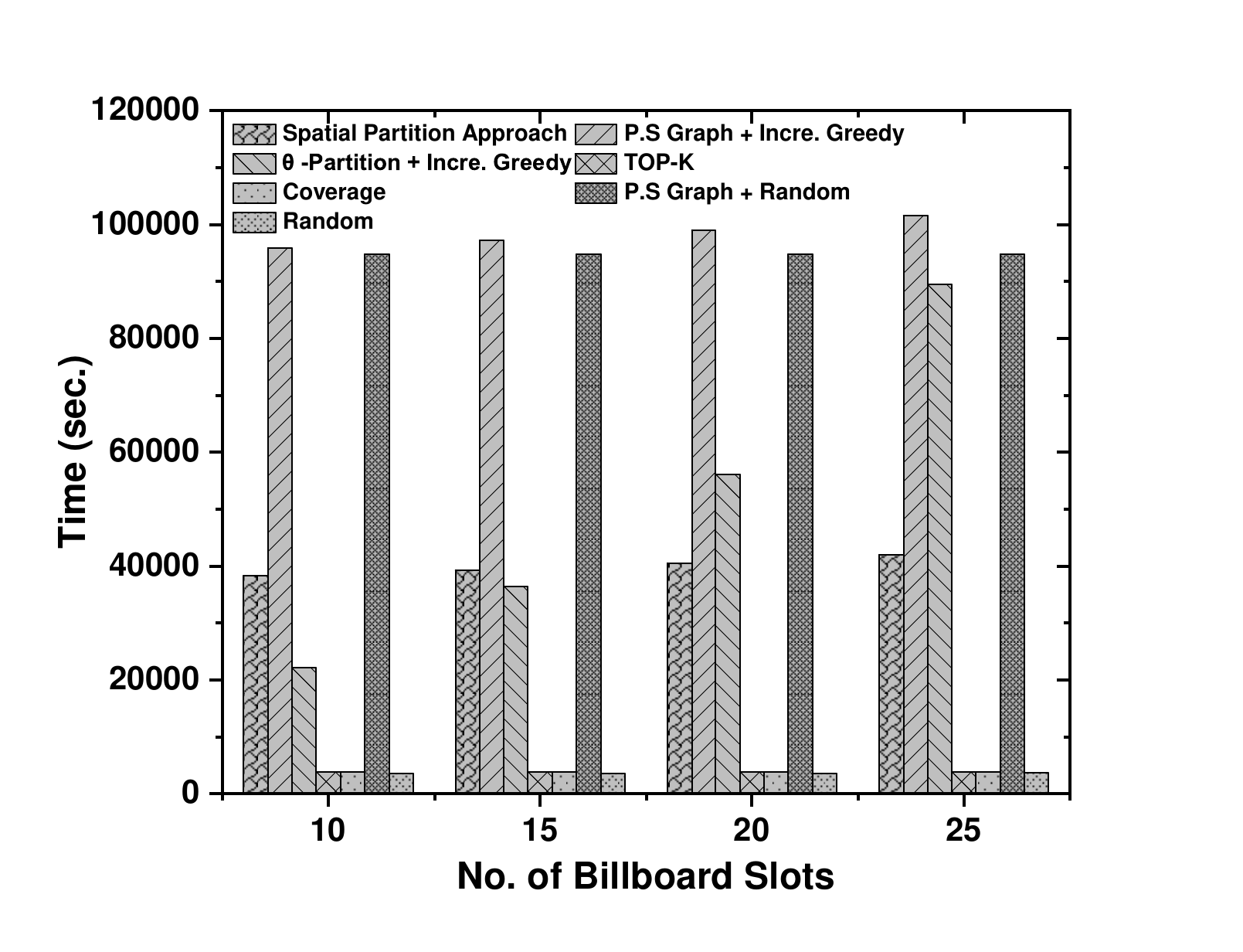} & \includegraphics[scale=0.16]{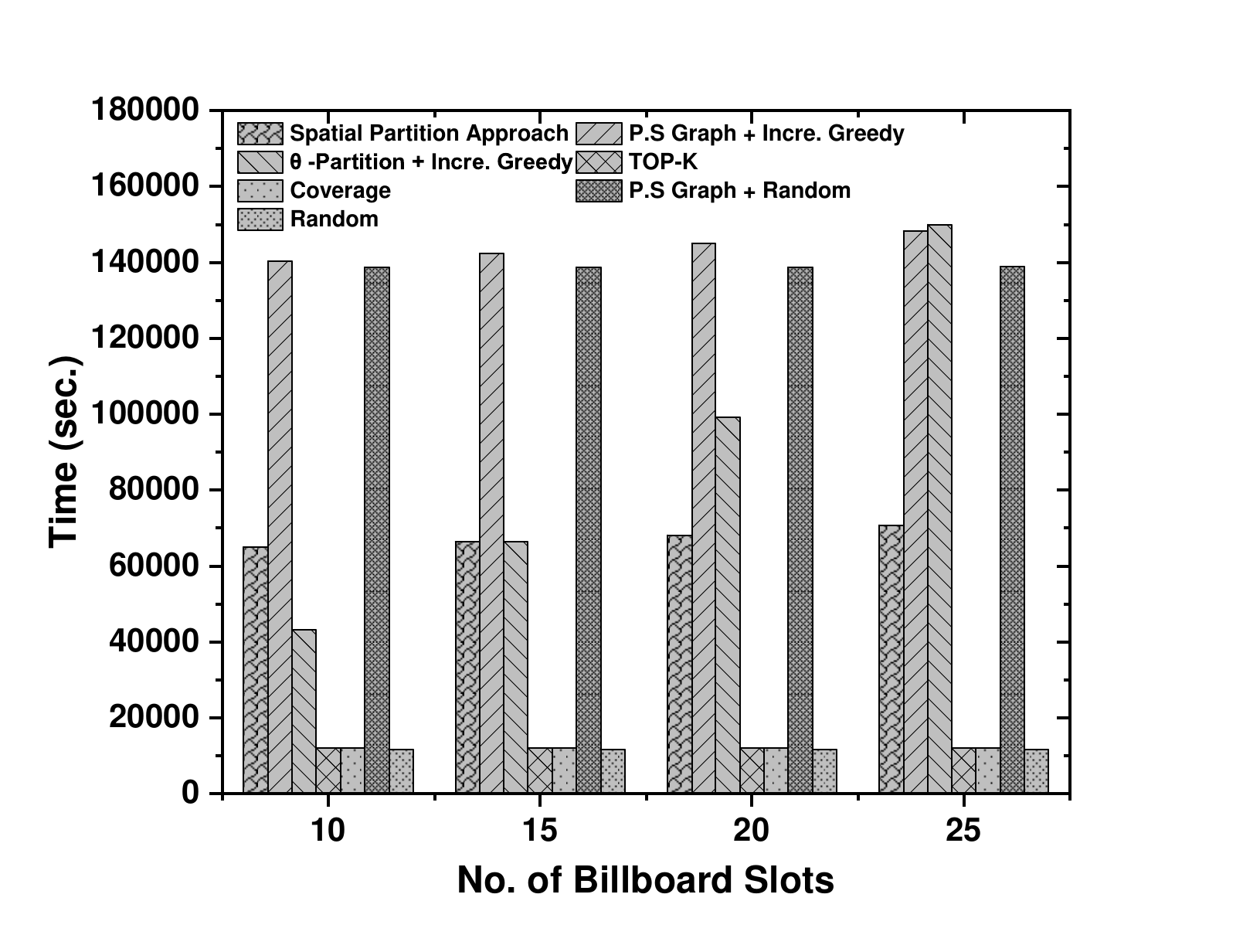} \\
(e) Location =`Airport' & (f) Location =`Bus Station' & (g) Location =`Train Station' &(h) Location =`Gym' \\

\end{tabular}
\caption{No. of Billboard Slot Vs. Time plots for different location}
\label{Fig:2_Time}
\end{figure*}
%%%%%%%%%%%%%%%%%%%%%%%%%%%%%%%%%%%%%%%%%%%%%%%%%%%%%%%%%%%

\paragraph{\textbf{$\theta$ Vs. Influence}} We vary the value of $\theta$ between $0.1$ to $0.4$ as selecting a good $\theta$ decides the effectiveness and efficiency of the proposed algorithms as shown in Figure \ref{Fig:3_Theta_Influence}.
The experimental results on four datasets (Beach, Mall, Bank, Park) shows the importance of $\theta$ value. From this, we have three main observations (1) With the increase of $\theta$ value, the influence quality decreases because of increasing influence overlap. (2) When $\theta$ value is $0.1$ the Spatial Partition approach and `PartSel' generate an excellent influence quality. (3) In the case of the location Bank, when $\theta$ value is $0.2$, both algorithms take average computational time while producing good quality influence. But in the case of $0.3$ and $0.4$ as $\theta$, the influence quality is deficient but takes less computational time. Beach, Mall, and Park results are very similar to Bank. Therefore, we choose the default $\theta$ value of $0.2$ for the rest of the experiments.

\begin{figure*}[!ht]
\centering
\begin{tabular}{cccc}
\includegraphics[scale=0.16]{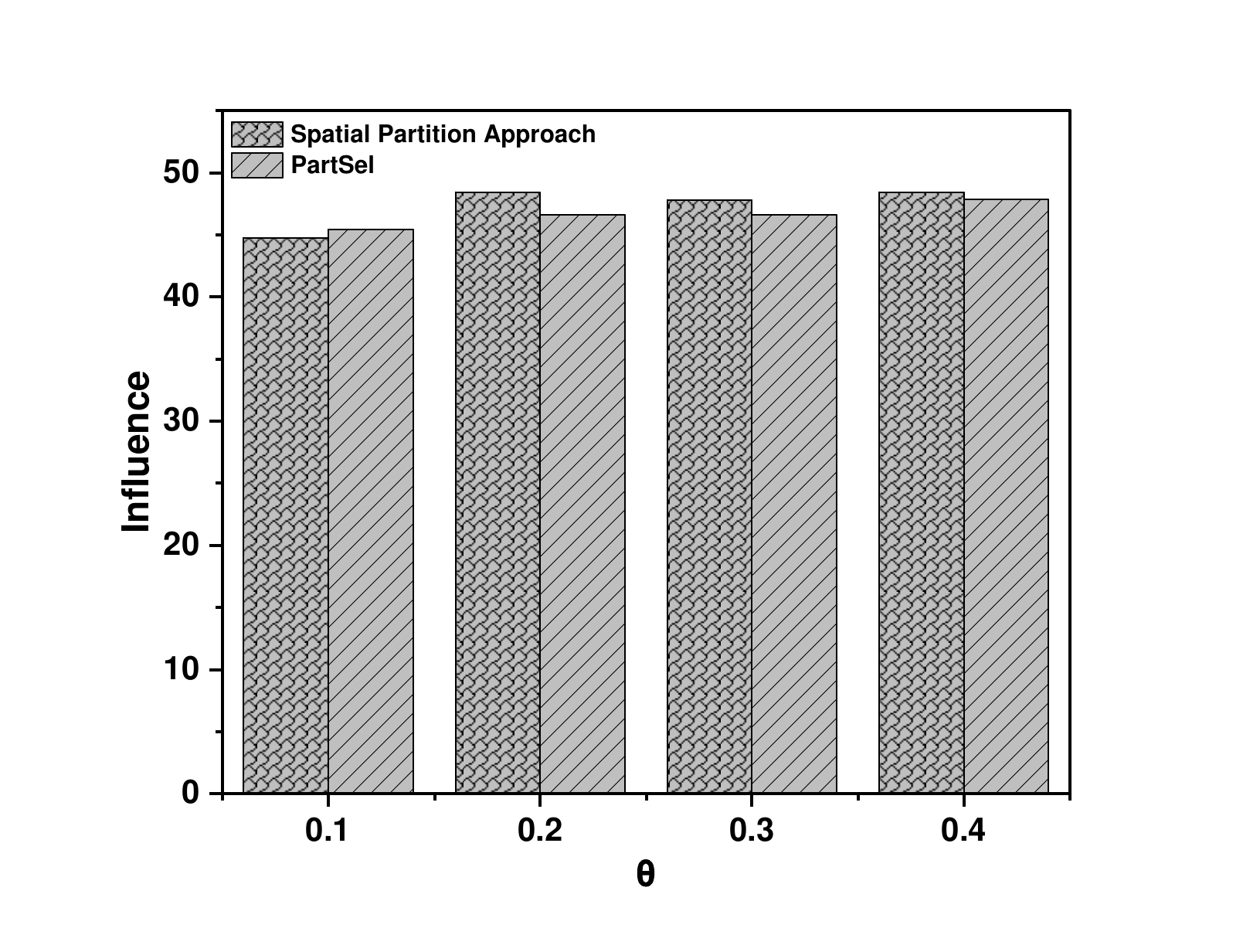} & \includegraphics[scale=0.16]{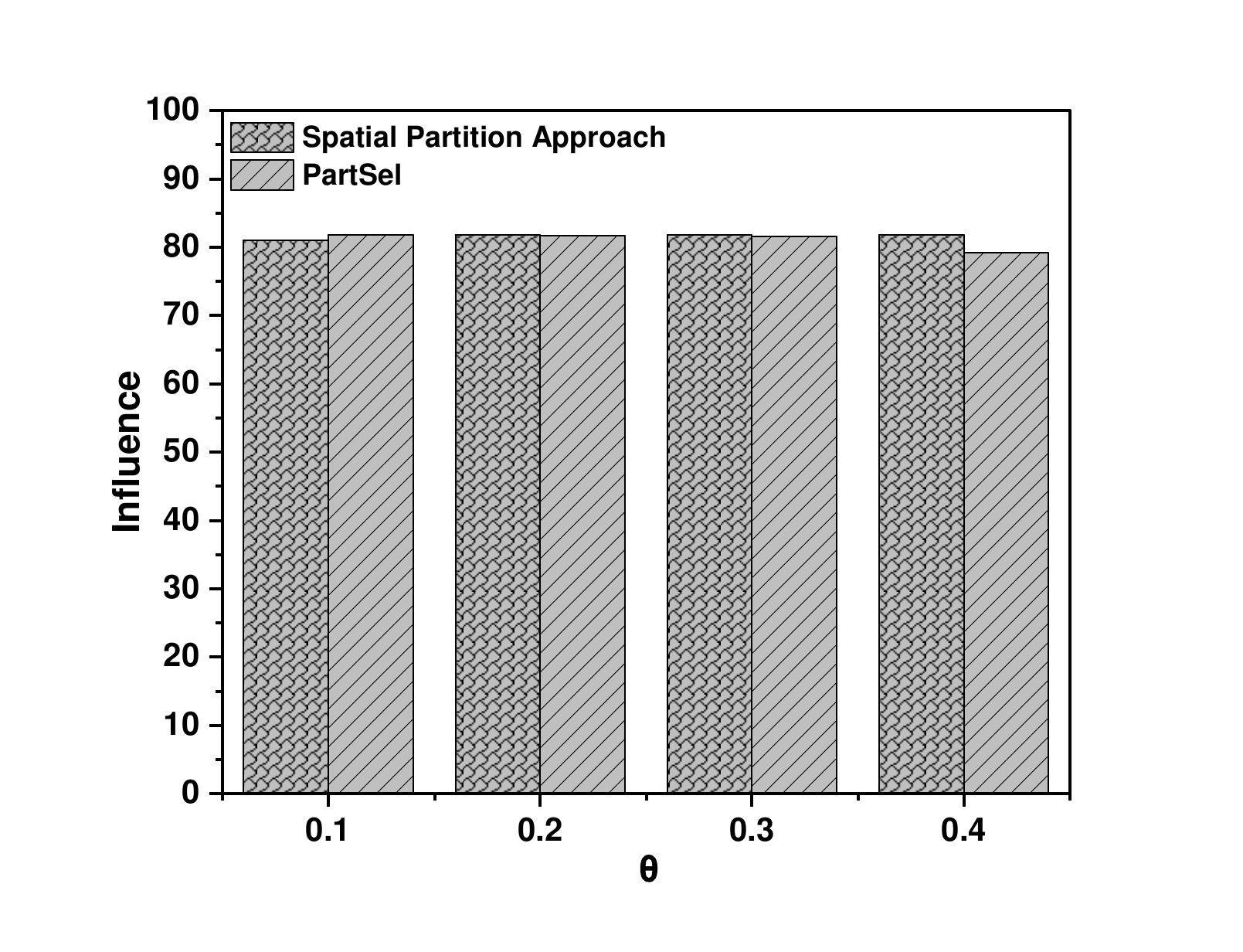} & \includegraphics[scale=0.16]{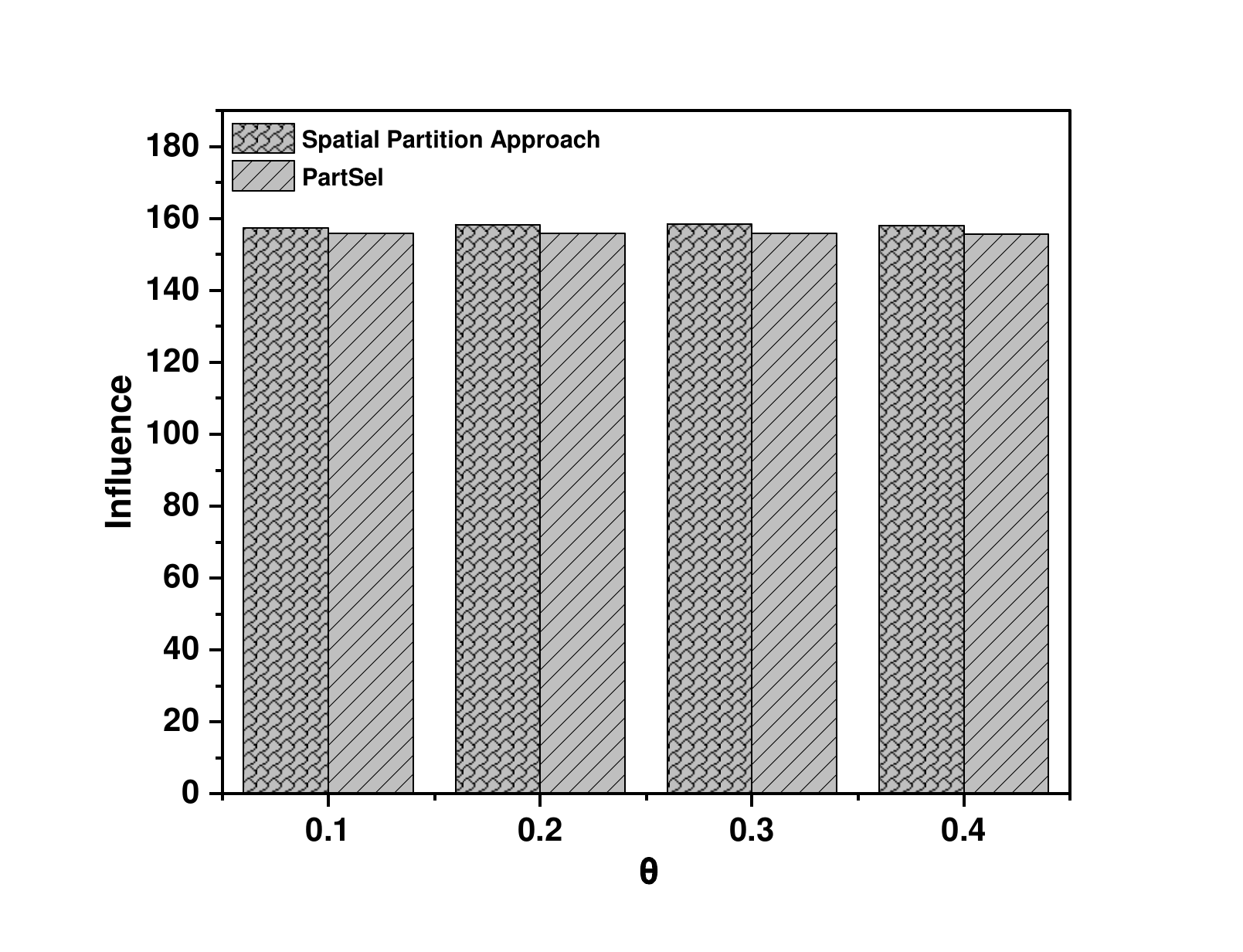} & \includegraphics[scale=0.16]{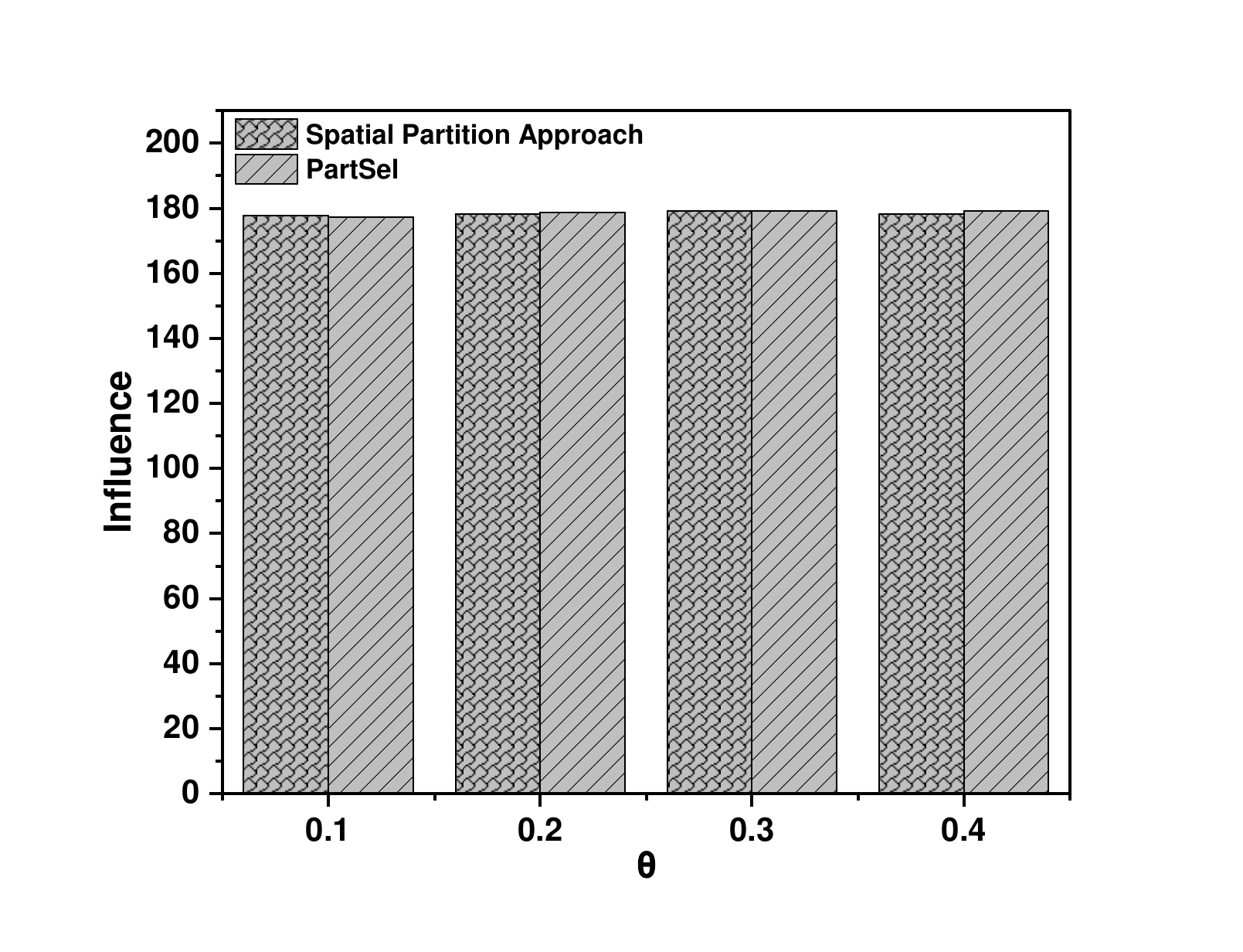} \\
(a) Location =`Beach' & (b) Location =`Mall' & (c) Location =`Bank' &(d) Location =`Park' \\
\end{tabular}
\caption{$\theta$ Vs. Influence plots for different location}
\label{Fig:3_Theta_Influence}
\end{figure*}
%%%%%%%%%%%%%%%%%%%%%%%%%%%%%%%%%%%%%%%%%%%%%%%%%%%%%%%%%%%%%%%%

\paragraph{\textbf{$\theta$ Vs. Time}}
In the Special Partition approach, as well as `PartSel', $\theta$ value is inversely proportional with time. When $\theta$ value increases, the computational time decreases, but at the same time, influence quality also decreases. For example, in case of location type `Mall' when $\theta$ is $0.1$ Spatial Partition approach takes $140$ secs but with the increase of $\theta$ , i.e., when $0.4$ as $\theta$ time requirement is $100$ sec, as shown in Figure \ref{Fig:4_Theta_Time}.

\begin{figure*}[!ht]
\centering
\begin{tabular}{cccc}
\includegraphics[scale=0.16]{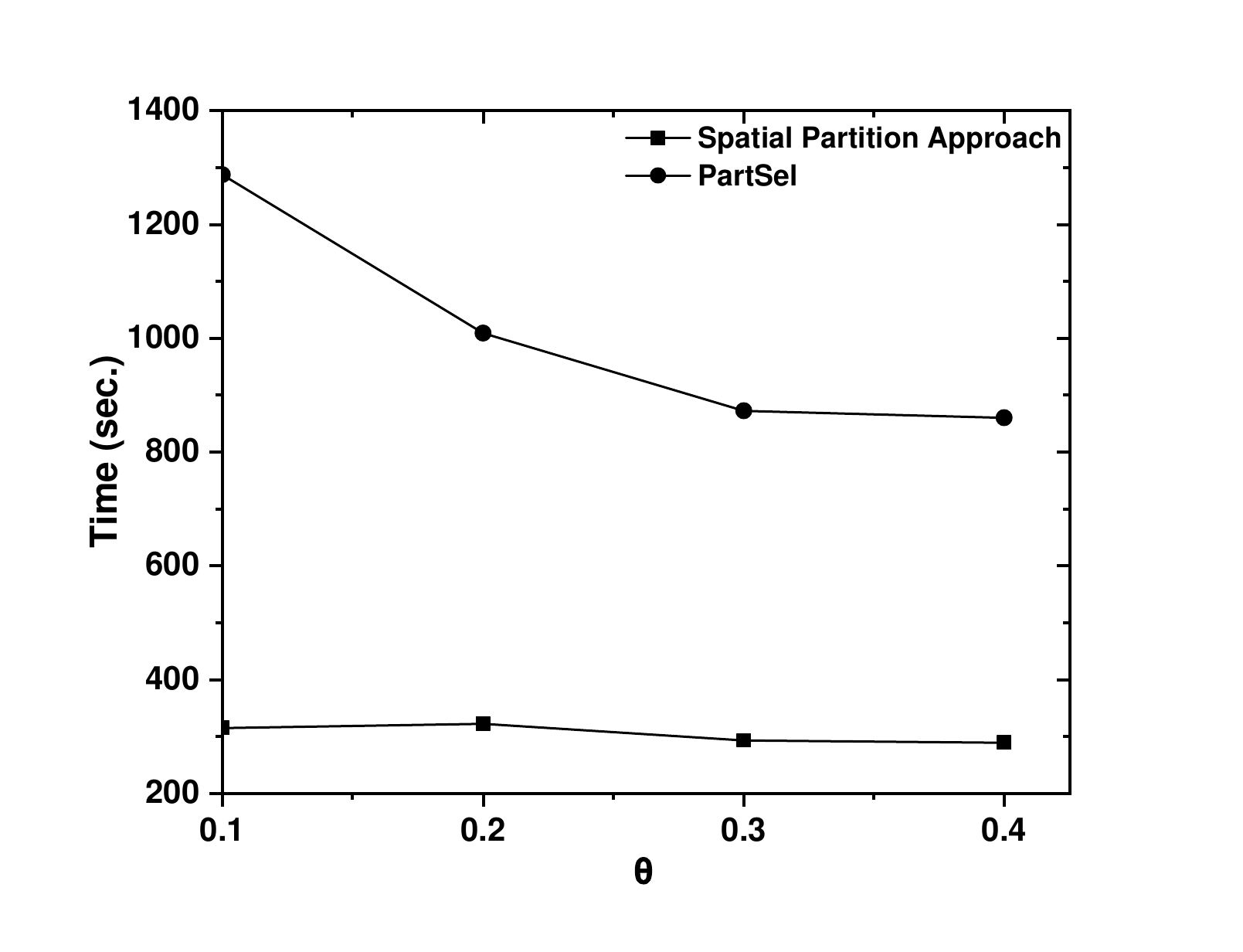} & \includegraphics[scale=0.16]{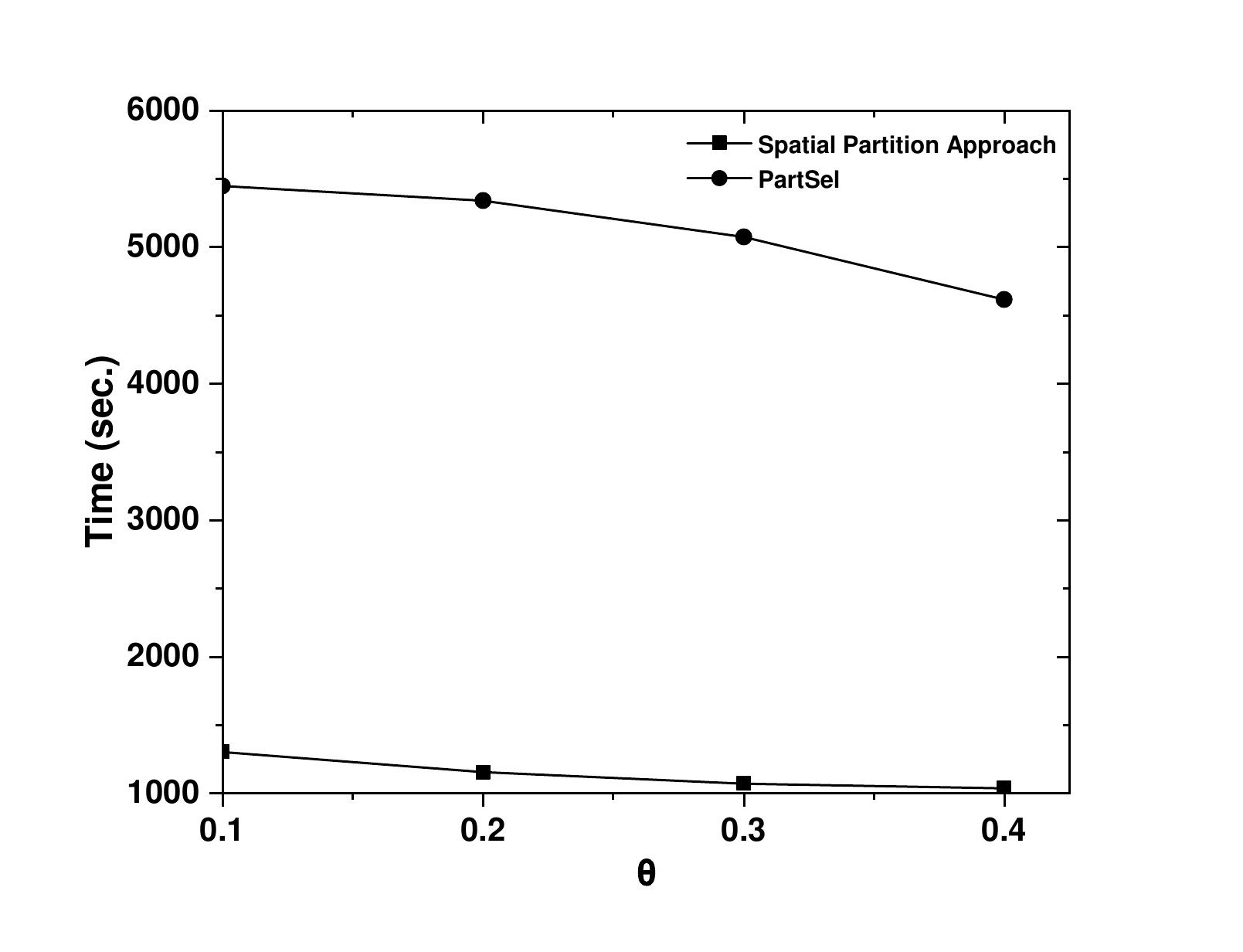} & \includegraphics[scale=0.16]{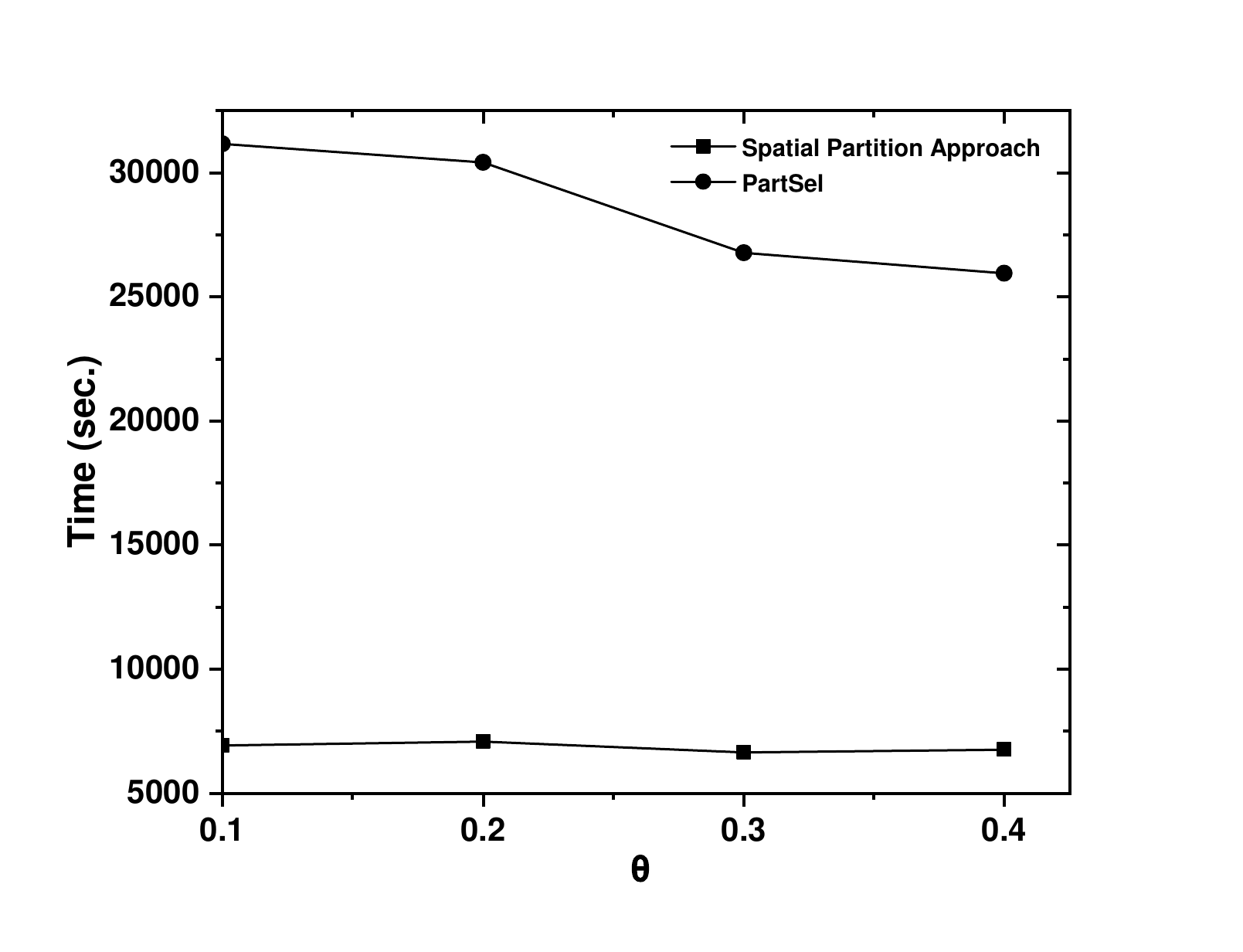} & \includegraphics[scale=0.16]{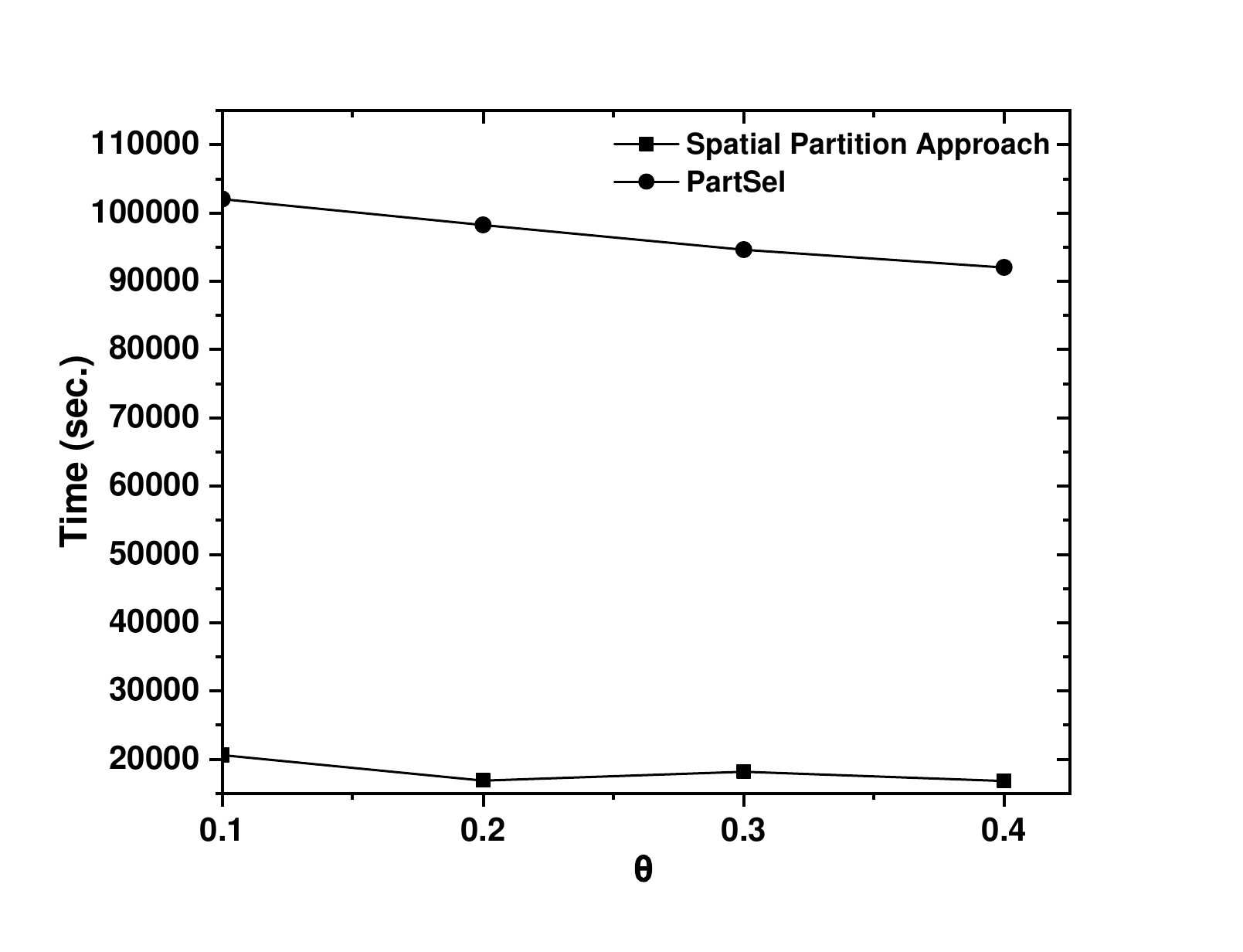} \\
(a) Location =`Beach' & (b) Location =`Mall' & (c) Location =`Bank' &(d) Location =`Park' \\
\end{tabular}
\caption{$\theta$ Vs. Time plots for different Location}
\label{Fig:4_Theta_Time}
\end{figure*}

%%%%%%%%%%%%%%%%%%%%%%%%%%%%%%%%%%%%%%%%%%%%%%%%%%%%%%%5
\paragraph{\textbf{$\theta$ Vs. No. of Cluster}}
To find out the importance of $\theta$ in terms of the number of cluster generations that will decide the influence quality, we vary $\theta$ value from $0.1$ to $0.4$, as shown in Figure \ref{Fig:5_Theta_Cluster}. With the increase of $\theta$, efficiency is improved and at the same time, the number of clusters also increases. For example, in the case of location type `Beach', there are $40$ clusters at $\theta$ = $0.2$, and efficiency does not much worse than that of
$\theta$ = $0.3$ as the largest cluster in `Beach’ covers almost $6\%$ of
total billboard slots. This observations are consistent with other location types also. 

\begin{figure*}[!ht]
\centering
\begin{tabular}{cccc}
\includegraphics[scale=0.16]{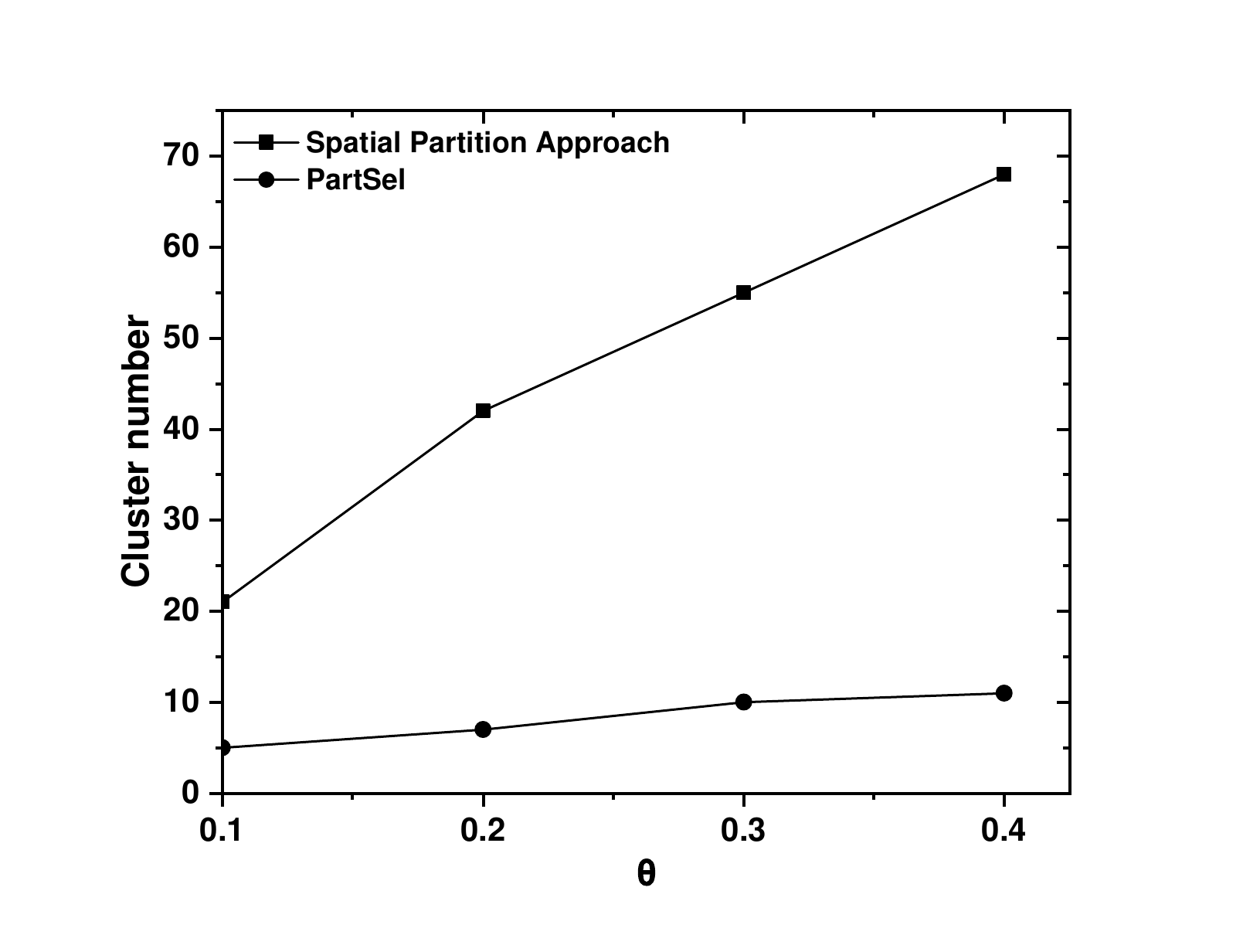} & \includegraphics[scale=0.16]{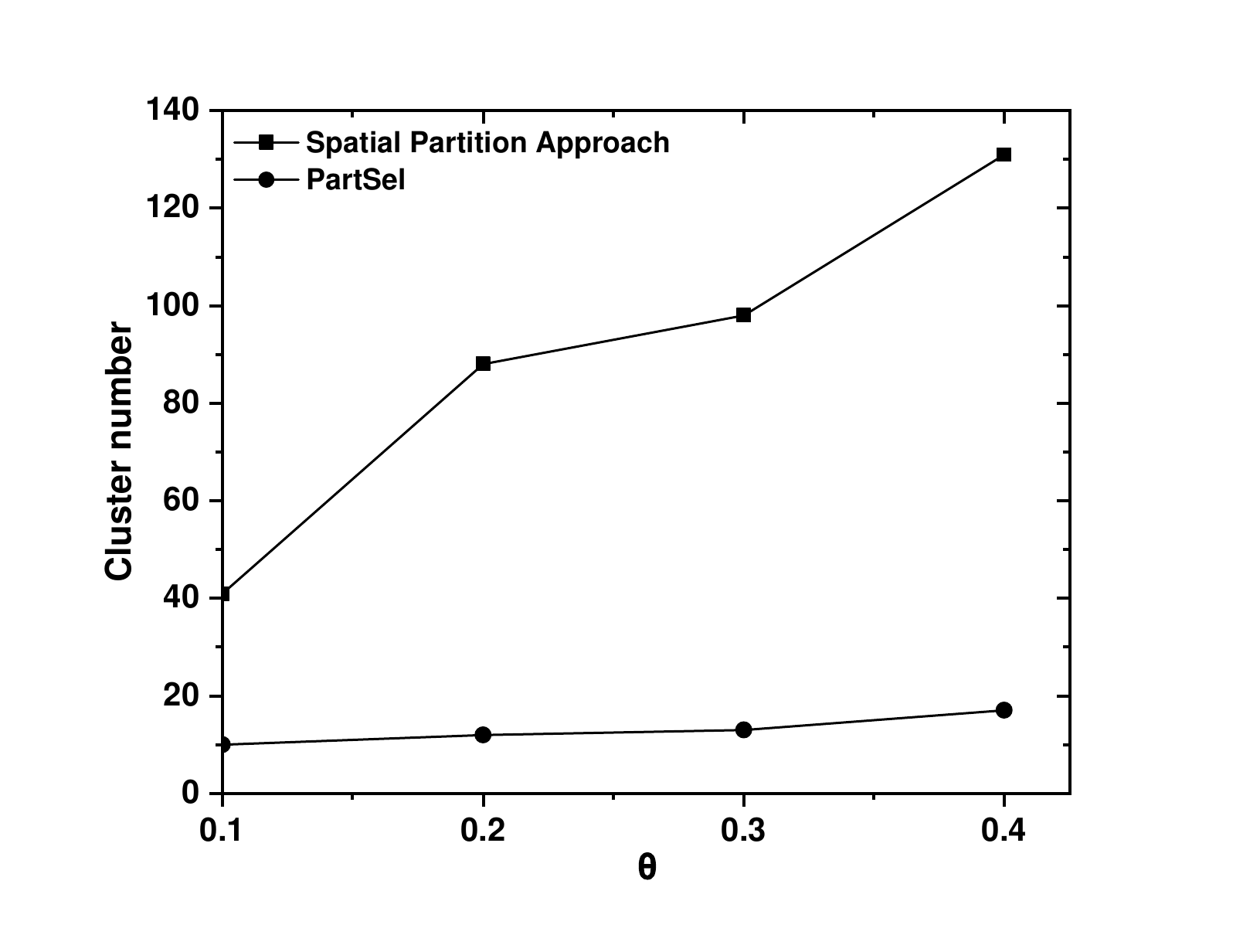} & \includegraphics[scale=0.16]{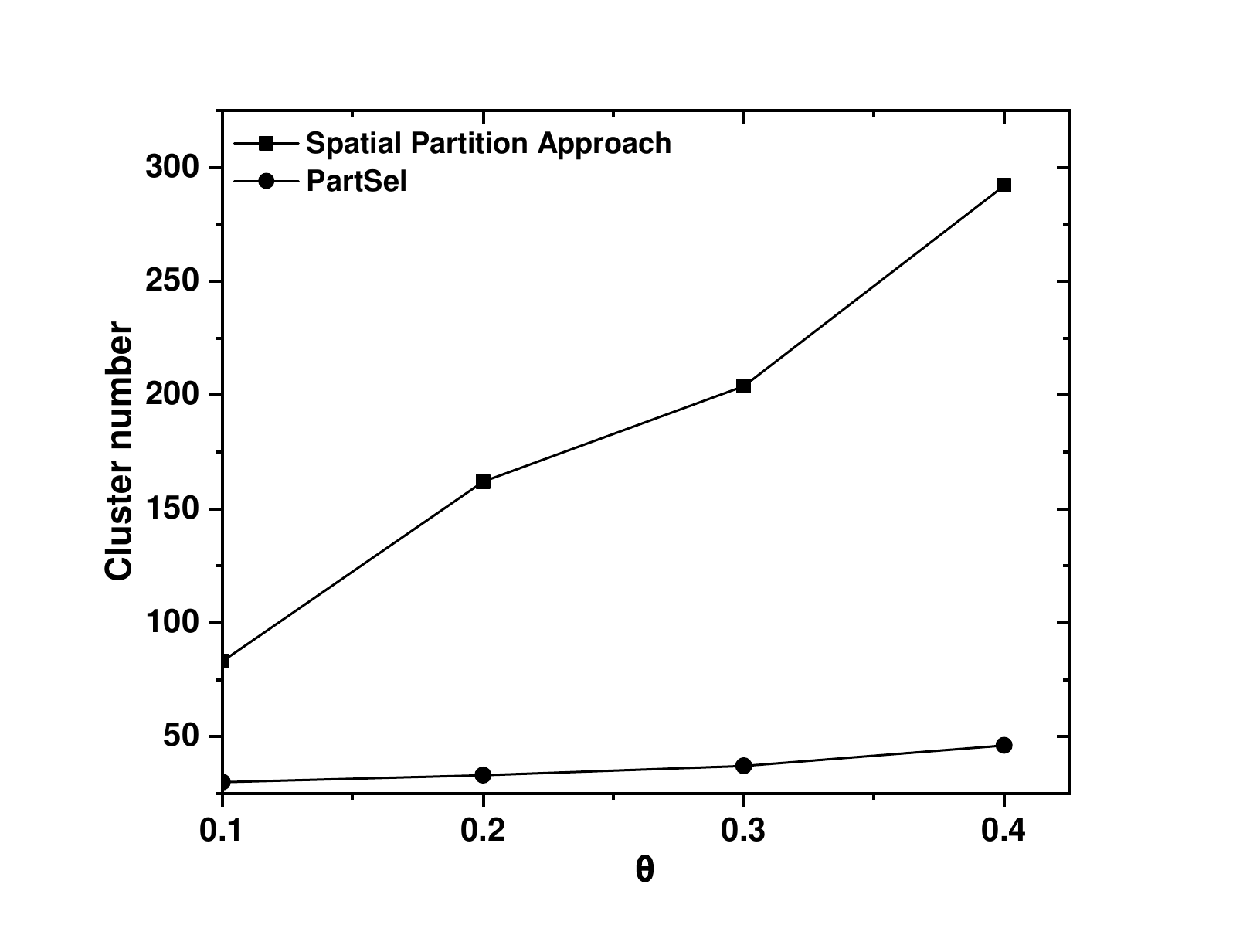} & \includegraphics[scale=0.16]{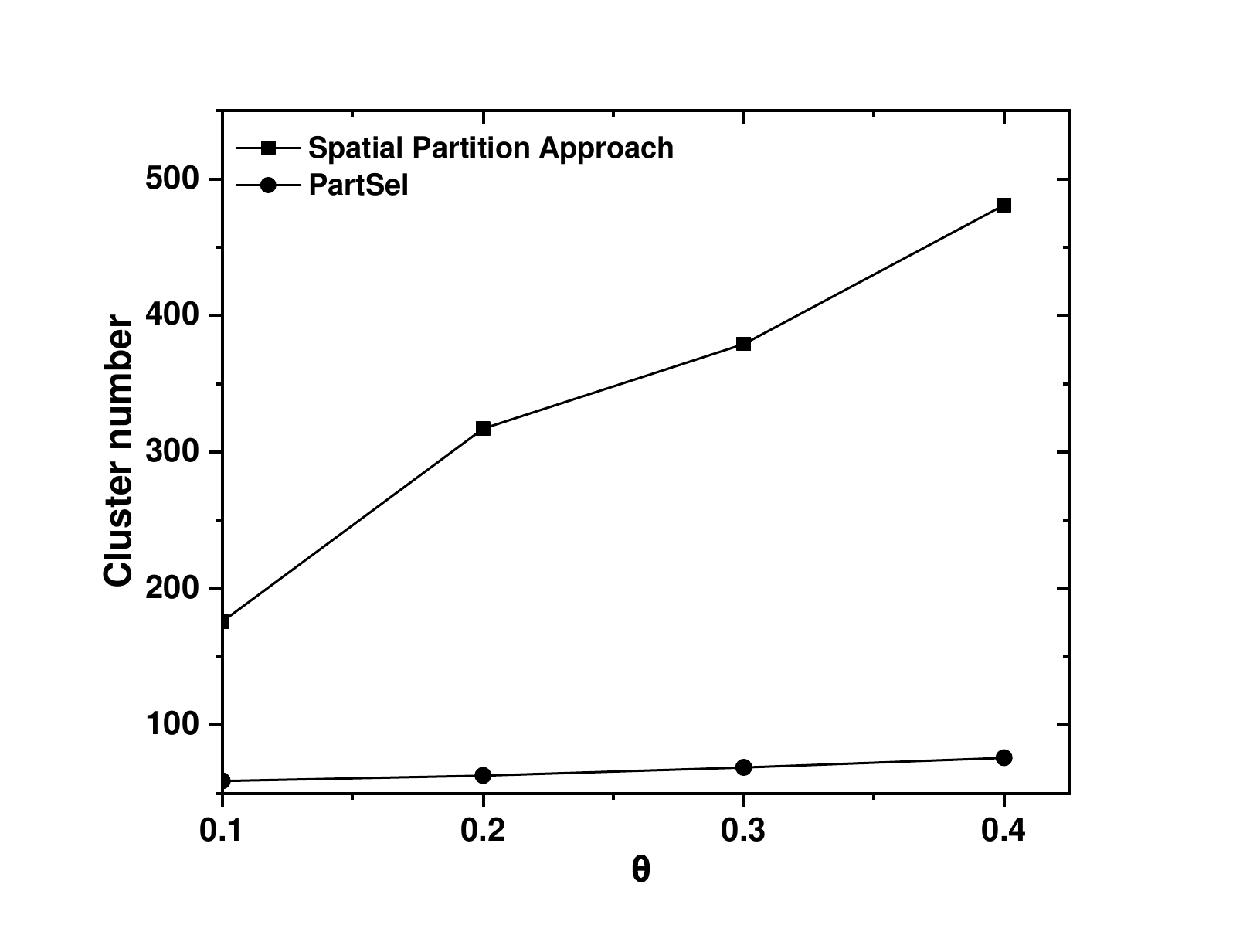} \\
(a) Location =`Beach' & (b) Location =`Mall' & (c) Location =`Bank' &(d) Location =`Park' \\
\end{tabular}
\caption{$\theta$ Vs. Number of Cluster plots for different location}
\label{Fig:5_Theta_Cluster}
\end{figure*}

%%%%%%%%%%%%%%%%%%%%%%%%%%%%%%%%%%%%%%%%%%%%%%%%%%%%%%%%%
%%%%% Distance Vs Influence Plots
\paragraph{\textbf{Distance Vs. Influence}}
To determine the importance of $\lambda$, we varied $\lambda$ value from $25$ to $100$ meters. We observed that it determines the relationships between billboard slots and trajectories, as shown in Figure \ref{Fig:6_DistanceVsInf}. From our experiments, we have made two observations (1) with the increase of $\lambda$, the influence quality increases. For example, when the location type is `Gym’ and $\lambda = 25$ meters, the proposed and existing algorithms influence is lesser than  when $\lambda = 100$ meters. This happens because the distance increases, and one billboard slot can influence more trajectories. (2) We observe that our proposed approaches outperform existing algorithms for all the values of $\lambda$. When $\lambda = 75$ meters and $\lambda = 100$ meters, slight changes are shown in influence quality. These observations are consistent with other location types also, and for this reason, we choose $\lambda = 100$ meter for all the location types.

\begin{figure*}[!ht]
\centering
\begin{tabular}{cccc}
\includegraphics[scale=0.16]{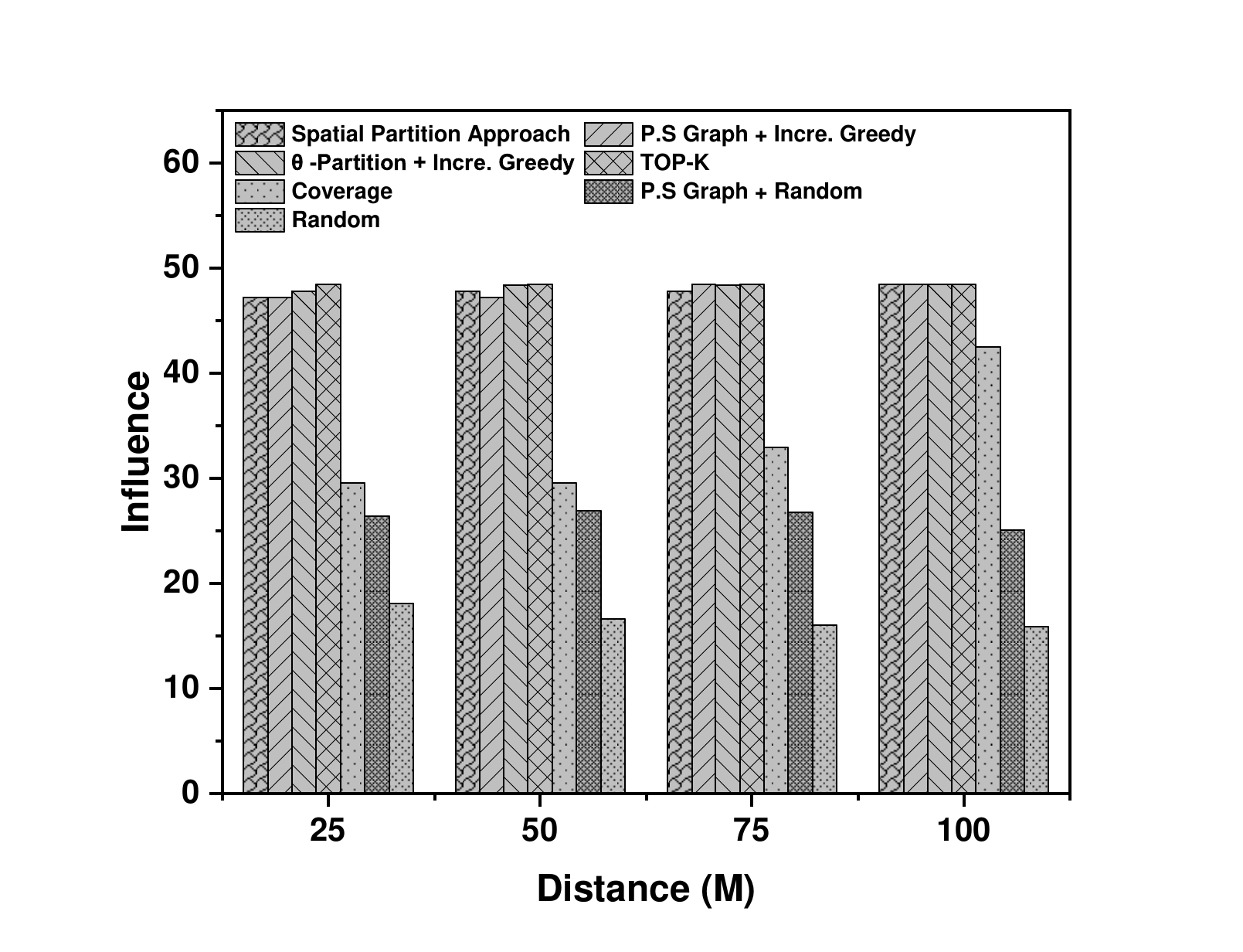} & \includegraphics[scale=0.16]{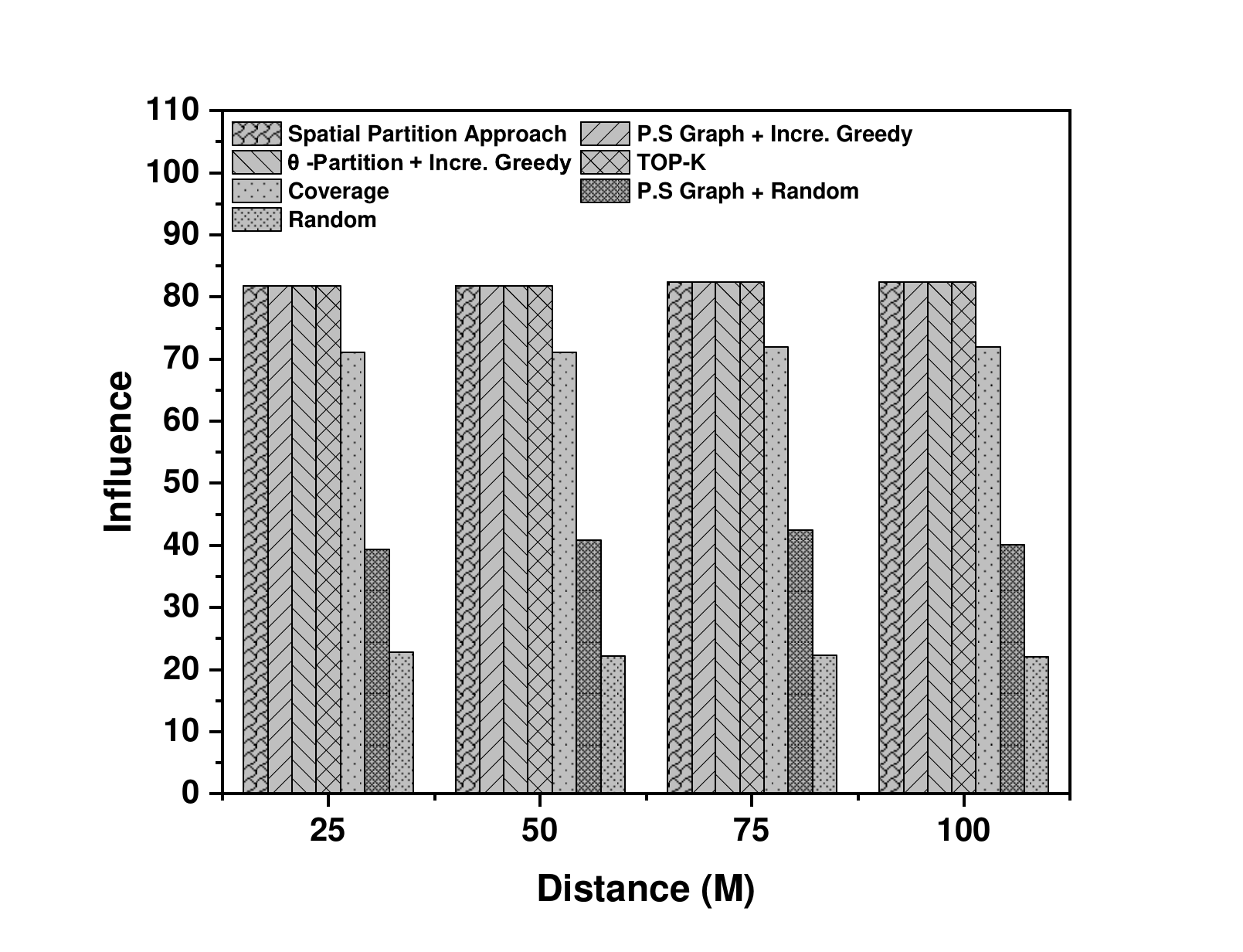} & \includegraphics[scale=0.16]{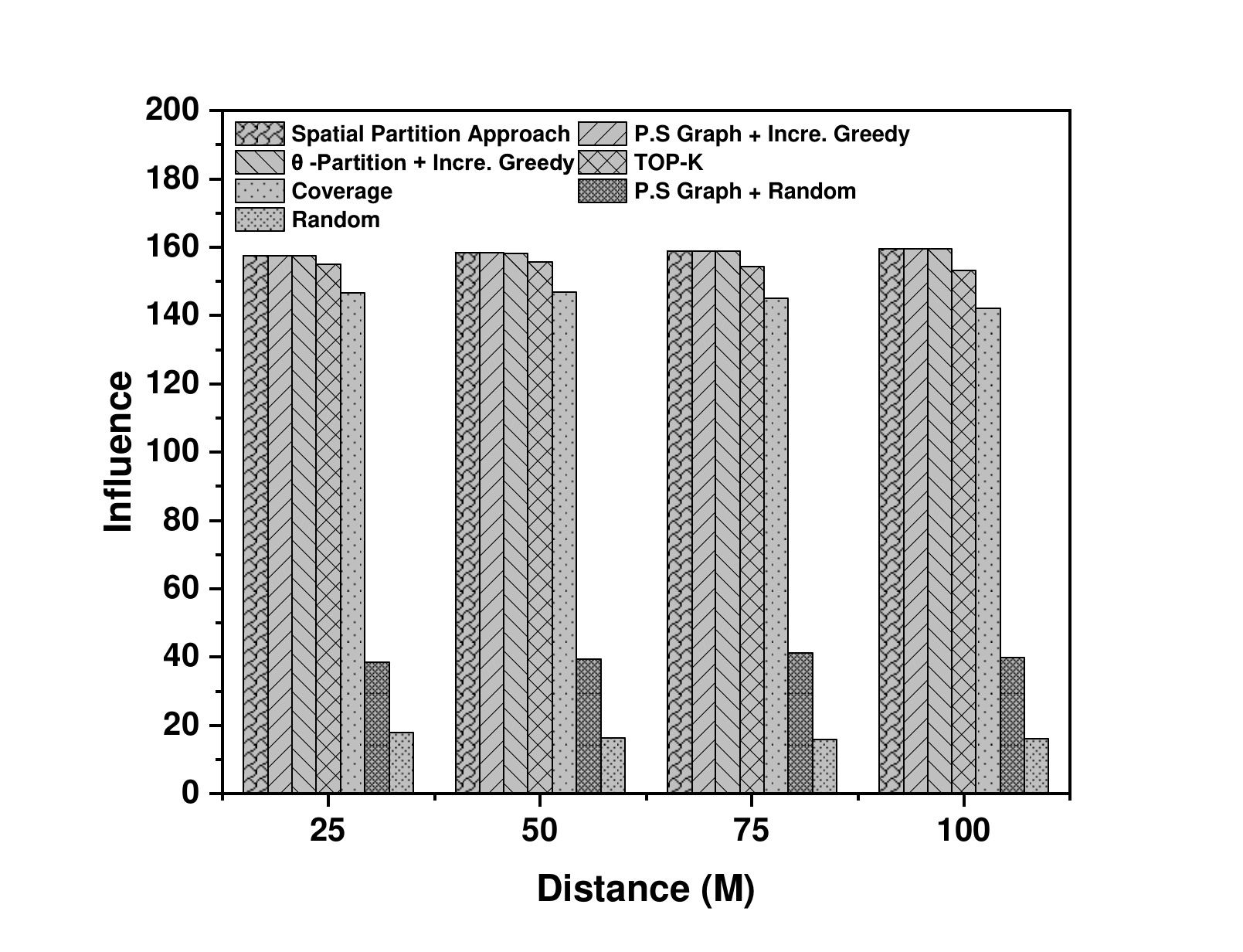} & \includegraphics[scale=0.16]{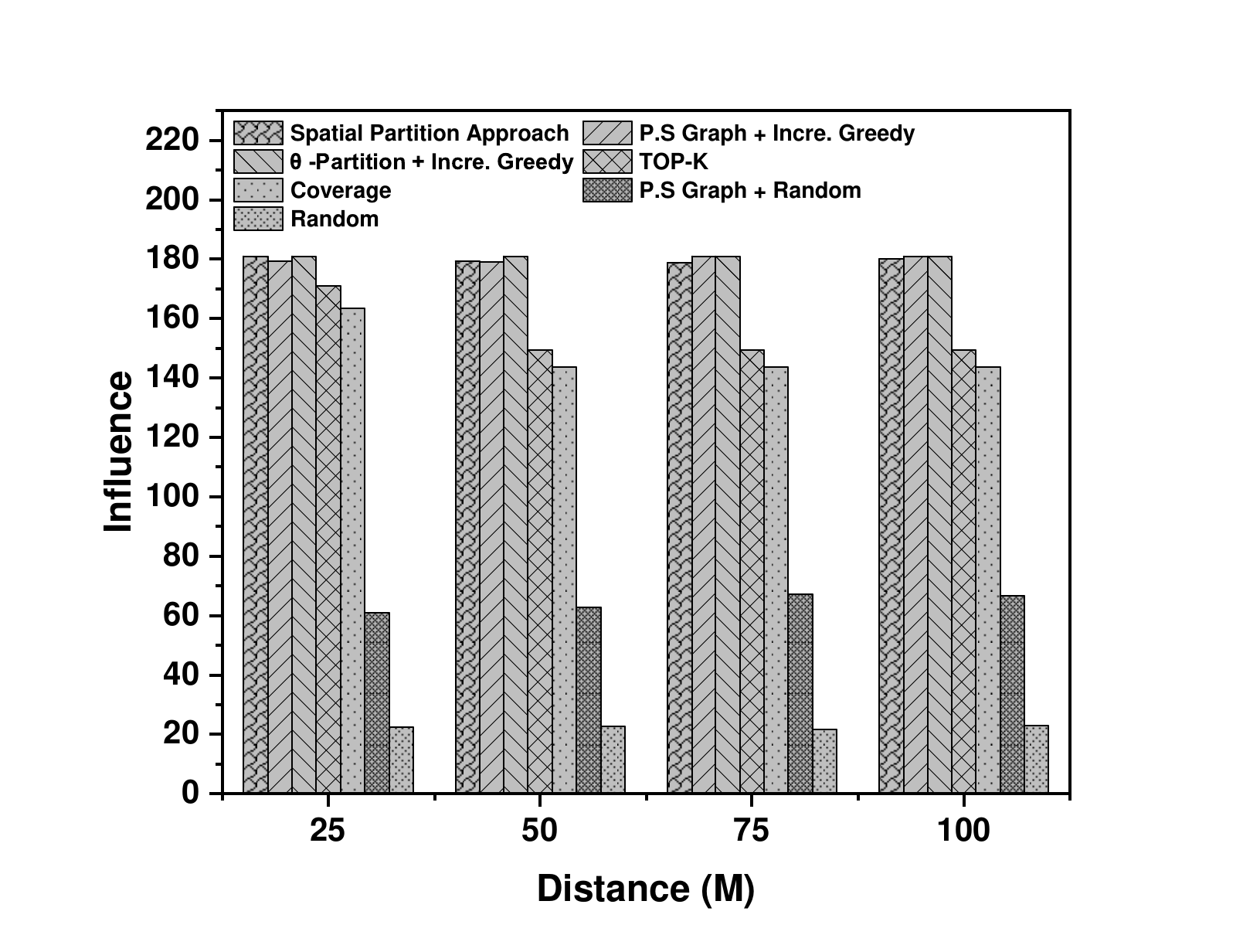} \\
(a) Location =`Beach' & (b) Location =`Mall' & (c) Location =`Bank' &(d) Location =`Park' \\
\includegraphics[scale=0.16]{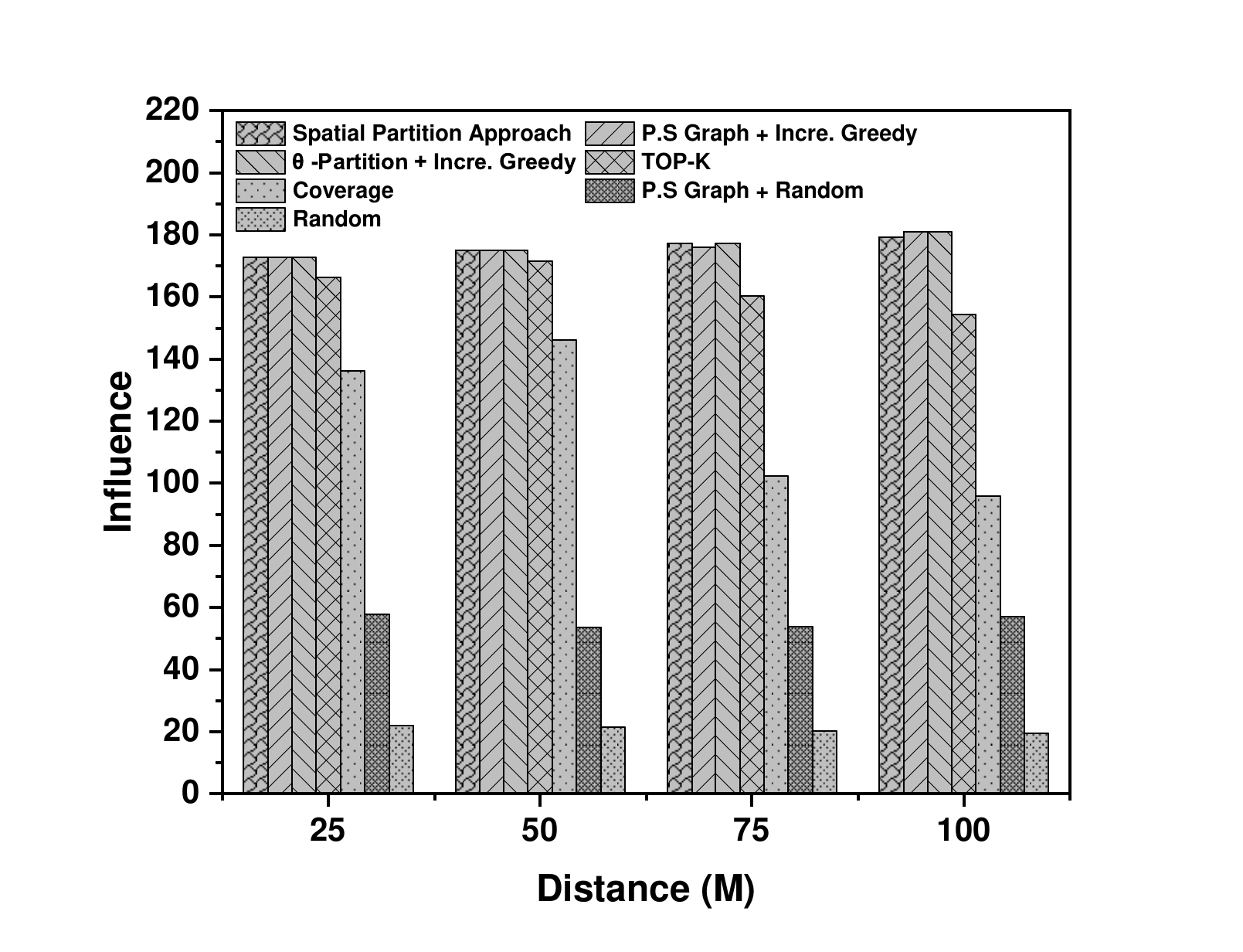} & \includegraphics[scale=0.16]{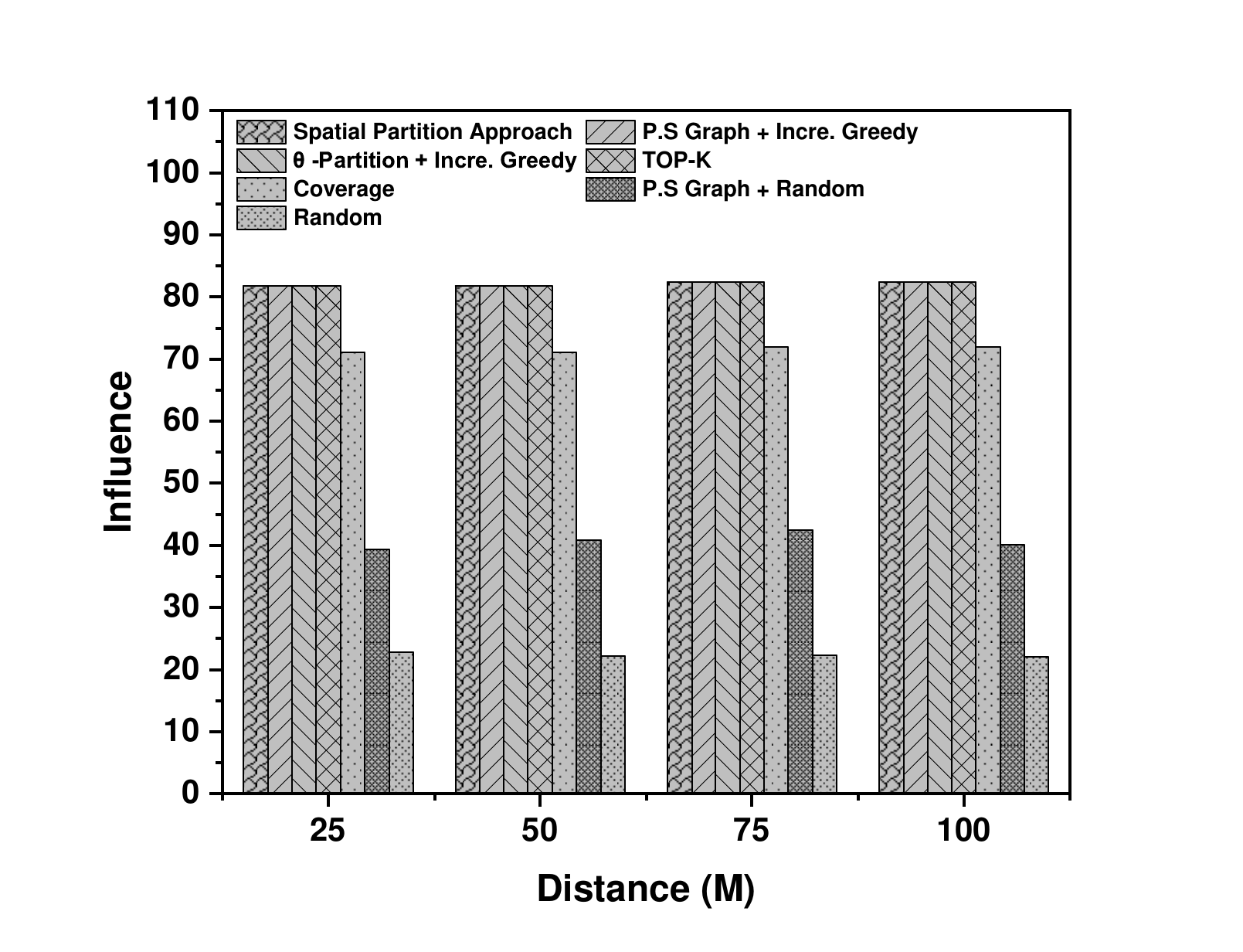} & \includegraphics[scale=0.16]{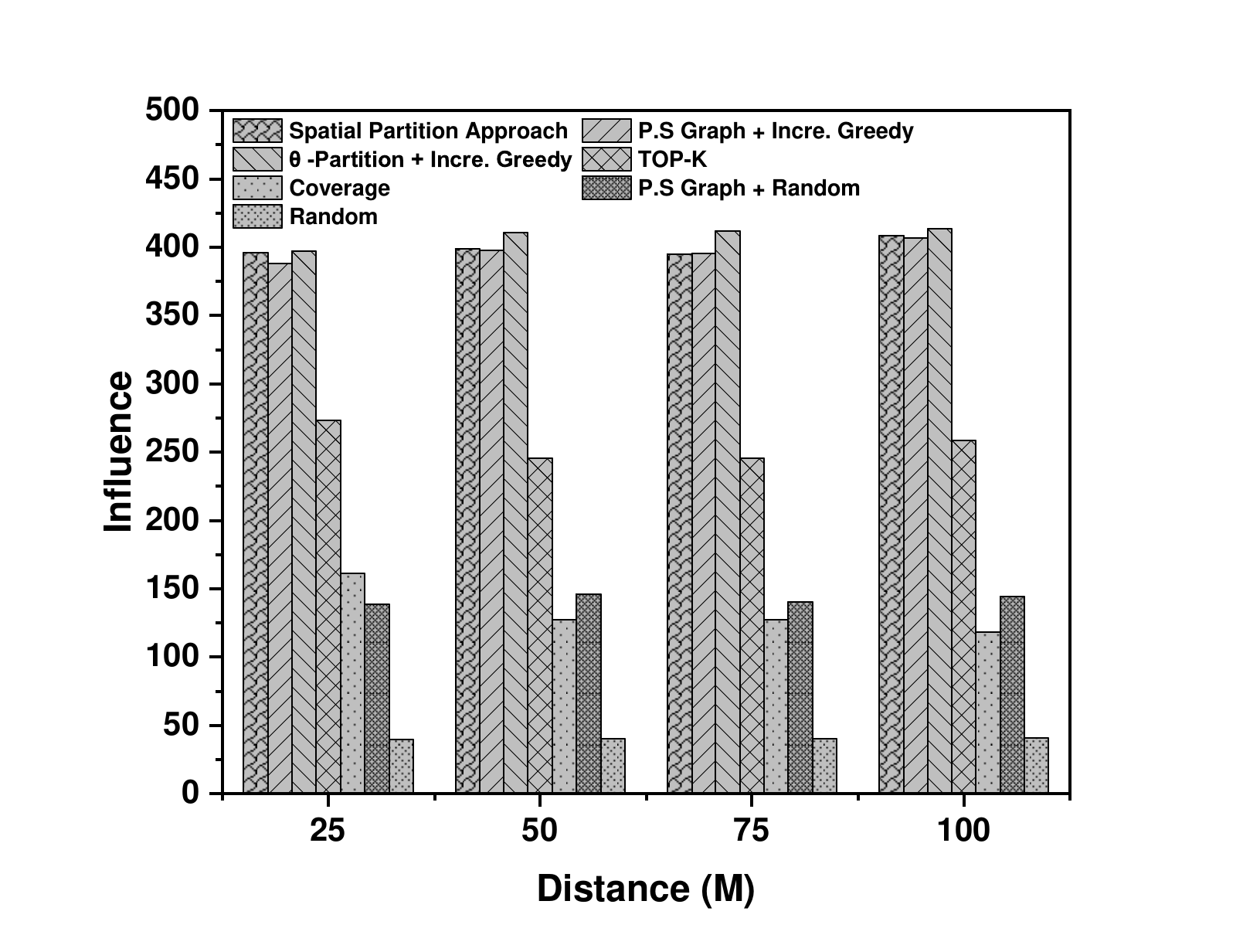} & \includegraphics[scale=0.16]{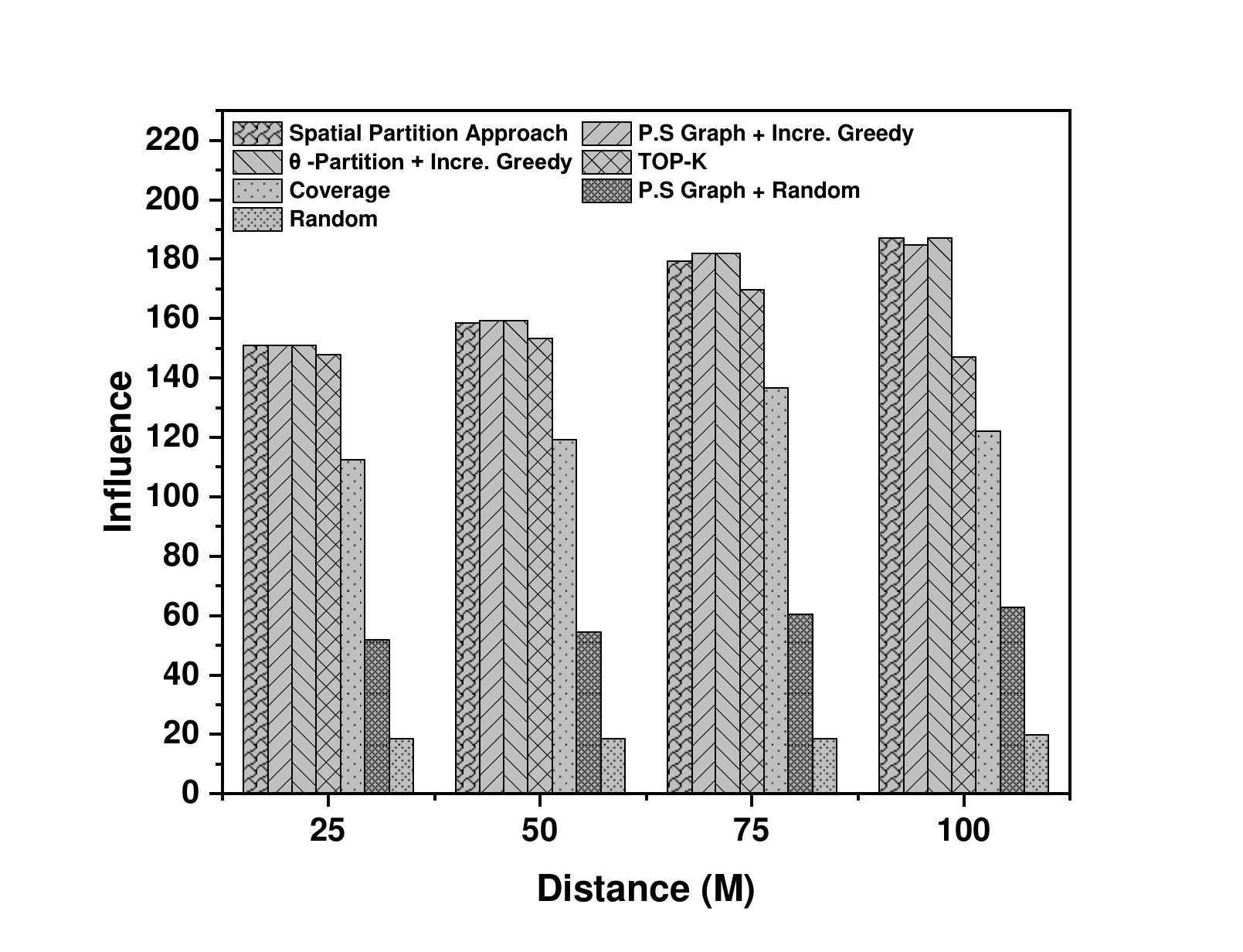} \\
(e) Location =`Airport' & (f) Location =`Bus Station' & (g) Location =`Train Station' &(h) Location =`Gym' \\
\end{tabular}
\caption{Distance Vs. Influence plots for different location}
\label{Fig:6_DistanceVsInf}
\end{figure*}

%%%%%%%%%%%%%%%%%%%%%%%%%%%%%%%%%%%%%%%%%%%%%%%%%%%

\section{Conclusion}\label{Sec:CFD}
In this paper, we examined the \textsl{IBSP} problem, assuming knowledge of trajectory information and billboard locations. We discovered that the problem is not only difficult but also challenging to approximate within a constant factor under a reasonable complexity theoretic assumption. We developed two efficient approaches to tackle this problem. The first approach utilizes the submodularity property of the influence function to identify the most influential slots and creates a pruned submodularity graph to produce an approximate solution. In the second approach, we divide the space into clusters based on influence overlapping ratio and apply the Pruned Submodularity Graph approach for pruning. We conducted experiments with real-world datasets and found that our proposed approaches yielded almost equivalent influence compared to the baseline method `PartSel', and outperformed other methods. Our future work will focus on developing even more efficient solution methodologies.

\bibliographystyle{IEEEtran}
\bibliography{Paper}

\begin{thebibliography}{10}
\providecommand{\url}[1]{#1}
\csname url@samestyle\endcsname
\providecommand{\newblock}{\relax}
\providecommand{\bibinfo}[2]{#2}
\providecommand{\BIBentrySTDinterwordspacing}{\spaceskip=0pt\relax}
\providecommand{\BIBentryALTinterwordstretchfactor}{4}
\providecommand{\BIBentryALTinterwordspacing}{\spaceskip=\fontdimen2\font plus
\BIBentryALTinterwordstretchfactor\fontdimen3\font minus
  \fontdimen4\font\relax}
\providecommand{\BIBforeignlanguage}[2]{{%
\expandafter\ifx\csname l@#1\endcsname\relax
\typeout{** WARNING: IEEEtran.bst: No hyphenation pattern has been}%
\typeout{** loaded for the language `#1'. Using the pattern for}%
\typeout{** the default language instead.}%
\else
\language=\csname l@#1\endcsname
\fi
#2}}
\providecommand{\BIBdecl}{\relax}
\BIBdecl

\bibitem{10.1145/3292500.3330829}
\BIBentryALTinterwordspacing
Y.~Zhang, Y.~Li, Z.~Bao, S.~Mo, and P.~Zhang, ``Optimizing impression counts
  for outdoor advertising,'' in \emph{Proceedings of the 25th ACM SIGKDD
  International Conference on Knowledge Discovery \&amp; Data Mining}, ser. KDD
  '19.\hskip 1em plus 0.5em minus 0.4em\relax New York, NY, USA: Association
  for Computing Machinery, 2019, p. 1205–1215. [Online]. Available:
  \url{https://doi.org/10.1145/3292500.3330829}
\BIBentrySTDinterwordspacing

\bibitem{10.2307/1153228}
\BIBentryALTinterwordspacing
G.~Feder, R.~E. Just, and D.~Zilberman, ``Adoption of agricultural innovations
  in developing countries: A survey,'' \emph{Economic Development and Cultural
  Change}, vol.~33, no.~2, pp. 255--298, 1985. [Online]. Available:
  \url{http://www.jstor.org/stable/1153228}
\BIBentrySTDinterwordspacing

\bibitem{f9b31366586b47d389955f209b69da27}
S.~Lee, ``\BIBforeignlanguage{English}{Examining the factors that influence
  early adopters' smartphone adoption: The case of college students},''
  \emph{\BIBforeignlanguage{English}{Telematics and Informatics}}, vol.~31,
  no.~2, pp. 308--318, May 2014.

\bibitem{SIERZCHULA2014183}
\BIBentryALTinterwordspacing
W.~Sierzchula, S.~Bakker, K.~Maat, and B.~{van Wee}, ``The influence of
  financial incentives and other socio-economic factors on electric vehicle
  adoption,'' \emph{Energy Policy}, vol.~68, pp. 183--194, 2014. [Online].
  Available:
  \url{https://www.sciencedirect.com/science/article/pii/S0301421514000822}
\BIBentrySTDinterwordspacing

\bibitem{wang2022data}
L.~Wang, Z.~Yu, B.~Guo, D.~Yang, L.~Ma, Z.~Liu, and F.~Xiong, ``Data-driven
  targeted advertising recommendation system for outdoor billboard,'' \emph{ACM
  Transactions on Intelligent Systems and Technology (TIST)}, vol.~13, no.~2,
  pp. 1--23, 2022.

\bibitem{dai2015personalized}
J.~Dai, B.~Yang, C.~Guo, and Z.~Ding, ``Personalized route recommendation using
  big trajectory data,'' in \emph{2015 IEEE 31st international conference on
  data engineering}.\hskip 1em plus 0.5em minus 0.4em\relax IEEE, 2015, pp.
  543--554.

\bibitem{qu2019profitable}
B.~Qu, W.~Yang, G.~Cui, and X.~Wang, ``Profitable taxi travel route
  recommendation based on big taxi trajectory data,'' \emph{IEEE Transactions
  on Intelligent Transportation Systems}, vol.~21, no.~2, pp. 653--668, 2019.

\bibitem{xue2019rapid}
Q.~Xue, K.~Wang, J.~J. Lu, and Y.~Liu, ``Rapid driving style recognition in
  car-following using machine learning and vehicle trajectory data,''
  \emph{Journal of advanced transportation}, vol. 2019, 2019.

\bibitem{zhang2020towards}
P.~Zhang, Z.~Bao, Y.~Li, G.~Li, Y.~Zhang, and Z.~Peng, ``Towards an optimal
  outdoor advertising placement: When a budget constraint meets moving
  trajectories,'' \emph{ACM Transactions on Knowledge Discovery from Data
  (TKDD)}, vol.~14, no.~5, pp. 1--32, 2020.

\bibitem{li2018influence}
Y.~Li, J.~Fan, Y.~Wang, and K.-L. Tan, ``Influence maximization on social
  graphs: A survey,'' \emph{IEEE Transactions on Knowledge and Data
  Engineering}, vol.~30, no.~10, pp. 1852--1872, 2018.

\bibitem{5767892}
Z.~Zhou, W.~Wu, X.~Li, M.~L. Lee, and W.~Hsu, ``Maxfirst for maxbrknn,'' in
  \emph{2011 IEEE 27th International Conference on Data Engineering}, 2011, pp.
  828--839.

\bibitem{10.1145/2588555.2588561}
\BIBentryALTinterwordspacing
G.~Li, S.~Chen, J.~Feng, K.-l. Tan, and W.-s. Li, ``Efficient location-aware
  influence maximization,'' in \emph{Proceedings of the 2014 ACM SIGMOD
  International Conference on Management of Data}, ser. SIGMOD '14.\hskip 1em
  plus 0.5em minus 0.4em\relax New York, NY, USA: Association for Computing
  Machinery, 2014, p. 87–98. [Online]. Available:
  \url{https://doi.org/10.1145/2588555.2588561}
\BIBentrySTDinterwordspacing

\bibitem{7534856}
D.~Liu, D.~Weng, Y.~Li, J.~Bao, Y.~Zheng, H.~Qu, and Y.~Wu, ``Smartadp: Visual
  analytics of large-scale taxi trajectories for selecting billboard
  locations,'' \emph{IEEE Transactions on Visualization and Computer Graphics},
  vol.~23, no.~1, pp. 1--10, 2017.

\bibitem{10.14778/1687627.1687754}
\BIBentryALTinterwordspacing
R.~C.-W. Wong, M.~T. \"{O}zsu, P.~S. Yu, A.~W.-C. Fu, and L.~Liu, ``Efficient
  method for maximizing bichromatic reverse nearest neighbor,'' \emph{Proc.
  VLDB Endow.}, vol.~2, no.~1, p. 1126–1137, aug 2009. [Online]. Available:
  \url{https://doi.org/10.14778/1687627.1687754}
\BIBentrySTDinterwordspacing

\bibitem{10.5555/1083592.1083701}
T.~Xia, D.~Zhang, E.~Kanoulas, and Y.~Du, ``On computing top-t most influential
  spatial sites,'' in \emph{Proceedings of the 31st International Conference on
  Very Large Data Bases}, ser. VLDB '05.\hskip 1em plus 0.5em minus 0.4em\relax
  VLDB Endowment, 2005, p. 946–957.

\bibitem{10.1016/S0020-0190(99)00031-9}
\BIBentryALTinterwordspacing
S.~Khuller, A.~Moss, and J.~S. Naor, ``The budgeted maximum coverage problem,''
  \emph{Inf. Process. Lett.}, vol.~70, no.~1, p. 39–45, apr 1999. [Online].
  Available: \url{https://doi.org/10.1016/S0020-0190(99)00031-9}
\BIBentrySTDinterwordspacing

\bibitem{7929916}
L.~Guo, D.~Zhang, G.~Cong, W.~Wu, and K.-L. Tan, ``Influence maximization in
  trajectory databases,'' in \emph{2017 IEEE 33rd International Conference on
  Data Engineering (ICDE)}, 2017, pp. 27--28.

\bibitem{8118111}
S.~Wang, Z.~Bao, J.~S. Culpepper, T.~Sellis, and G.~Cong, ``Reverse $k$ nearest
  neighbor search over trajectories,'' \emph{IEEE Transactions on Knowledge and
  Data Engineering}, vol.~30, no.~4, pp. 757--771, 2018.

\bibitem{8604082}
L.~Wang, Z.~Yu, D.~Yang, H.~Ma, and H.~Sheng, ``Efficiently targeted billboard
  advertising using crowdsensing vehicle trajectory data,'' \emph{IEEE
  Transactions on Industrial Informatics}, vol.~16, no.~2, pp. 1058--1066,
  2020.

\bibitem{zahradka2021price}
J.~Zahr{\'a}dka, V.~Machov{\'a}, and J.~Ku{\v{c}}era, ``What is the price of
  outdoor advertising: A case study of the czech republic?'' \emph{Ad Alta:
  Journal of Interdisciplinary Research}, 2021.

\bibitem{10.14778/2752939.2752950}
\BIBentryALTinterwordspacing
C.~Aslay, W.~Lu, F.~Bonchi, A.~Goyal, and L.~V.~S. Lakshmanan, ``Viral
  marketing meets social advertising: Ad allocation with minimum regret,''
  \emph{Proc. VLDB Endow.}, vol.~8, no.~7, p. 814–825, feb 2015. [Online].
  Available: \url{https://doi.org/10.14778/2752939.2752950}
\BIBentrySTDinterwordspacing

\bibitem{zhang2021minimizing}
Y.~Zhang, Y.~Li, Z.~Bao, B.~Zheng, and H.~Jagadish, ``Minimizing the regret of
  an influence provider,'' in \emph{Proceedings of the 2021 International
  Conference on Management of Data}, 2021, pp. 2115--2127.

\bibitem{nemhauser1978analysis}
G.~L. Nemhauser, L.~A. Wolsey, and M.~L. Fisher, ``An analysis of
  approximations for maximizing submodular set functions—i,''
  \emph{Mathematical programming}, vol.~14, no.~1, pp. 265--294, 1978.

\bibitem{fisher1978analysis}
M.~L. Fisher, G.~L. Nemhauser, and L.~A. Wolsey, ``An analysis of
  approximations for maximizing submodular set functions—ii,'' in
  \emph{Polyhedral combinatorics}.\hskip 1em plus 0.5em minus 0.4em\relax
  Springer, 1978, pp. 73--87.

\bibitem{zhou2017scaling}
T.~Zhou, H.~Ouyang, J.~Bilmes, Y.~Chang, and C.~Guestrin, ``Scaling submodular
  maximization via pruned submodularity graphs,'' in \emph{Artificial
  Intelligence and Statistics}.\hskip 1em plus 0.5em minus 0.4em\relax PMLR,
  2017, pp. 316--324.

\bibitem{10.1007/978-3-031-22064-7_17}
D.~Ali, S.~Banerjee, and Y.~Prasad, ``Influential billboard slot selection
  using pruned submodularity graph,'' in \emph{Advanced Data Mining and
  Applications}, W.~Chen, L.~Yao, T.~Cai, S.~Pan, T.~Shen, and X.~Li,
  Eds.\hskip 1em plus 0.5em minus 0.4em\relax Cham: Springer Nature
  Switzerland, 2022, pp. 216--230.

\bibitem{10.1007/3-540-57182-5_65}
D.~Wagner and F.~Wagner, ``Between min cut and graph bisection,'' in
  \emph{Mathematical Foundations of Computer Science 1993}, A.~M. Borzyszkowski
  and S.~Soko{\l}owski, Eds.\hskip 1em plus 0.5em minus 0.4em\relax Berlin,
  Heidelberg: Springer Berlin Heidelberg, 1993, pp. 744--750.

\end{thebibliography}

\begin{IEEEbiography}[{\includegraphics[width=1in,height=1.25in,clip,keepaspectratio]{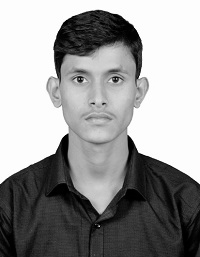}}]{Dildar Ali} is currently pursuing his Ph.D. Degree in the Department of Computer Science and Engineering at Indian Institute of Technology Jammu. Prior to this, he worked as a Project Associate at Indian Institute of Technology Bhilai, India. He obtained his Master of Technology in Computer Science and Engineering from Aliah University, Kolkata, India. His broad research interest includes Graph Theory, Submodular Function Optimization and Applications.
\end{IEEEbiography}

\begin{IEEEbiography}[{\includegraphics[width=1in,height=1.25in,clip,keepaspectratio]{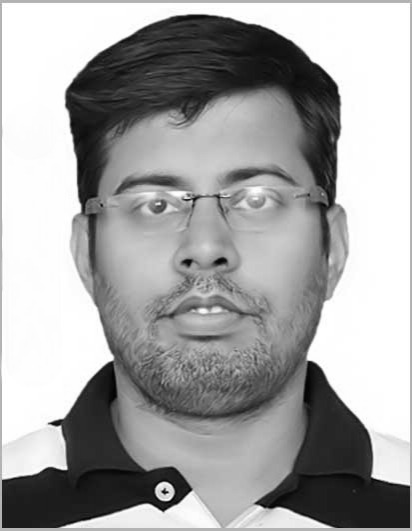}}]{Suman Banerjee} is currently working as an Assistant Professor in the Department of Computer Science and Engineering at Indian Institute of Technology Jammu. Prior to this he was a post doctoral fellow in the Department of Computer Science and Engineering of Indian Institute of Technology Gandhinagar. His broad research interest includes Algorithmic Data Management, Social Network Analysis, Graph Theory and Graph Algorithms, and Parameterized Complexity. He has published more than 30 papers in journals and conferences of international repute. 
\end{IEEEbiography}

\begin{IEEEbiography}[{\includegraphics[width=1in,height=1.25in,clip,keepaspectratio]{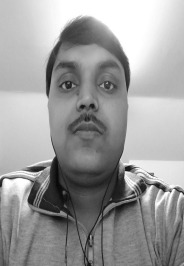}}]{Yamuna Prasad} is currently working as Assistant Professor in the department of computer science and Engineering, Indian Institute of Technology Jammu, India. He received his Ph.D. in computer science and Engineering from Indian Institute of Technology Delhi, India. He was a Postdoctoral Fellow of Thompson Rivers University, BC, Canada from 2017 to 2018, and a visiting scholar with University of Cincinnati, OH, USA, in 2018. He has authored over dozens of articles in peer-reviewed journals and conferences in the area AI, ML and Bioinformatics. His research interests include intersection of artificial intelligence, optimization, soft computing, machine learning, deep learning and security.
\end{IEEEbiography}

\end{document}